\documentclass[11pt]{article}

\usepackage{algorithm}
\makeatletter
\renewcommand{\fnum@algorithm}{\fname@algorithm}
\makeatother

\usepackage{mathrsfs}
\usepackage{enumitem}
\usepackage{bbm}
\usepackage{tikz}
\usepackage{float}
\usetikzlibrary{backgrounds}
\usepackage{color}
\usepackage{graphicx}
\usepackage{latexsym}
\usepackage{amsfonts}
\usepackage{pifont,xspace,fullpage,epsfig, wrapfig}
\usepackage{amsmath, amssymb, amsthm} %
\usepackage{multirow}
\usepackage{dsfont}
\usepackage{array}
\usepackage[pageanchor=false,pdfstartview=FitH,colorlinks,linkcolor=blue,filecolor=blue,citecolor=blue,urlcolor=blue]{hyperref}
\usepackage{adjustbox}
\usepackage{caption}
\usepackage{aliascnt}
\usepackage[numbers]{natbib} %
\usepackage{cleveref}

\usepackage{tabularx}
\newcounter{protocol}
\newenvironment{protocol}[1]
{
  	\refstepcounter{protocol}
	\par\addvspace{\topsep}
   \noindent
   \tabularx{\linewidth}{@{} X @{}}
    \hline
    \textbf{Protocol \theprotocol:} #1 \\
    \hline
}
{
	 \\
    \hline
   \endtabularx
   \par\addvspace{\topsep}
}

\DeclareSymbolFont{AMSb}{U}{msb}{m}{n}
\DeclareMathSymbol{\N}{\mathbin}{AMSb}{"4E}
\DeclareMathSymbol{\Z}{\mathbin}{AMSb}{"5A}
\DeclareMathSymbol{\R}{\mathbin}{AMSb}{"52}
\DeclareMathSymbol{\Q}{\mathbin}{AMSb}{"51}
\DeclareMathSymbol{\erert}{\mathbin}{AMSb}{"50}
\DeclareMathSymbol{\I}{\mathbin}{AMSb}{"49}

\ifdefined\IsDraft
\newcommand{\mynote}[2]{{\textcolor{#1}{ #2}}}
\newcommand{\mnote}[1]{\mynote{green}{Moni: {#1}}}
\newcommand{\enote}[1]{\mynote{red}{Eran: {#1}}}
\newcommand{\unote}[1]{\mynote{blue}{Uri: {#1}}}
\newcommand{\bnote}[1]{\mynote{purple}{Bar: {#1}}}
\newcommand{\changed}[1]{{\color{cyan} #1}}
\newcommand{\deleted}[1]{{\color{cyan}~Deleted:~{\color{red} #1}}}
\newcommand{\TODO}[1]{{\color{red} TODO:} {\color{blue} #1}}
\else
\newcommand{\mynote}[2]{}
\newcommand{\mnote}[1]{}
\newcommand{\enote}[1]{}
\newcommand{\unote}[1]{}
\newcommand{\bnote}[1]{}
\newcommand{\changed}[1]{#1}
\newcommand{\deleted}[1]{}
\newcommand{\TODO}[1]{}
\fi

\definecolor{gray}{gray}{0.4}

\newcommand{\remove}[1]{}

\renewcommand{\cref}{\Cref}

\newtheorem{theorem}{Theorem}[section]

\newaliascnt{lemma}{theorem}
\newtheorem{lemma}[lemma]{Lemma}
\aliascntresetthe{lemma}

\newaliascnt{claim}{theorem}
\newtheorem{claim}[claim]{Claim}
\aliascntresetthe{claim}

\newaliascnt{corollary}{theorem}
\newtheorem{corollary}[corollary]{Corollary}
\aliascntresetthe{corollary}

\newaliascnt{proposition}{theorem}

\aliascntresetthe{proposition}

\newaliascnt{conjecture}{theorem}

\aliascntresetthe{conjecture}

\newaliascnt{definition}{theorem}
\newtheorem{definition}[definition]{Definition}
\aliascntresetthe{definition}

\newaliascnt{remark}{theorem}
\newtheorem{remark}[remark]{Remark}
\aliascntresetthe{remark}

\newaliascnt{example}{theorem}

\aliascntresetthe{example}

\newaliascnt{fact}{theorem}
\newtheorem{fact}[fact]{Fact}
\aliascntresetthe{fact}

\crefname{lemma}{Lemma}{Lemmas}
\crefname{figure}{Figure}{Figures}
\crefname{claim}{Claim}{Claims}
\crefname{corollary}{Corollary}{Corollaries}
\crefname{proposition}{Proposition}{Propositions}
\crefname{conjecture}{Conjecture}{Conjectures}
\crefname{definition}{Definition}{Definitions}
\crefname{remark}{Remark}{Remarks}
\crefname{example}{Example}{Examples}
\crefname{protocol}{Protocol}{Protocols}
\crefname{fact}{Fact}{Facts}

\newcommand{\ie}{{\it i.e.,\ }}
\newcommand{\eg}{{\it e.g.,\ }}

\newcommand{\AAA}{\mathcal A}

\newcommand{\cM}{\mathcal M}

\newcommand{\eps}{\varepsilon}

\newcommand{\zo}{\{0,1\}}
\newcommand{\zos}{\{0,1\}^*}

\newcommand{\NN}{\mathbb{N}}

\newcommand{\diam}{{\rm diam}}

\newcommand{\dist}{{\rm dist}}

\newcommand{\polylog}{\mathop{\rm polylog}}

\newcommand{\poly}{\mathop{\rm poly}}

\newcommand{\negl}{\mathop{\rm negl}}

\newcommand{\set}[1]{\left\{ #1 \right\}}
\newcommand{\floor}[1]{\left\lfloor #1 \right\rfloor}

\def\Q{\operatorname*{\mathbb{Q}}}
\def\poly{\mathop{\rm{poly}}\nolimits}

\mathchardef\mhyphen="2D

\newcommand{\from}{\leftarrow}
\newcommand{\la}{\gets}

\newcommand{\Pc}{\MathAlgX{P}}
\newcommand{\hatPc}{\hat{\Pc}}

\newcommand{\Vc}{\MathAlgX{V}}
\newcommand{\Hc}{\MathAlgX{H}}
\renewcommand{\P}{{\cal P}}

\newcommand{\sk}{\mathsf{sk}}
\newcommand{\pk}{\mathsf{pk}}
\newcommand{\evk}{\mathsf{evk}}
\newcommand{\pp}{\mathsf{pp}}

\newcommand{\vs}{\vect{s}}
\newcommand{\vv}{\vect{v}}
\newcommand{\vu}{\vect{u}}
\newcommand{\vx}{\vect{x}}

\newcommand{\vS}{\vect{S}}

\newcommand{\Adv}{\adv}

\newcommand{\secParam}{\kappa}

\newcommand{\Sim}{\MathAlgX{Sim}}
\newcommand{\aux}{\mathsf{aux}}

\newcommand{\wrt} {with respect to\xspace}
\newcommand{\wlg} {without loss of generality\xspace}

\newcommand{\abs}[1]{\left\lvert #1 \right\rvert}

\newcommand{\abort}{\MathAlgX{abort}}

\newcommand{\cond}{\;|\;}
\newcommand{\bigcond}{\;\Big|\;}

\newcommand{\SMbox}[1]{\mbox{\scriptsize {\sc #1}}}
\newcommand{\REAL}{\SMbox{REAL}}
\newcommand{\IDEAL}{\SMbox{IDEAL}}

\newcommand{\Ensuremath}[1]{\ensuremath{#1}\xspace}
\newcommand{\MathAlg}[1]{\mathsf{#1}}
\newcommand{\MathAlgX}[1]{\Ensuremath{\MathAlg{#1}}}
\newcommand{\sff}[1]{\Ensuremath{\MathAlg{#1}}}
\newcommand{\Recon}{\MathAlgX{Recon}}
\newcommand{\Share}{\MathAlgX{Share}}
\newcommand{\VC}{\MathAlgX{VC}}
\newcommand{\VCSetup}{\MathAlgX{VC.Setup}}
\newcommand{\VCCom}{\MathAlgX{VC.Com}}
\newcommand{\VCOpen}{\MathAlgX{VC.Open}}
\newcommand{\VCVerify}{\MathAlgX{VC.Verify}}
\newcommand{\phC}{\MathAlgX{C}}

\newcommand{\phCCom}{\MathAlgX{C.Com}}

\newcommand{\Sig}{\MathAlgX{Sig}}
\newcommand{\SigGen}{\MathAlgX{Sig.Gen}}
\newcommand{\SigSign}{\MathAlgX{Sig.Sign}}
\newcommand{\SigVerify}{\MathAlgX{Sig.Verify}}

\newcommand{\PKE}{\MathAlgX{PKE}}
\newcommand{\PKEGen}{\MathAlgX{PKE.Gen}}
\newcommand{\PKEEnc}{\MathAlgX{PKE.Enc}}
\newcommand{\PKEDec}{\MathAlgX{PKE.Dec}}

\newcommand{\FHE}{\MathAlgX{FHE}}
\newcommand{\FHEGen}{\MathAlgX{FHE.Gen}}
\newcommand{\FHEEnc}{\MathAlgX{FHE.Enc}}
\newcommand{\FHEEval}{\MathAlgX{FHE.Eval}}
\newcommand{\FHEDec}{\MathAlgX{FHE.Dec}}
\newcommand{\FHECPA}{\MathAlgX{FHE\mbox{-}CPA}}

\newcommand{\LECC}{\MathAlgX{LECC}}
\newcommand{\LECCEnc}{\MathAlgX{LECC.Enc}}

\newcommand{\PRG}{\MathAlgX{G}}

\newcommand{\IS}{{\mathcal{I}}}
\newcommand{\ssize}[1]{|#1|}
\newcommand{\sset}[1]{\{#1\}}
\newcommand{\vect}[1]{{ \bf{#1}}}
\newcommand{\ppr}[2]{\Pr_{#1}\left[#2\right]}
\newcommand{\ceil}[1]{\left\lceil #1 \right\rceil}
\newcommand{\err}{\textup{\textsf{err}}}

\newcommand{\pr}[1]{\Pr\left[#1\right]}
\newcommand{\felect}{f_{\MathAlgX{elect}}}
\newcommand{\felectstar}{\ensuremath{\felect}\xspace}
\newcommand{\fmany}{\ensuremath{f_\MathAlgX{many}}\xspace}
\newcommand{\fcircuit}{\ensuremath{f_{\MathAlgX{many},C}}\xspace}

\newcommand{\comp}{\stackrel{\mbox{\tiny C}}{{\equiv}}}

\newcommand{\su}{\subseteq}

\newcommand{\xor}{\oplus}

\newcommand{\Real}{\REAL}

\newcommand{\Ideal}{\IDEAL}
\newcommand{\outValue}{{y}}

\def\ceil#1{{\lceil {#1} \rceil}}
\def\eps{{\varepsilon}}
\def\partNum{{n}}

\def\secParam{{\kappa}}

\def\3PartyBias{{8/r}}

\newcommand{\comment}[1]{}

\def\C{{\cal C}}

\def\I{{\cal I}}

\def\N{{\mathbb N}}

\def\P{{\cal P}}
\def\Q{{\cal Q}}
\def\R{{\mathbb R}}

\def\Z{{\cal Z}}

\newcommand{\Dist}{\mathsf{D}}

\ifdefined\IsDraft
\newcommand{\authornote}[2]{{\bf [{\color{red} #1's Note:} {\color{blue} #2}]}}
\else
  \newcommand{\authornote}[2]{}
\fi

\newcommand{\partyInputTag}[1]{{x'_{#1}}}

\newcommand{\vecInput}{\vx}

\makeatletter
\newcommand{\thickhline}{%
    \noalign {\ifnum 0=`}\fi \hrule height 1pt
    \futurelet \reserved@a \@xhline
}
\newcolumntype{"}{@{\hskip\tabcolsep\vrule width 1pt\hskip\tabcolsep}}
\makeatother

\makeatletter
\newlength{\fboxhsep}
\newlength{\fboxvsep}
\newlength{\fboxtoprule}
\newlength{\fboxbottomrule}
\newlength{\fboxleftrule}
\newlength{\fboxrightrule}
\def\@frameb@xother#1{%
  \@tempdima\fboxtoprule
  \advance\@tempdima\fboxvsep
  \advance\@tempdima\dp\@tempboxa
  \hbox{%
    \lower\@tempdima\hbox{%
      \vbox{%
        \hrule\@height\fboxtoprule
        \hbox{%
          \vrule\@width\fboxleftrule
          #1%
          \vbox{%
            \vskip\fboxvsep
            \box\@tempboxa
            \vskip\fboxvsep}%
          #1%
          \vrule\@width\fboxrightrule}%
        \hrule\@height\fboxbottomrule}%
    }%
  }%
}
\long\def\fboxother#1{%
  \leavevmode
  \setbox\@tempboxa\hbox{%
    \color@begingroup
    \kern\fboxhsep{#1}\kern\fboxhsep
    \color@endgroup}%
  \@frameb@xother\relax}

\makeatother

\newcommand{\paren}[1]{\left(#1\right)}
\newcommand{\sparen}[1]{(#1)}
\newcommand{\sqparen}[1]{\left[#1\right]}

\newcommand{\sabs}[1]{|#1|}
\newcommand{\LIVE}{\mathcal{L}}
\newcommand{\of}[1]{\paren{#1}}
\newcommand{\sof}[1]{\sparen{#1}}

\newcommand{\cI}{\mathcal{I}}
\newcommand{\cH}{\mathcal{H}}
\newcommand{\cA}{\mathcal{A}}
\newcommand{\corrset}{\cI}
\newcommand{\honset}{\cH}
\newcommand{\adv}{\cA}

\newcommand{\CommSetup}{{\bf \texttt{CommitteeSetup}}\xspace}
\newcommand{\Committee}{{\bf \texttt{Committee}}\xspace}

\newcommand{\ManyCommittees}{{\bf \texttt{ManyCommittees}}\xspace}

\newcommand{\PerComm}{{\bf \texttt{PersonalCommittee}}\xspace}

\newcommand{\ProofOfCorrectness}{{\bf \texttt{EfficientMutliverifierProof}}\xspace}

\newcommand{\GenProto}{{\bf \texttt{GeneralProtocol}}\xspace}
\newcommand{\MainProto}{{\bf \texttt{MainProtocol}}\xspace}

\newcommand{\comsize}{\secParam}
\newcommand{\pcsize}{\secParam}

\newcommand{\numbin}{b}

\newcommand{\cS}{\mathcal{S}}
\newcommand{\cT}{\mathcal{T}}
\newcommand{\cQ}{\mathcal{Q}}

\newcommand{\cD}{\mathcal{D}}

\newcommand{\cV}{\mathcal{V}}
\newcommand{\cX}{\mathcal{X}}

\newcommand{\EE}{\mathbb{E}}
\newcommand{\ex}[1]{\EE\sqparen{#1}}

\newcommand{\rndsu}{\cS}
\newcommand{\rndsuT}{\cT}
\newcommand{\queryset}{\cQ}
\newcommand{\block}{\mathcal{B}}

\newcommand{\minority}{\mathsf{minority}}
\newcommand{\bin}{\mathcal{B}}

\newcommand{\trustp}{\mathsf{T}}

\newcommand{\Gdist}[2]{\ensuremath{\mathcal{G}(#1,#2)}}
\newcommand{\diamLength}{\ell}
\newcommand{\out}{\normalfont{\texttt{out}}}
\newcommand{\intext}{\normalfont{\texttt{in}}}
\newcommand{\cN}{\mathcal{N}}
\newcommand{\neigh}{\cN}

\newcommand{\isampled}{\neigh^{\out}}
\newcommand{\sampledme}{\neigh^{\intext}}
\newcommand{\totalsampled}{\neigh}

\newcommand{\isample}{\isampled}

\newcommand{\totalsample}{\totalsampled}

\newcommand{\alive}{\ensuremath{\mathsf{alive}}\xspace}
\newcommand{\rej}{\ensuremath{\mathsf{reject}}\xspace}
\newcommand{\acc}{\ensuremath{\mathsf{accept}}\xspace}

\newcommand{\FILreps}{\ensuremath{\ceil{\frac{\log (n/4)}{\log(\secParam/4)}}+1}\xspace}

\ifdefined\IsLLNCS
\newenvironment{proofof}[1]{\begin{proof}[of~#1]}{\end{proof}}
\else
\newenvironment{proofof}[1]{\begin{proof}[Proof of~#1]}{\end{proof}}
\fi

\newcommand{\ith}{\ensuremath{i^{\text{th}}}\xspace}
\newcommand{\jth}{\ensuremath{j^{\text{th}}}\xspace}
\newcommand{\ellth}{\ensuremath{\ell^{\text{th}}}\xspace}

\newcommand{\encin}{\hat{x}}
\newcommand{\encout}{\hat{y}}

\newcommand{\codex}{\tilde{x}}
\newcommand{\codey}{\tilde{y}}
\newcommand{\codes}{\tilde{s}}

\newcommand{\ppt}{{\sc ppt}\xspace}
\newcommand{\pptm}{{\sc pptm}\xspace}

\newcommand{\Dom}{\cX}
\newcommand{\Rng}{\mathcal{Y}}

\newcommand{\nxtmsg}{\mathsf{NxtMsg}}
\newcommand{\board}{\ensuremath{\mathsf{Board}}\xspace}

\newcommand{\PCPP}{\operatorname{PCPP}}
\newcommand{\pub}{\mathsf{pub}}

\newcommand{\LOVE}{GMPC\xspace}
\newcommand{\register}{\MathAlgX{register}}
\newcommand{\get}{\MathAlgX{get}}
\newcommand{\fromuser}{\MathAlgX{from}}

\newcommand{\PCSG}{\ensuremath{\mathsf{PCSG}}\xspace}
\newcommand{\PCSGhon}{\ensuremath{\PCSG_{\mathsf{hon}}}\xspace}
\newcommand{\PCSGmal}{\ensuremath{\PCSG_{\mathsf{mal}}}\xspace}

\newcommand{\shuffle}{\mathsf{Shuffle}}

\newcommand{\hc}{\hat{c}}

\newcommand{\badset}{\mathsf{Bad}}

\newcommand{\AKS}{\mathsf{AKS}}
\newcommand{\fAKS}{f_\AKS}

\begin{document}

\title{MPC for Tech Giants (GMPC):\\ 
Enabling Gulliver and the Lilliputians to Cooperate Amicably
}
\ifdefined\IsAnon
\author{}
\else
\author{
Bar Alon\thanks{Department of Computer Science Ariel University, Ariel,
Israel.  Email: \texttt{alonbar08@gmail.com}. }
\and
Moni Naor\thanks{Department of Computer Science and Applied Mathematics, Weizmann Institute of Science, Rehovot,
Israel. Incumbent of the Judith Kleeman Professorial Chair. Email: \texttt{moni.naor@weizmann.ac.il}. }
\and
Eran Omri\thanks{Department of Computer Science Ariel University, Ariel,
Israel.  Email: \texttt{omrier@ariel.ac.il}. }
\and
Uri Stemmer\thanks{Tel Aviv University and Google Research. Email: \texttt{u@uri.co.il}. }
}
\fi

\date{April 6, 2023}
\maketitle

\begin{abstract}
In the current digital world, large organizations (sometimes referred to as tech giants) provide service to extremely large numbers of users. The service provider is often interested in computing various data analyses over the private data of its users, which in turn have their incentives to cooperate, but do not necessarily trust the service provider. %

In this work, we introduce the \emph{Gulliver multi-party computation model} (GMPC) to realistically capture the above scenario. The GMPC model considers a single highly powerful party, called the {\em server} or {\em Gulliver}, that is connected to $n$ users over a star topology network (alternatively formulated as a full network, where the server can block any message). The users are significantly less powerful than the server, and, in particular, should have both computation and communication complexities that are polylogarithmic in $n$.
Protocols in the GMPC model should be secure against malicious adversaries that may corrupt a subset of the users and/or the server. 

Designing protocols in the GMPC model is a delicate task, since users can only hold information about $\polylog(n)$ other users (and, in particular, can only communicate with $\polylog(n)$ other users). In addition, the server can block any message between any pair of honest parties. Thus, reaching an agreement becomes a challenging task. Nevertheless, we design generic protocols in the GMPC model, assuming that at most $\alpha<\changed{1/8}$ fraction of the users may be corrupted (in addition to the server). Our main contribution is a variant of Feige's committee election protocol [FOCS 1999] that is secure in the GMPC model. Given  this tool %
we show: 
\begin{enumerate}
    \item Assuming fully homomorphic encryption (FHE), any computationally efficient function with $O\left(n\cdot\polylog(n)\right)$-size output can be securely computed in the GMPC model.

    \item Any function that can be computed by a circuit of $O(\polylog(n))$ depth, $O\left(n\cdot\polylog(n)\right)$ size, and bounded fan-in and fan-out can be securely computed in the GMPC model {\em without assuming FHE}.

    \item In particular, {\em sorting} can be securely computed in the GMPC model without assuming FHE. This has important applications for the {\em shuffle model of differential privacy}, and resolves an open question of Bell et al.\ [CCS 2020].
\end{enumerate}

\end{abstract}

\tableofcontents

\section{Introduction}\label{sec:intro}

Consider a large organization (such as Apple or Google; henceforth the {\em server} or {\em Gulliver}) that is interested in executing a secure computation over its users' data (we call the users {\em Lilliputians}). 
As the server is way more powerful than the average user, we would like it to do most of the heavy-lifting in the computation, thereby allowing the users to do only lightweight computations. %
In addition, it is reasonable to assume that users can communicate only via the server (i.e., a {\em star} communication model). This is beneficial both for the server (who can better monitor and enforce the progress of the computation) and for the users (who only need to communicate with a supposedly reliable server). For example, in a protocol between Apple and its users, such a restricted communication model makes more sense than having 100M iPhone users communicate arbitrarily.

This suggests a computation model with a powerful server that on the one hand is supposed to do most of the computation work, and on the other hand, if corrupted, might tamper with the communication of the other users (because of the star topology). This combination poses several challenges (to be surveyed next), and indeed all prior works in this vein only presented constructions for {\em specific functionalities} such as addition; see the related works section. In this work, we set out to explore secure multiparty computation under these (and similar) conditions and present the first {\em general MPC construction} in this setting.

\subsection{The Gulliver MPC (GMPC) Model}\label{sec:GMPCintro} 
We consider a protocol involving a single {\em server} and $n$ {\em users}. We think of $n$ as being too big to allow users more than $\polylog(n)$ computational time and communication complexities. On the other hand, we assume that $\poly(n)$ time and communication bandwidth are within the server's reach.\footnote{For example, the server might be running on a computer cluster, and the users might be running on iPhones.} 
We assume that a subset of at most $\alpha n<\changed{n/8}$ of the users (fixed in advance), as well as the server, might be controlled by a malicious adversary.  
In particular, controlled users might halt prematurely and become unresponsive throughout the execution (e.g., iPhone users going off-line). 

Utility wise, we want to ensure that if the server is honest, then it gets a {\em correct} outcome for the computation, even if $\alpha n$ users are corrupted.
In other words, if the server is honest (and if at most $\alpha n$ users are corrupted), then output delivery is guaranteed. 
The privacy requirement is that an adversary controlling the server, as well as $\alpha n$ users, cannot learn anything more than the prescribed outcome. That is, even if the server is corrupted (as well as at most $\alpha n$ users), the only information that should be learned about users' inputs is what can be derived from the outcome itself.\footnote{Technically, as we consider a model that allows a malicious server to block \changed{a small fraction of honest users}, the privacy requirement states that it does not obtain more information than what can be derived from the output over the remaining users.}

The mindset is that all of the communication goes through the server. However, one issue we must address is that the server {\em should not have the ability to ``invent'' new users, or to replace too many honest users with malicious ones}. This restriction is challenging to enforce in a star topology, and hence, as in all previous works in this vein, we assume a trusted PKI setup for this model \cite{Bonawitz2017,RSY21,BellBGL020}. As we mentioned,  these previous works only presented constructions for specific functionalities such as addition. Our main result is to present the first general MPC construction under these conditions.

Towards achieving this, we introduce a slightly different computation model, which we call the GMPC model. Instead of assuming a star topology with a trusted PKI setup, in the GMPC model we assume a complete communication graph (secure point-to-point channels) where the server (or an adversary controlling it) sees all of the communication patterns and can block messages at will. That is, the server knows who is sending a message to whom, and can block any message before it reaches its destination. The server does not get to see the content of the messages, and cannot modify messages. This still allows the server to tamper with the communication. In particular, if an honest user $i$ wants to broadcast (via the server) a message to all of the users, then the server might block this message, or it can forward the message only to a {\em subset} of the users.

Throughout most of the paper, we focus on the above-mentioned communication model (point-to-point channels with server blocking, without assuming PKI). In \cref{sec:star} we show that protocols in the GMPC model can be compiled to operate in a star network with a PKI setup, and hence, working in the GMPC is without loss of generality.\footnote{We use a specific form of PKI suitable for the GMPC model, which we call GPKI. %
} 

\paragraph{The GMPC model allows for smoother use of composition theorems.}
We believe that operating in the GMPC model is a better design choice, as it is both a clean abstraction and can help in simplifying the analysis. That is, we believe that the GMPC model is interesting both on its own and as a stepping stone toward constructing protocols in a star network with a PKI setup. In more detail, in our protocol we need to execute many ($\poly(n)$) computations in parallel, and to argue about the security of the resulting overall protocol using composition. Our GMPC setting allows us to isolate key components of our protocol which are information-theoretic, and as a result, argue about their concurrent composition in a smooth way. In contrast, with a general PKI setup, this becomes more challenging.\footnote{Arguing about composability when there is a general PKI setup is subtle. Basically, the difficulty is that the distinguisher is non-uniform. Thus, if the PKI is common to all sessions, then the distinguisher may know the secret key (as its auxiliary input). Now it can distinguish a true encryption of an honest party's input from a simulated encryption (as the simulator does not know the honest party's input).} Once we know that our protocol is secure in the GMPC model (after composition has been done under the GMPC umbrella), then we can easily compile it {\em as a single protocol} to operate in a star network with a trusted PKI setup. See \cref{sec:star}.

\subsection{Our Results}
Our main contribution is showing a generic construction in the GMPC model for any functionality where the users' inputs are short, of size $\polylog(n)$, and where only the server gets an outcome, of size $\tilde{O}(n)$, assuming less than \changed{1/8} of the users are corrupted (and possibly also the server). 
\begin{theorem}[Informal version of \cref{thm:main}]\label{thm:intro_main}
	 Let $f:(\zo^{\polylog(n)})^{n+1}\mapsto\zo^{n\cdot\polylog(n)}$ be a computationally efficient functionality, and let $\alpha<\changed{1/8}$ be a constant. Then, $f$ can be securely computed in the \LOVE model, assuming at most $\alpha n$ users are corrupted, and where the server is possibly corrupted as well.
	 Furthermore, the number of rounds in the protocol is at most $\polylog(n)$, the communication and computational complexity of each user is $\polylog(n)$, and the communication and computational complexity of the server is $\poly(n)$. Our construction assumes the existence of signature schemes, vector commitment schemes, pseudorandom generators, and fully homomorphic encryption schemes. 
\end{theorem}

\begin{remark}
Two remarks are in order:
\begin{enumerate}
    \item Our construction can also be used to compute functionalities with a long (polynomial) output, say of length $T$, at the cost of having users' complexities grow with $\tilde{O}(T/n)$. We chose to focus on functionalities with output length $\tilde{O}(n)$ in order to keep users' complexities polylogarithmic in $n$.
    \item 
    We are not restricted to solitary-output functionalities. We can securely compute multi-output asymmetric functionalities (where every party gets an output), at the price of having security with abort: If the server is honest then output delivery is guaranteed, and otherwise some parties may fail to receive their output.
\end{enumerate}
\end{remark}

Following \cite{BCDH18,BCG21,BGT13}, the main tool we use towards proving \cref{thm:intro_main} is a \emph{committee election protocol}. We construct a protocol in the \LOVE model that allows all parties to agree on a $\polylog(n)$-sized committee, which will contain an honest majority except with negligible probability. 
\begin{theorem}[Committee election, informal version of \cref{thm:committee}]\label{thm:intro_committee}
	Let $\alpha<\changed{1/8}$ be a constant. There exists a protocol in the \LOVE model such that the following holds, even if $\alpha n$ of the users and possibly the server are corrupted. Either all honest users abort, or at least $(1-2\alpha)n$ honest users output the same committee $\C$ satisfying
	\begin{enumerate}
		\item $|\C|=\polylog(n)$.
		\item $\C$ contains a vast honest majority, except with negligible probability.
	\end{enumerate}
	Furthermore, the communication and computational complexity of each user is $\polylog(n)$, while the communication and computational complexity of the server is $\poly(n)$. Our construction assumes the existence of vector commitment schemes.
\end{theorem}

\subsection{Challenges of the Gulliver MPC Model}\label{sec:IntroChallenges}

Before presenting our constructions, 
we highlight a few of the challenges that arise when constructing protocols in the GMPC model. We explain how we address some of these challenges in Section~\ref{sec:techniques}.

\paragraph{Challenge 1: Blocks vs.\ aborts.}
The first challenge we need to address is that a corrupted server can block user messages at will. In particular, if an honest user catches the server cheating, and tries to advertise this fact to the other users, then the server can simply block this user from the rest of the computation. Furthermore, the other users would not be able to distinguish between whether this user was blocked (meaning that the server is corrupted) and whether this user simply aborted the computation. 

To better illustrate this issue, let us first explain why a natural variant of Feige's committee election protocol~\cite{Feige99} does not work. Recall that in Feige's protocol, in order to agree on a committee of size $\approx k$, every party $i\in[n]$ samples a bin $b_i\in[n/k]$ (and broadcasts its choice), and then all parties agree on the lightest bin as the committee. We cannot implement this protocol directly in the GMPC model since (i) there is no broadcast channel, and (ii) users' complexities are bounded by $\polylog(n)$ so a user cannot send/receive messages to/from all other users directly. As an attempt to overcome these issues, consider implementing Feige's protocol in the GMPC model via the server. That is, every user sends to the server its choice of the bin, and the server then (supposedly) sends to all users the committee, defined as the users in the lightest bin, alongside the index of the bin. Each user then verifies that it is in the committee if it chose this bin, and not in the committee otherwise. The problem here is that a corrupted server can simply block all of the users who chose the lightest bin, and place corrupted users at that bin instead of the blocked users. The blocked users will have no way of ``screaming for help'', and the remaining users will not know that the server is cheating. We will overcome this challenge by designing a more complex committee election protocol, based on Feige's protocol, tailored to the GMPC model. 

\paragraph{Challenge 2: DDoS attacks.}
Although all parties are connected to each other, the protocols we present have the property that the communication and computational complexity of each user is polylogarithmic in $\partNum$. In particular, the number of messages each user sends and receives is polylogarithmic. However, note that an adversary can always blow up the communication complexity of a protocol by
flooding honest parties with many garbage messages (even if the server is honest). Boyle et al.~\cite{BCDH18,BCG21} handled this by counting only messages that are actually processed by honest parties.\footnote{This is formalized by splitting the ``receive phase'' into two phases. The first being the \emph{filtering phase}, where each party inspects its incoming messages according to some filtering rules specified by the protocol, and discards some of the messages. This phase is followed by a \emph{processing phase}, where the parties compute their next-message functions based on the non-filtered messages.} 

Since we also require the computational complexity of the users to be polylogarithmic in $\partNum$, the solution of \cite{BCDH18,BCG21} is insufficient for our case. We propose two solutions for this issue. The first solution is purely combinatorial, however, it will allow the adversary to cause a few honest users to abort even if the server is honest. The second solution will prevent that, however, it is somewhat more involved. We refer the reader to \cref{sec:Background} for more details.

\paragraph{Challenge 3: Concurrent composition.} In our protocol, we need each user to run a secure two-party computation with the server for coin tossing. These protocols must be executed in parallel, as otherwise, the number of total rounds in our construction would be too large. There are two standard approaches for arguing about the composition of these parallel executions: Either using the {\em universal composition (UC)} paradigm of \citet{Canetti01} or using {\em bounded concurrent composition} \cite{Lindell08}. However, these techniques are not directly applicable in our setting: 
\begin{enumerate}
    \item Existing UC-secure protocols for coin-tossing require additional assumptions which we do not have \cite{CanettiF01,PassLV12}. Importantly, we do {\em not} assume a common random string. Indeed, agreeing on a random string in the GMPC model is one of our main technical contributions.
    \item With bounded concurrent composition, the complexity of each user depends polynomially on the number of concurrent executions \cite{Lindell04}. In our case, the number of concurrent executions is $n$, i.e., the number of users, and hence we cannot use bounded concurrent composition (as we aim for $\polylog(n)$ user complexities).
\end{enumerate} 
To overcome this challenge, we deviate from the standard real vs.\ ideal paradigm, and provide an alternative security proof for our coin-tossing protocols (we do not simulate them). Instead, we identify ``good events'' and claim that they hold with overwhelming probability. (See, for example, Theorem~\ref{thm:intro_committee} above, where we claim that a ``good'' committee is elected with overwhelming probability.) We stress that the security of our full construction is proven using the real vs.\ ideal paradigm (assuming that our predefined ``good event'' has occurred).

\paragraph{Challenge 4: Verifying long statements.} We do not assume any specific bound on the runtime needed to evaluate the functionality we compute, only that it can be evaluated in polynomial time. For example, it can require time $n^5$. As users' runtime is restricted to $\polylog(n)$ this means that we must delegate essentially all of the work to the server. This must be done carefully, as the server might be corrupted. In particular, we will need the server to prove to the users that it did the computation correctly. However, as each user is only polylogarithmic in $n$, any one of them cannot even read the
entire statement (this is going to be a statement about a computation with $n$ inputs). To tackle this, we let the users work together as one verifying unit when verifying long statements.

\begin{remark}
Except for Challenge 2 (DDoS), these challenges are also present when working directly in a star network with PKI. That is, they are not a byproduct of the GMPC model. Moreover, Challenge 3 (composition) becomes even more severe in a star network with PKI, as we explained in \cref{sec:GMPCintro}. 
\end{remark}

\subsection{Our Techniques}\label{sec:techniques}
We now provide an overview of our techniques. We start by addressing the following two challenges that arise from the server's ability to block users:

\begin{enumerate}
	\item A malicious server can fail-stop honest users and lie about it. That is, it can block an honest user $i$ and tell another user $j$ that $i$ aborted, even though it did not.
	
	\item If an honest user catches the server cheating, then the server can block this user and prevent it from ``screaming for help''.  
\end{enumerate}

In order to deal with these two issues, we construct two tools.

	\paragraph{Tool 1: Personal committees.} To solve the first issue, each user and the server will replace the user with a \emph{personal committee} (PC) of size $\polylog(n)$ chosen uniformly at random. Now, instead of interacting directly with the users, the server will henceforth interact with each PC (as if it is the single party it represents). If the server is honest, then this ensures that {\em all} personal committees, including those of malicious users, will have an honest majority, and thus they effectively become honest parties and will not abort. Therefore, if the server is malicious and is lying to some PC $\P$ about another PC aborting, then $\P$ knows that the server is cheating and therefore it will abort.
	
	That is, once we know that a user $i$ is represented by a PC $\P$, which was agreed upon by both the server and the user, then this PC should never abort, and the server cannot claim that it aborted. The question is how do we advertise the set of PCs (one for each user), in a way that remains secure even if the server is malicious. In other words, we need to figure out how can the users {\em agree} on the set of PCs.

    \paragraph{Agreeing on the set of PCs.} The idea is as follows. The server will commit (for every user separately) on an array of length $n$ whose \ith entry is the \ith PC (set to $\bot$ if user $i$ was caught cheating during the sampling of the PC). Each user $i$ then verifies that its PC appears in this array. Then, the user chooses a random set of $\polylog(n)$ other users to compare the information with them. Specifically, it sends them the commitment it got from the server for the array (as well as the needed public parameters) and its PC. These users respond with the commitments that they got from the server and their PCs. User $i$ then verifies with the server that the information it got from these other users is valid.
    Note that some users might not respond, or provide information on which the server will not agree (either malicious or blocked), but not too many of them, as otherwise user $i$ learns that the server is cheating, and aborts.

    If user $i$ did not abort until now, then in particular, its PC appears correctly in all of the (validated) commitments it collected from the other users. Since these users were chosen at random, by the Chernoff bound (Fact~\ref{fact:chernoff}), this means that user $i$'s PC is stated correctly in the {\em vast majority} of the commitments given to the users in the protocol (even ones which user $i$ did not sample). Also, by this point, every user $j$ possesses $\polylog(n)$ randomly sampled commitments. Hence, again by the Chernoff bound, with overwhelming probability, for every user $j$ (even one who did not communicate with user $i$), it must hold that the PC of user $i$ is stated correctly in the vast majority of the commitments possessed by user $j$. That is, by this point the (remaining) users have reached an agreement on the array of PCs in the following sense: Every user holds a collection of commitments (validated by the server) such that for every user who did not abort the computation till now, its PC is stated correctly in the majority of these commitments.

    Finally, to ensure that most honest users remain active, each of them samples $\polylog(n)$ random users, notifies them it is alive, and requests a feedback. If more than \changed{1/8} of them did not respond, then the user aborts. If a user $i$ did not abort till now, it can trust that at least \changed{7/8} of the users are active, and furthermore, the PC of every active (honest) user is correctly stated and is consistent in the commitments given by the server \changed{to  active honest users}. 
    From this point on, instead of interacting directly with the users, the server and the remaining users interact with the PCs (as if they are the single party they represent). By the above discussion, if the server is honest then {\em all} of the PCs are honest, and if the server is corrupted then at least \changed{3/4} fraction of the PCs are honest.\footnote{If the server is corrupted, then it can get away with blocking \changed{1/8} fraction of the users and replacing their PCs with malicious PCs. Overall, there could be at most \changed{2/8=1/4} malicious PCs.}
    As we mentioned, the benefit here is that now the server cannot lie about a PC aborting, because no PC should ever abort.

\paragraph{Tool 2: Sampling a communication graph.} For the second issue, we let the parties sample an undirected communication graph  where each vertex represents a PC and the diameter of the graph is at most $O(\log n)$.
Once this is done, after every message sent from the server, the PCs can simply check if their neighbors are still active for $\log n$ steps, and abort if they are not. Recall that (assuming the server is honest) each PC should contain an honest majority and should never abort the execution unless it catches the server cheating. Therefore, such a communication graph allows users (or PCs) to ``scream for help'' in the sense that if they abort the computation, then in the following $\log(n)$ rounds we will have that {\em all} of the users abort the computation.
	
We sample this graph by having every PC sample $\polylog(n)$ other PCs as their neighbors in the graph. %
\changed{To see why the graph has a small diameter, consider any two vertices $u$ and $v$, and  
suppose that we start from the vertex $u$ and sample $\polylog(n)$ neighbors. Then, each of its neighbors samples $\polylog(n)$ more vertices, and so on. Observe that as long as less than half of the vertices corresponding to honest PCs were sampled, on expectation, at least half of the vertices sampled at each iteration are new. Therefore, by Hoeffding's inequality it follows that with overwhelming probability at least a quarter of them are new. Since this corresponds to an exponential growth, with overwhelming probability the number of iterations before covering half of the honest vertices is logarithmic. Finally, once the process covers half of the honest vertices, then one of these honest vertices samples the vertex $v$ with overwhelming probability.}\footnote{This is similar to arguments on the distances in small world graphs~\cite{Kleinberg00,MankuNW04}.}

\subsubsection{Committee Election Protocol}
We now turn to describe the ideas behind the proof of \cref{thm:intro_committee}. 
Armed with the above two tools, sampling a committee of size $n'=\polylog(n)$ can be done using the following variant of Feige's committee election protocol \cite{Feige99}. % 
Let $k$ denote the number of PCs that remain after the sampling of the communication graph. Each PC $\P_i$ sends to the server a randomly chosen bin $b_i\in[k/n']$, and the server replies with the set of users $\C$ who chose the lightest bin alongside the index $i$ of this bin. Each PC then verifies the following:
\begin{enumerate}
	\item It is indeed in $\C$ if it chose the bin $i$, and not in $\C$ otherwise.
	
	\item $\C$ is not larger than the largest size the minimal bin can have, \ie the PC verifies that $|\C|\leq n'$.
	
	\item Its neighbors in the graph received the same messages from the server.
	
	\item Finally, each PC asks if its neighbors are still active for $\log n$ steps.
\end{enumerate}

We next claim that if no PC aborts, then they all agree on a small committee $\C$ that contains an honest majority, except with negligible probability. First, observe that if the server is honest then no PC aborts (even those that were sampled for malicious users). Next, consider the case where the server is malicious. By the properties of the graph, if any PC aborts then all PCs abort. Thus, we may assume that no PC aborted the execution. This implies that the server sent to all PC's of honest parties the same ``small'' committee $\C$ and the same bin-index $b\in[k/n']$. Moreover, it must be the case that for all PC's of honest users, they are in $\C$ if and only if they sampled $b$.

Finally, to show that $\C$ contains an honest majority with overwhelming probability, we use the same analysis as Feige's original protocol. 
First, as we explained above, the fraction of PC's of \emph{malicious} users is less than \changed{1/4}, and there is a \emph{strong majority} of PC's of honest users. Now, by Chernoff's bound, for any bin, the probability that the fraction of honest users is less than $\changed{3/4}$ is negligible. Then by the union bound, the probability that there exists such a bin is negligible. Thus, except with negligible probability, the fraction of honest users in the lightest bin (which is of size at most $k$) must be at least $\changed{3/4}$. Note that \changed{even if the server blocks another 1/8-fraction of the users in the elected committee, then the fraction of {\em honest active} users among the users of the elected committee is still more than $1/2$.} This allows the elected committee to securely compute any efficient functionality.\footnote{Specifically, if the server is honest then output delivery is guaranteed, and if the server is corrupted then we get security with abort.} We refer the reader to \cref{sec:comm_love} for further details. We stress that as the size of the committee is only polylogarithmic in $n$, then it can run computations of polynomial complexities in the number of participants.

\subsubsection{Secure Computation in the \LOVE Model}
We proceed with outlining the proof of \cref{thm:intro_main}.
Since we want  low communication and computation load on the users while keeping their inputs hidden, we employ a fully homomorphic encryption (FHE) scheme. The pattern we would like to follow is that the users choose {\em collectively} a key to such an FHE scheme, encrypt their private inputs with the key, and send the result to the server. The server in turn computes the desired function on the encrypted input. The users then collectively decrypt the result. 
With a trusted (or semi-honest) server, this would have been a reasonable protocol,  where the server does not learn about the users' inputs. But with a malicious server,  how can the users know that the value they are decrypting is indeed the correct one and not some other function of the inputs that the server computed? 

To resolve this, 
first, the users agree on a committee $\C$ using our committee election protocol. Once the users agreed on a committee $\C$, this committee samples keys for a fully-homomorphic encryption scheme (where the secret key is shared among the members of the committee). Then, similarly to \cite{BCDH18,BCG21,BGT13}, there is a ``tree process'' in which using $\log n$ rounds every committee spawns two committees (both with the state of holding the secret key). After this process, we have $n+1$ elected committees that everybody ``knows'' (actually, the users cannot really know all the $n+1$ committees since we aim for $\polylog(n)$ complexity for the users). With overwhelming probability, all committees will contain a vast honest majority. Thus we can effectively view them as honest parties. 

Next, we let the committees hold the users’ inputs. To do so efficiently, the \ith committee will request the \ith user for its input. To prove that the committee is honest and was sampled by the original elected committee (and not by a malicious server), we let the original committee generate signature keys and provide the committees it spawns with the secret key. Thus it can sign a message for the user to verify. (For this we need to advertise the public signature key that the original committee sampled. Advertising this public key can be done efficiently via our communication graph.) 
Next, every committee encrypts its input using a fully homomorphic encryption, and sends the encrypted input to the server. The server then homomorphically computes the function over the encrypted inputs.

Now we want each of the $n+1$ committees to decrypt for the server one chunk of the encrypted result. However, before they decrypt, they need to be convinced that the server did the computation correctly. Thus, the next step in our protocol is to let the server prove to the committees the following statement: ``Here is a list of all the encrypted inputs I got, and here is the encrypted result of applying the desired function to these encrypted inputs''.

The issue now is that the users (and each committee) are too weak to read even the statement itself (even without the proof). To overcome this challenge, we want the server to encode this statement using a (linear) error-correcting code such that afterwards the users would only need to verify that the encoding was done (roughly) correctly and they would not need to read all of the statement. So the server commits to the statement, its encoding, and the proof. Then, the $n+1$ committees together verify $\polylog(n)$ random positions of the encoding (they can do it efficiently since verifying a single position is like computing a summation, which the $n+1$ committees can do using a ``tree process''). After the proof verification is over, the committees decrypt the result for the server.

\subsection{Motivation from Differential Privacy: the Shuffle Model}

Differential privacy \cite{DMNS06} is a mathematical definition for privacy that aims to enable statistical analyses of datasets while ensuring that individual-level information does not leak. Traditionally, differentially private algorithms work in two main modalities: curator (centralized) and local. The curator model assumes a trusted centralized curator that collects all the personal information and then analyzes it and publishes the results. The local model, on the other hand, does not involve a central repository. Instead, each piece of personal information is randomized by its provider to protect privacy {\em before} submitting it to an (untrusted) server, who aggregates all the noisy messages it receives.

While the local model provides a preferable {\em trust model}, it involves a significant amount of noise (as noise is added to every piece of the data) which results in degraded performances compared to the curator model. In theory, existing MPC constructions can be used to implement any curator model computation without the need for a central repository (thereby essentially matching the trust assumptions of the local model). However, these constructions are still not practical enough to be implemented at the scale at which differentially private algorithms are executed, say when Apple is interacting with 100M iPhone users.

As a result, new models for differential privacy have emerged to alleviate the (inevitable) low accuracy of the local model, while hopefully maintaining its trust guarantees. This includes significant amount of work on the {\em shuffle model}~\cite{IKOS06,BittauEMMRLRKTS17,ErlingssonFMRTT19,CheuSUZZ19,BalleBGN19,BalleBGN20,BeimelHNS20,CheuU21,GhaziG0PV21,TenenbaumKMS21} as well as  the {\em hybrid model}~\cite{blender,AventDK20,BeimelKNSS20,KohenS22}. We view our work as providing another piece of the puzzle, taking an important step towards bridging between MPC and the conditions at which differential privacy is currently being deployed.

\paragraph{Application for the shuffle model.} The shuffle model augments the local model of differential privacy with a {\em shuffle functionality}, that takes inputs from all the users and delivers them to the server in a random order. The line of work on the shuffle model shows that this assumption allows for significant utility improvements over the (plain) local model. However, it is not clear how one should implement such a shuffle functionality securely and efficiently. We provide a possible answer to this question in the GMPC model. Although our main result (Theorem~\ref{thm:intro_main}) requires the use of heavy cryptographic machinery (e.g., fully homomorphic encryption), we stress that various simple yet important functionalities can be computed more efficiently and without relying on the more exotic types of cryptographic primitives once the parties elect a committee. In particular, aggregation and shuffling can be securely computed {\em without assuming FHE}. % 
We elaborate on this in \cref{sec:log_depth,sec:shuffle}.

\subsection{Related Work}

We are not the first to study a model where a strong server interacts with many ``weak'' users using a star topology, and similar models were considered by \cite{Bonawitz2017,RSY21,BellBGL020}. However, to the best of our knowledge, we are the first to present a general MPC construction in this setting. For example, \citet{BellBGL020} presented a constant-round secure-protocol for {\em summation in finite groups}, in which each user has only $\polylog(n)$ time and communication complexities, but 
extending their construction to other functionalities is not clear. In particular, it is not clear if their protocol can be used for bit summation.\footnote{For bit summation one needs to work in a finite group of size at least $n$, in order to avoid summation overflows, but then a single malicious user can modify the result almost arbitrarily by choosing its input to be a random number in this group. That is, a single malicious user can affect the sum by a factor of $n$.} \citet{RSY21} presented a protocol for a somewhat more general family of functionalities (functionalities requiring limited homomorphism), but their protocol only works in the semi-honest setting and it requires users' runtime to be polynomial in $n$.

Also very related to our work is the line of work on {\em communication locality}, initiated by \citet*{BGT13}. Similarly to our work, in this model, the parties are connected via a complete graph, but each party communicates only with $\polylog(n)$ other parties. Assuming PKI, \citet{BGT13} constructed a protocol computing any efficient functionality, tolerating at most $\alpha n$ corruptions, for some constant $\alpha<1/3$. \citet*{CCGGOZ15} later showed how to handle {\em adaptive adversaries} that can corrupt at most $\alpha n$ parties, for some constant $\alpha<1/2$, assuming the parties are given a \emph{symmetric-key infrastructure}. \citet*{BCDH18} considered the communication graph induced by the interaction in communication-locality protocols. Assuming the parties are given a PKI, they presented a protocol with low locality tolerating $\alpha n$ corruptions, where $\alpha<1/4$ is a constant, such that with overwhelming probability the resulting communication graph is \emph{not an expander}. \citet*{BCG21} showed communication-locality protocols with the additional property of low {\em average} communication complexity. That is, not only does each party communicate with only $\polylog(n)$ other parties, but also the total communication in the protocol is at most $n\cdot\polylog(n)$ bits.
We stress that these works on communication locality are fundamentally different than ours: These works do not involve a server that does all the heavy lifting and can block messages at will, which is the main challenge we address in our work. As these works on communication locality do not involve a server, the {\em computational complexity} of the parties in all these works is polynomial in $n$ (unlike polylogarithmic in $n$ in our work). %

The idea of using small committees to gain efficiency dates back  at least to Bracha~\cite{Bracha85} in 1985. Since then it has been employed by several papers in several settings. In particular, this was used by Boyle et al.~\cite{BGT13} in their communication-locality MPC protocol, and was used by Cohen et al.~\cite{CHOR18} to obtain efficient security uplifting reductions.

Lower bounds for the {\em bottleneck} communication complexity were investigated by Boyle et al.~\cite{BJPY18}. They showed that there exists a function such that for any protocol computing it correctly, \ie without requiring security, there exists a party whose communication complexity is $\Omega(n)$. % 
However, this does not say much about the GMPC model, as the communication complexity of the server is large ``by design''.

The issue of a server who talks with weaker processors who only communicate locally and only have local information has also been investigated in the context of interactive proofs. \citet{NaorPY20} showed how to adapt various techniques in proof complexity in order to achieve low communication protocols for proving the correctness of many functions (essentially all NC or even P if one is satisfied with arguments) where the data is distributed among many verifiers. 
\citet{BonehBCGI19} showed how to use linear PCPs is order to get zero-knowledge proofs when  the input statement is not fully available to any single verifier, but can still be efficiently accessed via linear queries.

Another direction to achieve security with many different users is to employ the paradigm of serverless computing where the participants are stateless. In particular, \citet{GentryHKMNRY21} have suggested the ``You Only Speak Once" secure MPC. One important  difference with this work is that there are no inputs to the parties that are being aggregated and computed.

\section{Preliminaries}
%%%%%%%%%%%%%%%%%%%%%%%%%%%%%%%%%%%%%%%%%%%%%%%%%%%%%%%
\subsection{Notation}
\label{sec:notations}
We use calligraphic letters to denote sets, uppercase for random variables and distributions, lowercase for values, and we use bold characters to denote vectors. For $n\in\NN$ we let $[n]=\{1,2\ldots n\}$. For a set $\cS$ we write $s\from\cS$ to indicate that $s$ is selected uniformly at random from $\cS$. Given a random variable (or a distribution) $X$, we write $x\from X$ to indicate that $x$ is selected according to $X$. For a set $\cS\su\NN$ and a natural number $n\in\NN$, we denote $\cS+n=\sset{s+n:s\in\cS}$. For a vector $\vv$ of dimension $n$ and for $i\in[n]$, we write $v_i$ or $v(i)$ for its \ith entry. For a set $\cI\su[n]$ we denote by either $\vv_\cI$ or $\vv(\cI)$ the vector  $(v_i)_{i\in\cI}$. A \ppt is probabilistic polynomial time, and a \pptm is a \ppt (interactive) Turing machine. 

Given a graph $G=(V,E)$ and a vertex $v\in V$, we let $\neigh_G(v)=\sset{u\in V:\sset{u,v}\in E}$ be the set of neighbors of $v$ in the graph $G$. For two vertices $u,v\in V$, we denote by $\dist_G(u,v)$ the distance between $u$ and $v$, defined as the length of the shortest path between them (set to $\infty$ if there is no such path). Finally, we let $\diam(G)=\max_{u,v\in V}\dist_G(u,v)$ denote the diameter of $G$.

A function $\mu(\cdot)$ is \emph{negligible} 
if for every positive polynomial $q(\cdot)$ and all sufficiently large $\secParam$ it holds that $\mu(\secParam) < 1/q(\secParam)$. We write $\negl(\cdot)$, $\poly(\cdot)$, and $\polylog(\cdot)$ for an unspecified negligible, polynomial, and polylogarithmic function, respectively.

A \emph{distribution ensemble} $X=\sset{X_{a,n}}_{a\in\cD_n,n\in\NN}$ is an infinite sequence of random variables indexed by $a\in\cD_n$ and $n\in\NN$, where $\cD_n$ is a domain that might depend on $n$. Computational indistinguishability is defined as follows.
\begin{definition}\label{def:compInd}
	Let $X=\sset{X_{a,n}}_{a\in\cD_n,n\in\NN}$ and $Y=\sset{Y_{a,n}}_{a\in\cD_n,n\in\NN}$ be two ensembles. We say that $X$ and $Y$ are computationally indistinguishable, denoted $X\comp Y$, if for every non-uniform \ppt distinguisher $\Dist$, there exists a negligible function $\mu(\cdot)$, such that for all $n$ and $a\in\cD_n$, it holds that
	\[
	\abs{\pr{\Dist(X_{a,n})=1}-\pr{\Dist(Y_{a,n})=1}}\leq\mu(n).
	\]
\end{definition}

\begin{fact}[Chernoff's bound]\label{fact:chernoff}
	Let $X_1,\ldots,X_n$ be i.i.d random variables over $\zo$, and let $\mu=n\cdot\ex{X_1}$. Then for every $t>0$ it holds that
	$$\pr{\sum_{i=1}^n X_i>(t+1)\mu}<e^{-\min\sset{t^2/4,t/2}\cdot\mu}.$$
\end{fact}

\begin{fact}[One-sided Hoeffding's inequality for random subsets]\label{fact:hoeffding}
	Let $m,n\in\NN$, where $n<m$, and let $\cS\su[m]$ be some set. Suppose a set $T\su[m]$ of size $n$ is sampled uniformly at random. Then for all $t>0$ it holds that
	$$\pr{|T\cap\cS|-n\cdot|\cS|/m\geq t}\leq e^{-2t^2/n},$$
	and that
	$$\pr{|T\cap\cS|-n\cdot|\cS|/m\leq -t}\leq e^{-2t^2/n},$$
\end{fact}

\subsection{Secure Computation and the Model of Computation}
\label{sec:Background}

The security of multiparty computation protocols is defined via the real
vs.~ideal paradigm~\cite{Gol04,Can00}. According to this paradigm, a protocol in the \emph{real-world} model, i.e., where actual
protocols are executed, is deemed secure if it ``emulates'' the formulated \emph{ideal-model} for executing the task. This ideal-model involves a trusted party whose functionality captures the security requirements of the task, specifically, an adversary in this model is very limited in what it can do. To show that the real-world protocol emulates the ideal-world  protocol, it is required to show that for any real-life adversary $\adv$, there exists an ideal-model adversary $\Sim$ (called simulator) such that the global output of an execution of the protocol with $\adv$ in the real-world model is
distributed similarly to the global output of running  $\Sim$ in the ideal
model.

\subsection*{The Real Model}
A multiparty protocol with $\partNum$ parties is defined by $\partNum$
interactive probabilistic polynomial-time Turing machines
$\Pc_1,\ldots,\Pc_\partNum$. Each Turing machine (party) holds at the beginning of the execution the common security parameter $1^{\secParam}$, a private input, and random coins.
The \emph{adversary} $\adv$ is a \emph{non-uniform} interactive Turing machine, receiving an auxiliary information $\aux\in\zos$, describing the behavior of the corrupted parties. It starts the execution with input that contains the identity of the corrupted parties, their inputs, and an additional auxiliary input $\aux$.

The parties execute the protocol over a synchronous network. That is, the execution proceeds in rounds: each round consists of a \emph{send phase} (where parties send their messages for this round) followed by a \emph{receive phase} (where they receive messages from other parties). 

Throughout the execution of the protocol, all the honest parties follow the instructions of the prescribed protocol, whereas the corrupted parties receive their instructions from the adversary. The adversary is considered to be \emph{malicious}, meaning that it can instruct the corrupted parties to deviate from the protocol in any arbitrary way. Additionally, the adversary has full-access to the view of the corrupted parties, which consists of their inputs, their random coins, and the messages they see throughout this execution. At the conclusion of the execution, the honest parties output their prescribed output from the protocol, the corrupted parties output nothing, and the adversary outputs a function of its view (containing the views of the corrupted party).

\subsubsection*{The \LOVE Model and \LOVE Functionalities}
In the \LOVE model, we consider $\partNum+1$ parties, $\Pc_0,\ldots,\Pc_\partNum$. We refer to $\Pc_0$ as the server and the other parties as users. The mindset is that there is a huge number of users, all of which are much less computationally powerful than the server. Formally, although all parties are polynomial in the security parameter $\secParam$, we assume the users $\Pc_1,\ldots,\Pc_\partNum$ to be \emph{polylogarithmic} in $\partNum$, and the server to be polynomial in $\partNum$. Furthermore, since we require the server to be computationally bounded, we assume that $\partNum$ is subexponential in $\secParam$, \ie $\partNum=2^{o(\secParam)}$, or equivalently, $\secParam=\omega(\log n)$. Finally, we let $\partNum$ be held in binary by the users, and in unary by the server. 
Additionally, when we say the adversary is \ppt, it means that it is polynomial in both $\secParam$ and $n$. Similarly, when we say a function is negligible, it is shorthand to saying it is negligible in both $\secParam$ and $n$.

We consider a complete point-to-point network. Here, every pair of parties is connected via a secure and authenticated channel, and thus the adversary cannot read or modify messages sent between two honest parties. We assume the parties \emph{do not} have access to a broadcast channel. We assume that the server, once corrupted, can \emph{disconnect} any pair of parties. In more details, before the ``send phase'' of any round, the server receives the list of pairs of users $(\Pc_i,\Pc_j)$, where $\Pc_i$ is going to send a message to $\Pc_j$ in the next round, together with the length of the message. Based on this information, a corrupted server can decide which of the messages it blocks (without obtaining the contents of the messages at any point). \changed{Additionally, users may abort, however, whenever this occurs we assume that this information is given only to server, and to users that try to interact with the aborting users.}

\begin{remark}
In \cref{sec:star} we consider a star network, where the users are connected only to the server. Additionally, in \cref{sec:star} (and only there) we assume PKI of a specific form, suitable for the GMPC model, which we call {\em GPKI}. Intuitively, GPKI is a PKI where Gulliver cannot invent names/public-keys or hide existing ones. We stress that, with the exception of \cref{sec:star}, we {\em do not} assume PKI in any of the other parts of this work.
\end{remark}

\paragraph{Distributed denial of service attacks.}
As mentioned in the introduction, an adversary can flood the network with garbage messages, even if the server is honest. We propose two solutions in order to prevent this. Let us start with the first solution. We define a global value $\Delta=\poly(\log n,\secParam)$ that bounds the number of users that any other user can send a message to in any given round. Moreover, we require the protocols to be such that in an honest execution, the probability that there exists a user receiving more than, say $100\Delta$, messages from different users is negligible (in both $\secParam$ and $n$). Now, whenever some user $i$ receives messages from more than, say, $\Delta^3$ different users, then the \emph{honest} server (who knows which parties interacted at any given round) blocks \emph{all} users that interacted with user $i$.
Note that this includes blocking {\em honest} users, but not too many of them since for every honest user that is being blocked there are {\em many} malicious users that are being blocked.
For our purposes, this solution suffices since this implies that the fraction of malicious users is getting smaller by performing this attack. Note that this will affect the definition of the ideal world. 

Let us now present the second solution. The protocol we construct has the property that for every round, either all users interact with a random subset of the users, or they interact with users that the server can infer from its view. In the latter set of rounds, if a malicious user tries to send a message to a user outside of the set held by the server, then the server simply blocks the user and labels it malicious. In the former case, we let each user and the server interact via a coin-tossing into to well (two-party) protocol where at the end of its execution, the user holds a random subset of the users, and the server holds a commitment to this set. When the users are required to sample a set of users to interact with, they simply open the commitment to the server, and use the subset to which they are committed to. If a malicious user tries to send a message to a user outside of its committed set, the server will catch it with overwhelming probability, and, similarly to the previous case, the server will block the user before it sends the messages.

To simplify the presentation, in all of our protocols and proofs, we assume that the adversary does not perform a DDoS attack in case the server is honest.

\paragraph{\LOVE protocol and functionalities.} 
We call a protocol in the above model an $n$-user \LOVE protocol. We next define \LOVE functionality. Roughly, these are solitary-output functionalities, where the server alone obtain the output. Furthermore, the length of the inputs and the output is $\tilde{O}(n)$.
\begin{definition}[\LOVE functionality]
	An $\partNum$-ary functionality $f=\set{f_\secParam}_{\secParam \in\N}$ is a sequence of polynomial-time computable, randomized mappings $f_\secParam:(\Dom_{\secParam})^\partNum \rightarrow (\Rng_{\secParam})^\partNum$.\footnote{The typical convention is to have $\Dom_{\secParam}=\Rng_{\secParam}=\zos$. However, in this work, we deal with functionalities over a domain that is polylogarithmic in $n$, which is why we introduce this notation.} To alleviate notations, we sometimes omit $\secParam$ from functions of $\secParam$, \eg we write $\Dom$ instead of $\Dom_{\secParam}$.
	We call an $(n+1)$-ary functionality $f$ \LOVE functionality if $\log|\Rng|=n\cdot\poly(\secParam,\log n)$ and only the first party (\ie the server) obtains an output.
\end{definition}

For a protocol $\Pi$ and an adversary $\adv$, we denote by $\Real_{\Pi,\Adv(\aux)}(\vx,\secParam,n)$ the joint output of the adversary $\adv$ and the server (assuming it is honest), in a random execution of $\Pi$ on security parameter $\secParam\in\NN$, inputs $\vx=(x_0,\ldots,x_n)\in\zos$, the number of users $n\in\NN$, and an auxiliary input $\aux\in\zo^*$.

\subsection*{The Ideal Model}
We consider an ideal computation with \emph{guaranteed output delivery} (also referred to as \emph{full security}), where a trusted party performs the computation on behalf of the parties, and the ideal-world adversary \emph{cannot} abort the computation. There is one subtlety in the definition. Observe that in the real world, a corrupted server can block some of the users, claiming them to be malicious. Thus, we have to allow a corrupted server to do same in the ideal world.

We next present an ideal execution for the computing a \LOVE functionality $f$ assuming a corrupted server. The case where the server is honest is presented in \cref{sec:honest_server} (where we define security assuming the adversary cannot perform the DDoS attack described in the previous section). %

\paragraph{Ideal world for full security with blocking assuming a corrupted server.} 
We next describe the interaction in the ideal world assuming the server is corrupted. Let $\adv$ be an adversary corrupting a subset $\corrset\su[n]$ of the users, which also corrupts the server. In the following, we let $\alpha$ denote an upper bound on the fraction of users that can be corrupted by an adversary (which is known to all parties).

\begin{description}

\item[Inputs:] Each party holds the security parameter $1^\secParam$ and the number of users $n$ (held in binary by the users and in unary by the server). Additionally, the server holds input $x_0$, user $\Pc_i$ holds $x_i$, and the adversary is given auxiliary input $\aux\in\zos$.

\item[Parties send inputs to trusted party:] The honest parties send their inputs to the  trusted party. For each corrupted party, the adversary $\adv$ sends to the trusted party some value from their domain as input. Additionally, the adversary sends a set $\block\su[n]$ of users of size at most $\alpha n$. For every user $i\in\block$ the trusted party replaces $x_i$ with a default input from the same domain. Denote by $(\partyInputTag{1},\ldots,\partyInputTag{\partNum})$ the
tuple of inputs received \changed{(and possibly modified)} by the trusted party.

\item[Trusted party sends output to the server:] The trusted party computes $\outValue\la f(\partyInputTag{1},\ldots,\partyInputTag{\partNum})$
with uniformly random coins and sends the output $\outValue$ to the server.

\item[Outputs:] An honest server outputs the value sent by the trusted party, and a corrupted server outputs nothing. Additionally, all users output nothing and $\adv$ outputs a function of its view (its inputs, the output, and the auxiliary input $\aux$).
\end{description}

Let $\Ideal_{f,\adv(\aux)}(\vecInput,\secParam,n)$ be the 
random variable consisting of the output of the adversary $\adv$ in this ideal world execution and the output of the honest parties in the execution.

We next define secure computation. For an elaborate discussion on this notion, see~\cite{Gol04}.
\begin{definition}[malicious security]
\label{def:1Overp-security}
Let $n=n(\secParam)$ denote the number of users, let $m=m(\secParam,n)=n\cdot\poly(\secParam,\log n)$ denote the length of each input and the output, let $t=t(\secParam,n)$ be such that $t<n$ denote a bound on the number of corrupted users, and let $f:\sparen{\zo^m}^{n+1}\mapsto\zo^m$ be a \LOVE functionality. An $\partNum$-party \LOVE protocol $\Pi$ computing $f$ is said
to be $t$-secure, if for every non-uniform probabilistic polynomial-time adversary $\Adv$ in the real model, controlling 
at most $t$ user and which possibly also corrupts the server, there
exists a non-uniform probabilistic polynomial-time adversary $\Sim$  in the ideal model, controlling the same parties as $\Adv$, such that the following holds
$$
\set{\Ideal_{f,\Sim(\aux)}(\vecInput,\secParam,n)}_{\aux \in \set{0,1}^*,\vecInput \in \Dom^\partNum,\secParam\in\N}
\quad\comp\quad
  \set{\Real_{\Pi,\Adv(\aux)}(\vecInput,\secParam,n)}_{\aux \in \set{0,1}^*,\vecInput \in \Dom^\partNum,\secParam\in\N}.
  $$
\end{definition}

%%%%%%%%%%%%%%%%%%%%%%%%%%%%%%%%%%%%%%%%%%%%%%%%%%%%%%%
\subsubsection*{The Hybrid Model}\label{sec:hybrid_model}

The \emph{hybrid model} is a model that extends the real model with a trusted party that provides ideal computation for specific functionalities. The parties communicate with this trusted party in exactly the same way as in the ideal models described above.

Let $f$ be a functionality. Then, an execution of a protocol $\Pi$ computing a functionality $g$ in the $f$-hybrid model involves the parties sending normal messages to each other (as in the real model) and in addition, having access to a trusted party computing $f$. 
We consider the setting where the parties may invoke several functionalities concurrently, and where the parties can invoke one functionality during the call to another.

%%%%%%%%%%%%%%%%%%%%%%%%%%%%%%%%%%%%%%%%%%%%%%%%%%%%%%%

\subsubsection*{Security Under Composition of Protocols}\label{sec:secComposition}
The security notion defined above is known as stand-alone security, as it deals with an execution of single protocol, executed in isolation. Generalized security notions take into consideration possible executions of other protocols that run concurrently over the same communication network. Most notable are security under general concurrent composition and universal composition \cite{Canetti01}. These definitions deal with a more realistic setting where protocols are executed in an unknown environment, and the protocol may be  liable to attacks that are not possible in the stand alone setting. Furthermore, and more relevant to this work, it is often useful to construct protocols as a composition of several sub-protocols that are executed concurrently.  

The composition theorem of \citet{Canetti01} states the following. Let $\rho$ be a protocol that securely computes $f$. Then, if a protocol $\pi$ computes $g$ in the $f$-hybrid model, then the protocol $\pi^\rho$, that is obtained from $\pi$ by replacing all ideal calls to the trusted party computing $f$ with the protocol $\rho$, securely computes $g$ in the real model.

\begin{theorem}[\cite{Canetti01}]\label{thm:Composition}
Let $n=n(\secParam)$ denote the number of users, let $m=m(\secParam,n)=n\cdot\poly(\secParam,\log n)$ denote the length of each input and the output, let $t=t(\secParam,n)$ be such that $t<n$ denote a bound on the number of corrupted users, and let $f:\sparen{\zo^m}^{n+1}\mapsto\zo^m$ be a \LOVE functionality. Suppose we are given a protocol $\rho$ computing $f$ with $t$ security. Further suppose that there exists a protocol $\pi$ computing $g$ with $t$-security in the $f$-hybrid model.
Then protocol $\pi^\rho$ computes $g$ with $t$ security in the real model.
\end{theorem}
%%%%%%%%%%%%%%%%%%%%%%%%%%%%%%%%%%%%%%%%%%%%%%%%%%%%%%%

\subsection{Cryptographic Tools}
we use fairly standard cryptographic tools, such as secret sharing, signatures and fully homomorphic encryption (the latter being the most ``exotic" one).  
\subsubsection{Secret Sharing}\label{sec:secretsharing}
A (threshold) secret-sharing scheme~\cite{Shamir79} is a method in which a dealer distributes shares of some secret to $n$ parties such that $t$ colluding parties do not learn anything about the secret, and any subset of $t+1$ parties can fully reconstruct the secret.

\begin{definition}[secret sharing]\label{def:TSS}
A \emph{$(t+1)$-out-of-$n$ secret-sharing scheme} over a message space $\cM$ consists of a pair of algorithms $(\Share, \Recon)$ satisfying the following properties:
\begin{enumerate}
    \item $t$-{\bf privacy:}
    For every secret $m\in \cM$, and every subset $\IS\subseteq[n]$ of size $\ssize{\IS}\leq t$, the distribution of the shares $\sset{s_i}_{i\in \IS}$ is independent of $m$, where $(s_1,\ldots,s_n)\gets \Share(m)$.
    
    \item $(t+1)$-{\bf reconstructability:}
    For every secret $m\in \cM$, every subset $\IS\subseteq[n]$ of size $t+1$, every set of shares $\vs=(s_1,\ldots,s_n)$ s.t.\  $\ppr{\vS\gets \Share(m)}{\vS=\vs}>0$ and every vector   $\vs'=(s'_1,\ldots,s'_n)$ s.t.\ $\vs_{\IS}=\vs'_{\IS}$ and $\vs'_{\bar{\IS}} = \bot^{\ssize{\bar{\IS}}}$ it holds that $m=\Recon(\vs')$.
\end{enumerate}
\end{definition}

An error-correcting secret-sharing (ECSS) scheme is a secret-sharing schemes, in which the reconstruction is guaranteed to succeed even if up to $t$ shares are faulty.
This primitive has also been referred to as \emph{robust secret sharing} or as \emph{honest-dealer VSS}~\cite{RB89,CFOR12,CDF01}.

\begin{definition}[error-correcting secret sharing]\label{def:ECSS}
A \emph{$(t+1)$-out-of-$n$ error-correcting secret-sharing scheme} (ECSS) over a message space $\cM$ consists of a pair of algorithms $(\Share, \Recon)$ satisfying the following properties:
\begin{enumerate}
    \item $t$-{\bf privacy:} As in \ref{def:TSS}.
    
    \item {\bf Reconstruction from up to $t$ erroneous shares:}
    For every secret $m\in \cM$, every shares $\vs = (s_1,\ldots,s_n)$, and every  $\vs' = (s'_1,\ldots,s'_n)$ such that $\ppr{\vS\gets \Share(m)}{\vS=\vs}>0$ and $\ssize{\sset{i \mid s_i=s'_i}}\geq n-t$, it holds that $m=\Recon(\vs')$ (except for a negligible probability).
\end{enumerate}
\end{definition}

ECSS can be constructed with perfect correctness when $t<n/3$ using Reed-Solomon decoding~\cite{BGW88} and with a negligible error probability when $t<n/2$ by authenticating the shares using one-time MAC~\cite{RB89}.
In case $t\geq n/2$ it is impossible to construct a $(t+1)$-out-of-$n$ ECSS scheme, or even a secret-sharing scheme that identifies cheaters~\cite{IOS12}.

\subsubsection{Vector Commitments}

Our discussion follows the treatment of Fisch~\cite{Fisch18}. A vector commitment (VC)~\cite{CatalanoF13,LibertY10} is a cryptographic commitment to an ordered sequence of $m$
values $(x_1, \ldots, x_m)$ that admits succinct openings at specific positions (e.g.,\ prove that $x_i$ is the $i^{\text{th}}$ committed message). For security, VCs are required to satisfy position binding, which states
that an adversary should not be able to open a commitment to two different values at the same position. Moreover, VCs are required to be concise, i.e.\ the size of the commitment string and of its openings is independent of the vector length. Usually VCs are also required to be hiding, meaning that opening at several positions does not leak any information about the committed values at other positions. A Merkle tree is an example of a simple vector commitment that is binding and concise but not hiding.

\paragraph{Vector commitment syntax.} We provide a redacted syntax for vector commitments, taken from~\cite{Fisch18}. A vector commitment scheme $\VC=(\VCSetup,\VCCom,\VCOpen,\VCVerify)$ is a 4-tuple of \ppt algorithms described as follows.

\begin{enumerate}
	\item {\bf Setup:} The setup algorithm $\pp\from\VCSetup(1^\secParam,m,\cM)$ is given the security parameter $\secParam$, length $m$ of the vector, and message space of vector components $\cM$. It outputs the public parameters $\pp$, which are implicit inputs to all the following algorithms.
	
	\item {\bf Commit:} The commitment algorithm $(\tau,c)\from\VCCom_{\pp}(\vv)$ takes an input vector $\vv=(v_1,\dots,v_m)$ and outputs a commitment $c$ and an advice $\tau$.
	
	\item {\bf Open:} The opening algorithm $\Lambda_{\cS}\from\VCOpen_{\pp}(\vv,c,\cS,\tau)$ opens the commitment $c$ of the message $\vv$ at locations $\cS\su[m]$. It is given the advice $\tau$ generated by the commit algorithm. The output $\Lambda_{\cS}$ proves that $v_i$ is the \ith committed element of $c$ for all $\cS(i)$ (where $\cS(i)$ denoted the \ith element of $\cS$).
	
	\item {\bf Verify:} The verification algorithm $b\from\VCVerify_{\pp}(c,v'_1,\dots,v'_q,\cS,\Lambda_{\cS})$
	takes as input the commitment $c$, a vector $\cS\su[m]$ of indices, and an opening proof $\Lambda_{\cS}$. It outputs a bit $b\in\zo$ such that $b=1$ (accept) if and only if $\Lambda_{\cS}$ is a valid proof that $c$ is a commitment to a vector $\vv$, \ie $v_{\cS(i)}=v'_i$ for all $i\in[q]$. If $\cS=\emptyset$ then $\Lambda_{\emptyset}$ should be a normal opening, \ie a proof that $c$ is a commitment to $\vv$, namely $v_i=v'_i$ for all $i\in[m]$.
\end{enumerate}

\paragraph{Binding commitments.} The main security property of vector commitments (of interest in the present work) is position binding. The security game augments the standard binding commitment game. Roughly, the security guarantees that no \ppt adversary can generate a commitment that can be decommitted to two different vectors.

\begin{definition}
A vector commitment scheme \VC is said to be \emph{position binding} if for all $\poly(\secParam,m)$-time  adversaries $\AAA$ and for all $\cS',\cS\subseteq[m]$ with $|\cS|=q$ and $|\cS'|=q'$, and for all $i\in[q]$ and $j\in[q']$, the following probability is at most negligible in $\secParam$:
$$
\Pr_{\substack{\pp\leftarrow\VCSetup(1^\secParam,m,\cM)\\(c,\vv,\vv',\Lambda,\Lambda')\leftarrow\AAA(\pp)}}\left[
\begin{array}{c}
\VCVerify(c,\vv,\cS,\Lambda)=1
\wedge
\VCVerify(c,\vv',\cS',\Lambda')=1\\
\wedge\;
\cS(i)=\cS'(j) \wedge v_i\neq v'_j
\end{array}
\right]
$$
\end{definition}

\subsubsection{Digital Signatures}

A digital signature is a scheme for presenting the authenticity of digital messages or documents. We follow the presentation in~\cite{KatzLindell2007}.

\begin{definition}
A {\em signature scheme} $\Sig=(\SigGen,\SigSign,\SigVerify)$ is a 3-tuple of \ppt algorithms described as follows.
\begin{enumerate}
	\item {\bf Key generation:} The {\em key generation} algorithm $(\pk,\sk)\from\SigGen\sof{1^{\secParam}}$ takes as input a security parameter $1^\secParam$ and outputs a pair of keys $(\pk,\sk)$, called the {\em public key} and {\em private key}, respectively. We assume for convenience that $\pk$ and $\sk$ each have length at least $\secParam$, and that $\secParam$ can be determined from $\pk$ and $\sk$.
	
	\item {\bf Sign:} The {\em signing algorithm} $\sigma\from\SigSign_{\sk}(m)$ takes as input the private-key $\sk$ and a message $m\in\{0, 1 \}^*$, and outputs a signature $\sigma$.
	
	\item {\bf Verify:} The deterministic {\em verification} algorithm $b\from\SigVerify_{\pk}\sof{m,\sigma}$ takes as input the public-key $\pk$, a message $m$, and a signature $\sigma$. It outputs a bit $b\in\zo$ such that $b=1$ (accept) if and only if $\sigma$ is a valid signature of $m$.
\end{enumerate}
It is required that for every $\secParam\in\NN$, every $(\pk,\sk)$ output by $\SigGen(1^\secParam)$, and every message $m\in\zos$, it holds that $\SigVerify_{\pk}(m,\SigSign_{\sk}(m))=1$.
\end{definition}

\paragraph{Security of signature schemes.} We next define security of signature schemes. We call $\sigma$ a {\em valid signature} on a message $m$ (with respect to some
public key $\pk$) if $\SigVerify_{\pk}(m ,\sigma) = 1$. We say that an adversary {\em forges} a signature if it outputs a message $m$ along with a valid signature $\sigma$ on $m$, and furthermore $m$ was not previously signed using the secret key.

Let $\Sig=(\SigGen,\SigSign,\SigVerify)$ be a signature scheme, and consider the following experiment for an adversary $\AAA$ and security parameter $\secParam$, denoted as $\sff{Sig\text{-}forge}_{\AAA,\Sig}(\secParam)$:

\begin{enumerate}
	\item $\SigGen(1^\secParam)$ is run to obtain keys $(\pk,\sk)$
	
	\item Adversary $\AAA$ is given $\pk$ and oracle access to $\SigSign_{\sk}(\cdot)$ (this oracle returns a signature $\SigSign_{\sk}(m)$ for any message $m$ of the adversary's choice). The adversary then outputs a message-signature pair $(m, \sigma)$. Let $\cQ$ denote the set of messages whose signatures were requested by $\AAA$ during its execution.
	
	\item The output of the experiment is 1 if $\SigVerify_{\pk}(m,\sigma)=1$ and $m\notin \cQ$, and is defined to be 0 otherwise.
\end{enumerate}

The following defines the unforgeability property required from signature schemes. We also consider security against adversaries that can run in polynomial time in a possibly larger parameter $n(\secParam)$.
\begin{definition}
A signature scheme $\Sig$ is {\em existentially
unforgeable under an adaptively chosen message attack} if for all \ppt adversaries $\AAA$, there exists a negligible function $\negl$ such
that:
$$
\Pr\left[ \sff{Sig\text{-}forge}_{\AAA,\Sig}(\secParam)=1 \right]\leq\negl(\secParam).
$$

For any function $n=n(\secParam)$, the signature scheme is said to be \emph{$n$-secure}  if the above holds \wrt any adversary that runs in $\poly(\secParam,n)$ time.
\end{definition}

\subsubsection{Fully Homomorphic Encryption}
We next define fully homomorphic encryption (FHE) schemes \cite{Gen09}.
\begin{definition}
	A {\em fully homomorphic encryption scheme}  over a circuit family $\C$, is a 4-tuple of \ppt algorithms $\FHE=(\FHEGen,\FHEEnc,\FHEDec,\FHEEval)$:
	\begin{enumerate}
		\item {\bf Key generation:} The {\em key generation} algorithm $(\pk,\sk)\from\FHEGen\sof{1^{\secParam}}$ takes as input a security parameter $1^\secParam$ and outputs a pair of keys $(\pk,\sk)$, called the public key and private key, respectively.% We assume for convenience that $\pk$ and $\sk$ each have length at least $\secParam$, and that $\secParam$ can be determined from $\pk$ and $\sk$.
		
		\item {\bf Encrypt:} The {\em encryption algorithm} $c\from\FHEEnc_{\pk}(m)$ takes as input the public-key $\pk$ and a plaintext $m\in\{0, 1 \}^*$, and outputs a ciphertext $c$.
		
		\item {\bf Decrypt:} The deterministic {\em decryption} algorithm $m^*=\FHEDec_{\sk}\sof{c}$ takes as input the private-key $\pk$ and a ciphertext $c$. It outputs $m^*\in\zos\cup\sset{\bot}$ which is either a valid plaintext or a special $\bot$ symbol denoting failure.
		
		\item {\bf Evaluate:} The {\em evaluation algorithm} $c'\from\FHEEval_{\pk}\sof{C,c}$ takes as input the public key $\pk$, a circuit $C\in\C$, and a ciphertext $c$. It outputs a new ciphertext $c'$,
	\end{enumerate}
	The following are required:
	\begin{enumerate}
		\item {\bf Correctness:} For every $\secParam\in\NN$, every $(\pk,\sk)$ output by $\FHEGen(1^\secParam)$, every circuit $C\in\C$, and every message $m\in\zos$, it holds that for a corresponding ciphertext $c=\FHEEnc_{\pk}(m)$, $$\FHEDec_{\sk}\of{\FHEEval_{\evk}\sof{C,c}}=C(m).$$
		
		\item {\bf Compactness:} There exists a polynomial $p$ such that for every $\secParam\in\NN$, the decryption algorithm $\FHEDec$ can be expressed as a circuit of size at most $p(\secParam)$.
	\end{enumerate}
\end{definition}
\paragraph{Security against chosen-plaintext attacks.} To
define chosen-plaintext attack (CPA) security of FHE schemes, let $\FHE=(\FHEGen,\FHEEnc,\FHEDec,\FHEEval)$ be an FHE scheme. Consider the following experiment for an adversary $\adv$ and security parameter $\secParam$, denoted $\FHECPA_{\adv,\FHE}(\secParam)$:
\begin{enumerate}
	\item $\FHEGen(1^{\secParam})$ is run to obtain keys $(\pk,\sk)$.
	
	\item Adversary $\adv$ is given $\pk$ and oracle access to $\FHEEnc_{\pk}(\cdot)$. The adversary outputs a pair of messages $m_0$ and $m_1$ of equal length.
	
	\item A bit $b\from\zo$ is sampled uniformly at random, and then the ciphertext $c\from\FHEEnc_{\pk}(m_b)$.
	
	\item The adversary (still having oracle access to $\FHEEnc_{\pk}(\cdot)$) outputs a bit $b'$.
	
	\item The output of the experiment is defined to be 1 if $b'=b$, and is defined to be 0 otherwise.
\end{enumerate}

The following defines CPA security as we require from FHE schemes. Similarly to signature schemes, we also consider security against adversaries that can run in polynomial times in a possibly larger parameter $n(\secParam)$.
\begin{definition}
	An FHE scheme is said to be \emph{CPA-secure} if for all \ppt adversaries $\adv$, there exists a negligible function $\negl$ such that
	$$\pr{\FHECPA_{\adv,\FHE}(\secParam)=1}\leq\frac12+\negl(\secParam).$$
	For any function $n=n(\secParam)$, the scheme is said to be \emph{$n$-secure} if the above holds \wrt any adversary that runs in $\poly(\secParam,n)$ time.
\end{definition}

\subsection{Probabilistically Checkable Proofs of Proximity}\label{sec:PCPP}
A \emph{probabilistically
	checkable proof of proximity} (PCPP) \cite{BGH+04a,DR04} is proof system, which allows a verifier to be convinced that the input is close to being in a language $L$. In more details, the verifier has an explicit input $x$ and an implicit input $y$ given as an oracle. The verifier accepts with high probability if $y$ is close to some $y'$ such that $(x,y')\in L$.  We next formalize the notion of PCPP. The definitions below are taken almost verbatim from \cite{BGHSV05}. We start with defining the Hamming distance.
\begin{definition}[Hamming distance]
	Let $n\in\NN$. The \emph{Hamming distance} between two strings $x,x'\in\zo^n$ is defined as $\Delta(x,x'):=\sabs{\sset{i\in[n]:x_i\ne x'_i}}$. For a string $x\in\zo^n$ and a set $\cS\su\zo^n$ we let $\Delta(x,\cS)=\min_{x'\in\cS}\delta(x,x')$. Finally, a string $x\in\zo^n$ is said to be $\delta$-far from a set $\cS\su\zo^n$ if $\Delta(x,\cS)>\delta$.
\end{definition}

\begin{definition}[Restricted verifier]	
	Let $r,q:\NN\mapsto\NN$ and $t:\NN\times\NN\mapsto\NN$. An \emph{$(r,q,t)$-restricted verifier} is a probabilistic oracle Turing machine $\Vc$, that is given a string $x$ number $k\in\NN$ (in binary), an oracle access to an input $y\in\zo^k$ and a proof $\pi\in\zos$, tosses $r(|x|+k)$ coins, queries the oracle $(y,\pi)$ at most $q(|x|+k)$ times, runs in time $t(|x|,k)$, and outputs either 0 or 1.
\end{definition}

For a pair language $L\su\zos\times\zos$ and a string $x\in\zos$, we let $L_x=\sset{y\in\zos:(x,y)\in L}$.
\begin{definition}[PCPP for pair languages]
	For functions $r,q:\NN\mapsto\NN$, $t:\NN\times\NN\mapsto\NN$, and $s,\delta:\NN\mapsto[0,1]$, a pair language $L\su\zos\times\zos$ is said to be in $\PCPP_{s,\delta}[r,q,t]$ if there exists an $(r,q,t)$-restricted verifier $\Vc$ such that the following holds.
	
	\begin{description}
		\item[Completeness:] If $(x,y)\in L$ then there exists a proof $\pi\in\zos$ such that
		$$\pr{\Vc^{y,\pi}\of{x,\abs{y}}=1}=1,$$
		where the probability is taken over the random coin tosses of $\Vc$.
		
		\item[Soundness:] If $(x,y)$ is such that $y$ is $\delta(|x|+|y|)$-far from $L_x\cap\zo^{|y|}$, then for every proof $\pi\in\zos$ it holds that
		$$\pr{\Vc^{y,\pi}\of{x,\abs{y}}=1}\leq s\of{\abs{x}+\abs{y}},$$
		where the probability is taken over the random coin tosses of $\Vc$.
	\end{description}
\end{definition}

\begin{theorem}[Efficient PCPPs for pair languages \cite{BGHSV05}]\label{thm:PCPP}
	For every pair language $L\in\operatorname{NP}$ and every constant $s>0$, it holds that $L\in\PCPP_{s,\delta}[r,q,t]$ where
	\begin{itemize}
		\item $\delta(m)=m/\polylog (m)$,
		\item $r(m)=O(\log(m))$,
		\item $q(m)=\polylog(m)$,
		\item $t(n,k)=\poly(n,\log (n+k))$.
	\end{itemize}

	Moreover, it is implicitly stated that for every $(x,y)\in L$ a correct proof can be generated in polynomial time in $|x|+|y|$ \changed{(given a witness)}.
\end{theorem}

Repeating the proof a polylogarithmic number of times results in a negligible soundness error.
\begin{corollary}\label{cor:PCPP}
	For every pair language $L\in\operatorname{NP}$ it holds that $L\in\PCPP_{s,\delta}[r,q,t]$ where
	\begin{itemize}
		\item $s(m)=e^{-\log^2 m}$.
		\item $\delta(m)=m/\polylog (m)$,
		\item $r(m)=\polylog(m)$,
		\item $q(m)=\polylog(m)$,
		\item $t(n,k)=\poly(n,\log (n+k))$.
	\end{itemize}
	
	Moreover, it is implicitly stated that for every $(x,y)\in L$ a correct proof can be generated in polynomial time in $|x|+|y|$. 
\end{corollary}

\subsection{Committee Election}\label{sec:Feige}
Feige's lightest-bin protocol~\cite{Feige99} is an elegant $n$-party, public-coin protocol, consisting of a single broadcast round, for electing a committee of size $n'<n$, in the information-theoretic setting.
Each party uniformly selects one of $\ceil{n/n'}$ bins and broadcasts it choice. The parties that selected the lightest bin are elected to participate in the committee. The protocol ensures that the ratio of corrupted parties in the elected committee is similar to their ratio in the population.
The original protocol in~\cite{Feige99} considered committees of size $\log(n)$, however, this results with a \emph{non-negligible} failure probability.
\cite[Lem.\ 2.6]{BGK11} analyzed Feige's protocol for arbitrary committee sizes and proved the following lemma.

\begin{lemma}[\cite{BGK11}]\label{lem:Feige}
For integers $n'<n$ and constants $0<\alpha<\alpha'<1$ define
\[
\err\left(n,n',\alpha,\alpha'\right)=\frac{n}{n'} \cdot e^{-\frac{(\alpha'-\alpha)^2 n'}{2(1-\alpha)}}.
\]
Feige's lightest-bin protocol is a $1$-round, $n$-party protocol for electing a committee $\C$, such that for any set of corrupted parties $\corrset\su[n]$ of size $t\leq\alpha n$ the following holds.
\begin{enumerate}
    \item
    $\sabs{\C}\leq n'$.
    \item
    $\pr{\sabs{\C\setminus\corrset}\leq (1-\alpha')\cdot n'} < \err(n,n',\alpha,\alpha')$.
    \item
    $\pr{\sabs{\C\cap \corrset}\geq \alpha' \cdot \sabs{\C}} < \err(n,n',\alpha,\alpha')$.
\end{enumerate}
\end{lemma}

Towards proving \cref{lem:Feige}, using Chernoff's inequality and the union bound \cite{BGK11} showed the following.
\begin{lemma}[Implicit in \cite{BGK11}]\label{lem:FeigeHonestEvenSplit}
	Fix integers $n'<n$, constants $0<\alpha<\alpha'<1$, and a set $\corrset\su[n]$ of corrupted parties. For every bin $b\in[\ceil{n/n'}]$ we let $X_b$ denote the number of honest parties that sampled $b$ in Feige's protocol. Then
	
	$$\pr{\exists b\in[\ceil{n/n'}]:X_b<(1-\alpha')\cdot n'}<\err\of{n,n',\alpha,\alpha'}.$$
\end{lemma}

\section{Preventing Blocking of Honest Users}\label{sec:prevent_block}

Towards constructing our secure protocols, we first develop two tools. The goal of these two tools is to prevent a malicious server from blocking honest users. The first tool is the sampling of \emph{personal committees} (PC) that are chosen uniformly at random and will replace each user. The idea is that with high probability every PC of an honest party (hereinafter, honest PC) will contain an honest majority. Furthermore, if the server is honest, then the same holds for malicious users. Thus, if a PC aborts after all of them were sampled, then this constitutes a proof of the server being malicious. The second tool is locally sampling a communication graph (\ie each PC samples a small number of neighbors) with a small diameter that the parties will use to essentially allow any single PC to raise a flag indicating to all other PCs that the server is malicious. If the server is honest, then since all PCs contain an honest majority, no PC would ever raise such a flag.

We stress that the security properties of the construction are not defined via the real vs. ideal paradigm. Instead we define a ``good'' event and show that it occurs except with negligible probability (in both the security parameter $\secParam$ and the number of users $n$). 
\changed{In \cref{sec:pc} we present a protocol for sampling personal committees. Afterwards, in \cref{sec:graph} we present the formal definition of the graph distribution we use, and prove that with high probability it has a small diameter. Then, in \cref{sec:setup} we combine the two results, and present the final setup protocol for sampling the PCs and graph with the desired security properties.}

%%%%%%%%%%%%%%%%%%%%%%%%%%%%%%%%%%%%%%%%%%%%%%%%%

\subsection{Personal Committee Protocol}\label{sec:pc}

We present a protocol for sampling the personal committees. Recall that the goal of the PCs is to replace each user, thus any malicious behaviour of  the committee as a whole implies that the server is malicious. Before formally describing the protocol, we first give a short overview of the construction. 

First, each user samples together with the server a PC using a coin-tossing protocol. From here, the goal of the honest users and the (honest) server, is to provide a proof for other users that the PC was sampled by them. However, as the users are assumed to be polylogarithmic in $n$, they cannot actually hold the entire proof. Instead, we let each user verify consistency of information with \changed{$\secParam=\omega(\log n)$} randomly sampled other users. If not too many users aborted, then except with negligible probability, the server must have played honestly.

The idea is as follows. The server will commit to an array of length $n$ whose \ith entry is the \changed{PC of} user $i$ (set to $\bot$ if the user aborted early). Each user then verifies that it appears in this array, and in addition, samples $\secParam$ other users uniformly at random. These $\secParam$ users will be used for consistency checks. %

The neighboring users compare their information, and further verify it with the server (by requesting to open the commitment \changed{at} the correct positions). Observe that it could be the case that some of the users are blocked by a malicious server, and the user cannot distinguish this from the case of a malicious user.
However, if the server is honest then, with overwhelming probability, at least $\approx(1-\alpha)\secParam$ users will respond with information that is consistent with the information provided by the server. Thus, we let each user verify that this is indeed the case.

We then let each user notify the users in its PC that they belong to its personal committee. These users then verify this with the server, and ignores the message (without aborting) from the user in case the server's answer is inconsistent. We then let the users in each PC \emph{broadcast} to all other users in the same PC that they participate in the PC,\footnote{Formally speaking, the parties compute the multicast functionality, where only a subset of the users obtain the output.} and verify that at least $1-\alpha$ fraction of them sent a message. To implement the broadcast channel, the users will use a variant of the broadcast protocol due to \cite{PSL80,LSP82} that is secure against at most $\alpha<\changed{1/8}$ corruptions. We refer the reader to \changed{\cref{sec:comm_love}} for a detailed description of the broadcast protocol. Note that if the server is honest, then each PC will contain more than roughly $(1-\alpha)\secParam$ honest users. Therefore, the users can safely abort if less than $(1-\alpha)\secParam$ users are active in the PC. Finally, to ensure that most honest users remain, each of them samples $\secParam$ random users, notifies them they are alive, and requests a feedback. If less than roughly $\alpha\secParam$ of them did not respond, then the user aborts.

We now present the personal committee protocol. In the following we let $$\VC=(\VCSetup,\VCCom,\VCOpen,\VCVerify)$$ be a vector commitment scheme, let $\phC$ be a perfectly hiding commitment scheme\footnote{Using a perfectly hiding commitment scheme is done only for convenience.}, and let \changed{$\eps=1/8-\alpha$} be a positive constant. %

\noindent
\begin{protocol}{\PerComm}\label{proto:pc}
	
	{\bf Common inputs:} All parties hold the security parameter $1^{\secParam}$ and the number of users $n$ (held in binary by the users, and in unary by the server). 
	
	\smallskip
	\begin{enumerate}[leftmargin=15pt,rightmargin=10pt,itemsep=1pt,topsep=0pt]
		
		\item\label{step:PC_sample} The server and each user $i$ interact as follows.
		\begin{enumerate}[topsep=0pt]
			\item Let $k=\pcsize\cdot\log n$ be the length of the representation of each PC. User $i$ samples a random string $r_i\from\zo^{k}$, computes a commitment $\hc_i\from\phCCom(r_i)$, and sends $\hc_i$ to the server.
			
			\item The server responds with a random string $s_i\from\zo^k$.
			
			\item Upon receiving $s_i$ from the server, user $i$ computes $\P_i:=r_i\xor s_i$ and sends to the server a decommitment to $\hc_i$. In case the server does not send $s_i$ to user $i$, the user outputs $\bot$. 
			
			\item If the decommitment of user $i$ is invalid, the server labels $i$ as inactive.

		\end{enumerate}
		
		\item Each remaining user $i$ computes the public parameters $\pp_i\from\VCSetup\sof{1^{\secParam},n,[n]\times \binom{[n]}{\pcsize}}$ for the commitment, and sends $\pp_i$ to the server.
		
		\item The server does the following:
		\begin{enumerate}[topsep=0pt]	
			\item Set $\vv$ to be the vector of length $n$ whose \ith entry is $\P_i$ in case $i$ is active, and is set to $\bot$ otherwise.
			
			\item For each remaining user $i$ do the following:
			\begin{enumerate}[topsep=0pt]
				\item Compute the commitment $(\tau_i,c_i)\from\VCCom_{\pp_i}\sof{\vv}$ of the vector $\vv$.
				
				\item Open the commitment at position $i$ to obtain $\Lambda_i\from\VCOpen_{\pp_i}\of{\vv,c_i,i,\tau_i}$.
				
				\item Send $c_i$ and $\Lambda_i$ to user $i$.
			\end{enumerate}
		\end{enumerate}
		
		\item\label{step:local_abort} User $i$ aborts if the server did not decommit properly, \ie $\VCVerify_{\pp_i}\of{c_i,\P_i,i,\Lambda_i}=0$.
		
		\item Each user $i$ does the following:
		\begin{enumerate}[topsep=0pt]
			\item Sample a set of users $\rndsu^{\out}_i\su[n]$ of size $\secParam$ uniformly at random and notify each of them and the server\footnote{Note that by the assumption on the network, the server can obtain $\rndsu^{\out}_i$ by simply observing which parties interact. We decided to explicitly have the user notify the server for the sake of presentation.} that they were sampled by it. 

            \item\label{step:too_many_neighbors} Set $\rndsu^{\intext}_i$ to be the set of users that sampled it. It aborts if it was sampled by too many users, \ie, $\sabs{\rndsu^{\intext}_i}>3\secParam$.
            
            \item\label{step:send_info_to_neighbors} Otherwise, send $(\pp_i,c_i)$ to all of users in $\rndsu^{\intext}_i$.
            
            \item\label{step:get_commitments} Let $\rndsu_i\su\rndsu^{\out}_i$ be the set of users $j$ that responded with a message $(\pp_j,c_j)$.

 			\item For each $j\in\rndsu_i$, compare the information with the server: The server decommits to $c_j$ at location $i$ by sending $\Lambda_{i,j}\from\VCOpen_{\pp_j}\sof{\vv,c_j,i,\tau_j}$ to user $i$.
			\item\label{step:global_abort} Check that at least $(1-2\alpha-\eps/2)\secParam$ of the users in $\rndsu_i$ sent a message that is consistent with the server, namely, $\VCVerify_{\pp_j}\sof{c_i,\P_i,i,\Lambda_{i,j}}=1$ for all $j\in\rndsu_i$. If this is not the case, then abort.
		\end{enumerate}
	\end{enumerate}
\end{protocol}

\addtocounter{protocol}{-1}

\begin{protocol}{\PerComm (Continued)}
	\begin{enumerate}[leftmargin=15pt,rightmargin=10pt,itemsep=1pt,topsep=0pt]
		\setcounter{enumi}{5}

		\item The server and each user $i$ interact as follows:
		\begin{enumerate}
		    \item The user samples a \emph{query set} $\queryset_i\su[n]$ and sends it to the server.
		    
		    \item The server opens each of the commitments $c_j$, where $j\in\rndsu_i\cup\sset{i}$ and locations $\queryset_i$. That is, it sends $\vv(\queryset_i)$  and $\Lambda^{\queryset}_{i,j}\from\VCOpen_{\pp_j}\sof{\vv,c_j,\queryset_i,\tau_j}$
		    for every $j\in\rndsu_i\cup\sset{i}$.
		    
		    \item\label{step:queryset_consistency} Let $\vu_i$ be the purported value of $\vv(\queryset_i)$ sent by the server. The user verifies all decommitments, \ie it checks that	    $\VCVerify_{\pp_j}\sof{c_j,\vu_i,\queryset_i,\Lambda^{\queryset}_{i,j}}=1$
		    for all $j\in\rndsu_i\cup\sset{i}$. If this is not the case then the user aborts.
		\end{enumerate}

		\item Each remaining user $i$ sends the set $\P_i$ to all users in $\P_i$.
		
		\item\label{step:too_many_PCs} If some user $j$ received more than $3\secParam$ messages, then it aborts. 
		
		\item Otherwise, for every $\P_i\ni j$ received, user $j$ compares this information with the server: the server decommits to $c_j$ at location $i$ by sending $\Lambda'_{i,j}\from\VCOpen_{\pp_j}\sof{\vv,c_j,i,\tau_j}$ to user $j$.
		
		\item If $\VCVerify_{\pp_j}\sof{c_j,\P_i,i,\Lambda'_{i,j}}=0$, then the user \emph{ignores} the message without aborting.

		\item\label{step:bc_in_pc} For each personal committee $\P_i$, all users in it broadcast \alive to all other users in $\P_i$ (see \changed{\cref{sec:comm_love}} for a secure implementation).
		
		\item The personal committee $\P_i$ is labeled inactive and does not participate in subsequent interaction, if more than $\alpha+\eps/2$ of its users did not send \alive, in which case user $i$ aborts.
		
		\item\label{step:many_remain} Each remaining user $i$ samples $\secParam$ users and sends them the message \alive. It then replies with \alive for every \alive received.
		
		\item User $i$ aborts if it received less than $(1-2\alpha-\eps/2)\secParam$ replies.
	\end{enumerate}

	\noindent {\bf Output phase: } Each user $i$ outputs $\sset{(j,\P_j)}_{j:i\in\P_j}$, $\pp_i$, and $c_i$. An honest server outputs the list $\LIVE=\sset{(i,\P_i)}_i$ of alive users and their corresponding PC, and $(\pp_i,c_i,\tau_i)_{i\in\LIVE}$.
\end{protocol}

The security properties of the protocol (stated as an event occurring with overwhelming probability) appears in \cref{sec:setup}, as part of the security properties of the setup protocol presented in the same section.

We next claim that all users are polylogarithmic \changed{in $n$}. By construction, this boils down to showing that no user will be sampled by too many users (in any of the sets sampled throughout the protocol). Observe that this directly follows from Chernoff's inequality (\cref{fact:chernoff}) since every user is sampled by at most $\secParam$ user on expectation.

\subsection{A Locally Sampled Communication Graph With a Small Diameter}\label{sec:graph}
In this section we discuss \changed{the  distribution on graphs that we use for sampling the communication graph.} Recall that the goal of the graph is to allow the users to quickly notify each other if any of them caught the server cheating. The distribution over the graphs that we use is the union of stars, each of a fixed size $d$, sampled uniformly at random. We next formalize the definition of this distribution.

\begin{definition}
	Let $n,d\in\NN$ where $1\leq d\leq n$ and let $\badset\su[n]$. We define the distribution \changed{\Gdist{n}{d,\badset}} over \emph{directed} graphs $G=([n],E)$ with $n$ vertices as follows. For each vertex \changed{$i\in[n]\setminus\badset$} sample a subset $\isampled_i\su [n]$ of size $d$ and set $(i,j)\in E$ (\ie there is an edge from $i$ to $j$) if and only if $j\in\isampled_i$.
\end{definition}

The following lemma asserts that with overwhelming probability, for a graph sampled according to the distribution \Gdist{n}{d,\cS} any induced subgraph over a sufficiently large constant fraction of the vertices of $[n]\setminus\badset$, will have a small diameter with high probability.

Looking ahead, each vertex in the graph will represent a personal committee of some user. Therefore, the lemma can be used to show that the induced subgraph over the set of honest PCs will have a small diameter with overwhelming probability. However, recall that a malicious server might abort some honest users. % 
\changed{Thus, the set of ``bad'' users, which the induced subgraph ignores, is of size $2\alpha n$, and not $\alpha n$.}
In the following, we let % 
\changed{$\beta=2\alpha<1/4$} denote an upper bound on the fraction of malicious and blocked users.

\begin{lemma}\label{lem:smalldiam}
	Let $n,d\in\NN$ such that $4\leq d\leq n$, let $\beta<\changed{1/4}$, let $\diamLength=\log_{d/4}\sof{n/3}$, and let $\badset\su[n]$ be of size $\sabs{\badset}\leq \beta n$. Consider the following distribution over graphs with $n$ vertices: sample $G\from\Gdist{n}{d,\badset}$ and let $G'$ denote the induced subgraph of $G$ over the vertices $[n]\setminus\badset$.  Then
	$$\pr{\diam\of{G'}>\diamLength+1}<n^2\cdot(\diamLength+1)\cdot e^{-\changed{4d/9}},$$
	where the probability is taken over the sampling of $G'$.
\end{lemma}

\begin{proof}%
	Let $u,v\in[n]\setminus\badset$ be two distinct vertices in $G'$. We show that
	\begin{align}\label{eq:smalldist}
		\pr{\dist_{G'}\of{u,v}>\diamLength+1}<\paren{\diamLength+1}\cdot e^{-\changed{4d/9}}.
	\end{align}
	The proof then follows from the union bound.
	
	We define the following sequence of random variables. For \changed{$i\in\sset{0,\ldots,\ell}$} let 
	$$S_i=\set{w\in[n]\setminus\badset:\dist\of{u,w}=i}$$
	be the set of vertices of distance \emph{exactly} $i$ from $u$ in the graph $G'$, and let $S^*_i:=\cup_{j=0}^{i}S_j$ denote the set of distance at most $i$ from $u$ in the graph $G'$. 
	It suffices to show that except with probability at most $\diamLength\cdot e^{-\changed{4d/9}}$, it holds that $\sabs{S^*_{\diamLength}}\geq n/3$. Indeed, let $\sampledme_v=\sset{i\in[n]:v\in\isample_i}$ denote the set of vertices that sampled $v$. Then
	$$\pr{\sampledme_v\cap S^*_{\diamLength}=\emptyset\cond \abs{S^*_{\diamLength}}\geq n/3}
	\leq\paren{1-\frac{d}{(1-\beta)n}}^{\abs{S_{\diamLength}}}<e^{-\frac{d}{3(1-\beta)}}< e^{-\changed{4d/9}}.$$

	Toward proving that $\sabs{S^*_{\diamLength}}\geq n/3$ holds with high probability, we prove the following claim. Roughly, the claim asserts that for every $i\leq\diamLength$, if $|S^*_{\diamLength}|<n/3$ and $S_{i-1}$ is somewhat large, then $S_i$ is somewhat large as well with high probability. 
	
	\begin{claim}\label{clm:graph_exp_growth}
		For every $i\in\sset{0,\ldots,\diamLength}$ it holds that
		$$\pr{\abs{S_i}<\paren{d/4}^i\bigcond\abs{S^*_{\diamLength}}<n/3\wedge \forall j\in[i-1]:|S_j|\geq(d/4)^j}\leq e^{-d/8},$$
		if $i>0$, and that
		$$\pr{\abs{S_i}<\paren{d/4}^i\bigcond\abs{S^*_{\diamLength}}<n/3}=0,$$
		if $i=0$.
	\end{claim}
	
	The claim is proven below. We first show that it indeed implies that $|S^*_\diamLength|\geq n/3$ with high probability. Since the $S_i$'s are pairwise disjoint and $\diamLength$ satisfies
	$$\sum_{i=0}^{\diamLength}(d/4)^i=\tfrac{(d/4)^{\diamLength+1}-1}{d/4-1}\geq(d/4)^\diamLength=n/3,$$
	it follows that if $|S^*_{\diamLength}|<n/3$ then there exists $i\in\sset{0,\ldots,\diamLength}$ such that $|S_i|<(d/4)^i$. %
	Thus
	\begin{align*}
		\pr{\abs{S^*_{\diamLength}}<n/3}
		&=\pr{\abs{S^*_{\diamLength}}<n/3\wedge\exists i\in[\diamLength]:\abs{S_i}<(d/4)^i}\\
		&\leq\pr{\exists i\in[\diamLength]:\abs{S_i}<(d/4)^i\bigcond\abs{S^*_{\diamLength}}<n/3}\\
		&\leq\sum_{i=1}^{\diamLength}\pr{\abs{S_i}<(d/4)^i\bigcond\abs{S^*_{\diamLength}}<n/3\wedge\forall j\in[i-1]:|S_j|\geq(d/4)^j}\\
		&\leq\diamLength\cdot e^{-d/8}
	\end{align*}
	where the second inequality follows from the union bound. \changed{It remains to prove \cref{clm:graph_exp_growth}.}
	\begin{proofof}{\cref{clm:graph_exp_growth}}
		The proof is  by induction on $i$. Clearly, the claim holds for $i=0$. Assume that the claim holds for $i\leq \diamLength-1$. We next prove it for $i+1$. 
		First, observe that if $|S^*_{\diamLength}|<n/3<\tfrac{(1-\beta)}{2}\cdot n$, then half of the vertices in $G'$ are not in $S_{i+1}$. Thus, for any fixation $\cS_i$ of $S_i$, the (conditional) expected size of $S_{i+1}$ is at least
		\begin{align*}
			\ex{\abs{S_{i+1}}\bigcond|S^*_{\diamLength}|<n/3\wedge S_i=\cS_i}
			\geq d\cdot|\cS_i|/2.
		\end{align*}
		Therefore,
		\begin{align*}%
			\ex{\abs{S_{i+1}}\bigcond|S^*_{\diamLength}|<n/3\wedge |S_i|\geq(d/4)^i}
			\geq (d/2)\cdot(d/4)^i.
		\end{align*}
		In particular, it holds that 
		\begin{align*}%
		\mu:=\ex{\abs{S_{i+1}}\bigcond|S^*_{\diamLength}|<n/3\wedge\forall j\in[i]:|S_j|\geq(d/4)^j}
			\geq (d/2)\cdot(d/4)^i.
		\end{align*}
		Thus, by \cref{fact:hoeffding}, it follows that
		\begin{align*}
			&\pr{\abs{S_{i+1}}<\paren{d/4}^{i+1}\bigcond\abs{S^*_{\diamLength}}<n/3\wedge\forall j\in[i]:|S_j|\geq(d/4)^j}\\
			&\quad\leq\pr{\abs{S_{i+1}}-\mu<-\paren{d/4}^{i+1}\bigcond\abs{S^*_{\diamLength}}<n/3\wedge\forall j\in[i]:|S_j|\geq(d/4)^j}\\
			&\quad\leq e^{-\frac{2(d/4)^2}{d}}\\
			&\quad=e^{-d/8}.
		\end{align*}
	\end{proofof}
\end{proof}

%%%%%%%%%%%%%%%%%%%%%%%%%%%%%%%%%%%%%%%%%%%%%%%%%

\subsection{Preparing For Computation in the GMPC Model: The Setup Protocol}\label{sec:setup}

In this section, we combine \cref{proto:pc} for sampling PCs and the result from \cref{sec:graph} for sampling a graph, and construct our setup protocol. Roughly, the protocol proceeds by first executing \cref{proto:pc}, which samples the PCs and forms some sort of agreement (see \cref{thm:PC_secure} below the protocol's description). Then, each PC samples its neighbors uniformly at random. To ensure that the graph is undirected, the user requests a receipt from each of the users it sampled, which in turn add the requesting users to their set of neighbors.

We abuse notions and describe the protocol as if each PC is a single party. Formally, whenever we say that the server sends a message to $\P_i$, it means that it (supposedly) shares it among the users of $\P_i$. Similarly, when we say that $\P_i$ sends a message to the server, it means that every user in $\P_i$ sends its share of the message to the server. Finally, whenever a PC $\P_i$ sends a message $\mathsf{msg}$ to another PC $\P_j$, it means that each user in $\P_i$ shares its share of $\mathsf{msg}$ among the users in $\P_j$. 

Formally, we describe the protocol in a hybrid world, where the hybrid functionalities compute the ``next-message'' function of each PC.\footnote{Although the PCs are sampled during the execution of the protocol, the next-message they compute are functionalities that are known in advance. Therefore, the hybrid functionality is well-defined.} Let $\nxtmsg_i$ denote the next-message function of $\P_i$. The inputs of the users to each call of $\nxtmsg_i$ are a $(1-\alpha-\eps/2)\secParam$-out-of-$\secParam$ Shamir's secret sharing scheme of the view of $\P_i$. If the server or $i$ are honest and at least $(1-\alpha-\eps/2)\secParam$ inputs are provided to the functionality, then it proceeds to compute the next-message as specified by the protocol, and sharing the output among the users in $\P_i$. If less then $(1-\alpha-\eps/2)\secParam$ inputs are provided, the functionality sends $\bot$ to all of its users. If both the server and $i$ are corrupted, then the adversary chooses the output the users receive from the functionality. In all cases, if a user obtain $\bot$ as the output from some $\nxtmsg_i$, then it aborts. In \changed{\cref{sec:comm_love}} we present a protocol for implementing each call.

We further describe the output \wrt the PC rather than the users. The output of each user is then defined as the set of PCs containing it, and a share of the output of each such PC, in a $(1-\alpha-\eps/2)\secParam$-out-of-$\secParam$ Shamir's secret sharing scheme.

\noindent
\begin{protocol}{\CommSetup}\label{proto:setup}
	
	{\bf Common inputs:} All parties hold the security parameter $1^{\secParam}$ and the number of users $n$ (held in binary by the users, and in unary by the server). 
	
	\smallskip
	\begin{enumerate}[leftmargin=15pt,rightmargin=10pt,itemsep=1pt,topsep=0pt]
		\item The parties execute protocol \PerComm. Recall that each user $i$ obtains a set $\sset{(j,\P_j)}_{j:i\in\P_j}$ of size at most $3\secParam$, and the server holds the set $\LIVE=\sset{(i,\P_i)}_i$ where each $\P_i$ is active. Additionally, the server is committed (via a vector commitment scheme) to a vector $\vv$ whose \ith entry is $(i,\P_i)$. %
		
		\item Each PC $\P_i$ samples a set of neighbors $\isample_i\su[n]$ of size $\secParam$ uniformly at random, and sends it to server.
		
		\item For each remaining $\P_i$, the server opens the commitment to every user in $\P_i$ at position $j$. A user aborts from $\P_i$ if the verification of the decommitment fails.
		
		\item The PC $\P_i$ then notifies each PC $\P_j$, where $j\in\isample_i$.
			
		\item\label{step:too_many_neighbors_graph} $\P_i$ sets $\sampledme_i$ to be the set of PCs that sampled it. It aborts if it was sampled by too many users, \ie, $\sabs{\sampledme_i}>3\secParam$.
		
		\item Let $\totalsample_i=\isampled_i\cup\sampledme_i$ be the total set of neighbors of PC $\P_i$ (known to both the PC and the server).
		
		\item Send $\P_i$ to each PC $\P_j$, where $j\in\totalsample_i$.
		
		\item\label{step:check_alive_setup} For \FILreps iterations, each personal committee $\P_i$ sends \alive to every $\P_j$ where $j\in\totalsample_i$. At any iteration, if $\P_i$ did not receive the message \alive from any of its neighbors, it aborts.
		
		\item Each honest $\P_i$ outputs the set of neighbors $\totalsample_i$, and an honest server outputs the list $\LIVE=\sset{(i,\P_i)}_i$ of alive users and their corresponding PC., and the (undirected) graph $G=(\LIVE,E)$, where $\sset{(i,\P_i),(j,\P_j)}\in E$ if and only if $i\in\totalsample_j$.
	\end{enumerate}
\end{protocol}

We next state the security properties of \cref{proto:setup}. Instead of describing it using the standard real vs. ideal paradigm, we formalize the security notion by describing an event, which captures the desired properties. We then show that the event occurs with overwhelming probability.
\begin{definition}[Personal committees and small graph event]
	We define the \emph{personal committees and small graph event}, denoted \PCSG, as \PCSGhon if the server is honest, and define it as \PCSGmal if the server is malicious, where \PCSGhon and \PCSGmal are defined as follows. 
	
	\PCSGhon is the conjunction of the following events.
	\begin{itemize}
		\item All honest users are alive, and all their PCs are labeled active.
		
		\item Each alive user $i$ (whether it is honest or malicious) holds a personal committee $\P_i$ of size $\secParam$ and a set of neighbors $\totalsample_i$ of size at most $4\secParam$, where both are known to the server.
		
		\item For every alive user $i$ (whether it is honest or malicious), its PC $\P_i$ contains at least $(1-\alpha-\eps/2)\secParam>\changed{7\secParam/8}$ honest users. 
		
		\item For every alive $i$, all of its neighbors in $\totalsample_i$ hold $(i,\P_i)$.
		
		\item For every alive $i$ and every alive $j\in\totalsample_i$, all users in $\P_i$ hold $(j,\P_j)$.
		
		\item The graph over the PCs of alive users induced by the set of neighbors $\sset{\totalsample_i}_i$, has diameter at most \FILreps.
	\end{itemize}
	\PCSGmal is the conjunction of the following events.
	\begin{itemize}
		\item At least $(1-2\alpha-\eps)n=(\changed{7/8}-\alpha)n$ honest users are alive, and their PCs are labeled active.
		
		\item Each alive user $i$ (whether it is honest or malicious), holds a personal committee $\P_i$ of size $\secParam$ and a set of neighbors $\totalsample_i$ of size at most $4\secParam$, where both are known to the server.
		
		\item For every alive and honest $i$, its PC $\P_i$ contains at least $(1-\alpha-\eps/2)\secParam>\changed{7\secParam/8}$ honest users.
		
		\item For every alive $i$, all of its neighbors in $\totalsample_i$ hold $(i,\P_i)$.
		
		\item For every alive $i$ and every alive $j\in\totalsample_i$, all users in $\P_i$ hold $(j,\P_j)$.
		
		\item The graph over the PCs of alive honest users induced by the set of neighbors $\sset{\totalsample_i}_i$, has diameter at most \FILreps.
	\end{itemize}
\end{definition}

The next theorem asserts that at the end of \cref{proto:setup}, the event \PCSG holds except with negligible probability.

\begin{theorem}\label{thm:PC_secure}
	Let $\alpha<\changed{1/8}$ be a constant (\ie it does not depend on $\secParam$ or the number of users $n$). Assume the existence of position binding vector commitment schemes. Let $\adv$ be a \ppt for \cref{proto:setup} corrupting at most $\alpha n$ users, and which possibly corrupts the server as well. Then the event \PCSG occurs except with negligible probability.
\end{theorem}

Observe that by construction, if a PC contains at most $(\alpha+\eps/2)\secParam$ corrupted users and it is labeled active at the end of the execution, then all of its users hold the neighbors of the user corresponding to the PC, and PCs of the neighbors. Thus, \cref{thm:PC_secure} directly follows from the next two lemmata, which split (and simplify) the statement depending on whether the server is honest or corrupted.
\begin{lemma}\label{lem:PC_secure_hon}
 	Let $\alpha<\changed{1/8}$ be a constant and assume the existence of position binding vector commitment schemes. Let $\adv$ be a \ppt adversary \changed{for \cref{proto:setup}}, corrupting a set $\corrset\su[n]$ of at most $\alpha n$ users and which does \emph{not} corrupt the server. Then the following hold.
	\begin{enumerate}
		\item\label{item:all_PC_many_honest} Every personal committee (including those of malicious users that did not abort) will contain less than $\changed{\secParam/8}$ corrupted users with overwhelming probability. That is, let $\corrset'\su\corrset$ denote the set of malicious users that are labeled active by the server by the end of the protocol's execution. Then
		$$\pr{\exists i\in\honset\cup\corrset':\abs{\P_i\cap\corrset}\geq\paren{\alpha+\frac{\eps}{2}}\cdot\pcsize}\leq n\cdot e^{-\frac{\eps^2}{2}\cdot\pcsize}.$$
		
		\item\label{item:all_remain} All honest users remain alive at the end of the execution, except with negligible probability. 
	\end{enumerate}
\end{lemma}

\begin{lemma}\label{lem:PC_secure_mal}
	Let $\alpha<\changed{1/8}$ be a constant and assume the existence of position binding vector commitment schemes. Let $\adv$ be a \ppt adversary \changed{for \cref{proto:setup}}, corrupting a set $\corrset\su[n]$ of at most $\alpha n$ users and which also corrupts the server. Let $\LIVE\su\honset$ denote the set of honest users that did not abort the execution of the protocol. Assume that $\LIVE\ne\emptyset$. Then the following hold.
	\begin{enumerate}
		\item\label{item:many_remain} $|\LIVE|\geq(1-2\alpha-\eps)n=(\changed{7/8}-\alpha)n$ except with negligible probability.
		
		\item\label{item:PC_many_honest} Every personal committee of an alive honest user in will contain less than $\changed{\secParam/8}$ malicious users, except with negligible probability. Formally,
		$$\pr{\exists i\in\LIVE:\abs{\P_i\cap\corrset}\geq\paren{\alpha+\frac{\eps}{2}}\cdot\pcsize}\leq n\cdot e^{-\frac{\eps^2}{2}\cdot\pcsize}.$$
		In particular, the PCs of alive honest users are active at the end of the execution.
	
		\item\label{item:small_diam} The graph over the active honest PCs induced by the sets $\sset{\totalsample_i}_i$ (\ie $\sset{\P_i,\P_j}$ is an edge if and only if $i\in\totalsampled_{j}$ or $j\in\totalsample_i$), has diameter at most \FILreps, except with negligible probability. Moreover, each user has at most $4\secParam$ neighbors.
		
		\item\label{item:agreement} For every alive honest user  $i\in\LIVE$, all its neighbors $j\in\totalsampled_i$ and the users in $\P_j$, hold $\P_i$ (as it was sent to them at Step~\ref{step:send_info_to_neighbors}).
	\end{enumerate}
\end{lemma}

We first prove \cref{lem:PC_secure_hon}, where the server is assumed to be honest.

\begin{proofof}{\cref{lem:PC_secure_hon}}	
	We first prove \cref{item:all_PC_many_honest}. First, by the binding property of the commitment scheme, any malicious user will decommit properly, except with negligible probability. Therefore, $\P_i$ will be uniformly random for all $i\in\corrset'$ except with negligible probability. Now, observe that for every $i\in\honset\cup\corrset'$ it holds that

	$$\ex{\abs{\P_i\cap\corrset'}}\leq\abs{\corrset}\cdot\frac{\pcsize}{n}\leq\alpha\pcsize.$$
	By Hoeffding's inequality (\cref{fact:hoeffding}), for every user $i\in\honset\cup\corrset$, it holds that
	\begin{align*}
		\pr{\abs{\P_i\cap\corrset}\geq\paren{\alpha+\frac{\eps}{2}}\cdot\pcsize}
		\leq\pr{\abs{\P_i\cap\corrset}-\ex{\abs{\P_i\cap\corrset}}\geq\frac{\eps}{2}\cdot\pcsize}
		\leq e^{-\frac{\eps^2}{2}\cdot\pcsize}.
	\end{align*}
	The claim now follows from the union bound.
	
	We now show \cref{item:all_remain}. Since the server is honest, all users will obtain a commitment to the same vector. Thus, all users remain if and only if no user will be sampled by too many other users (in Steps~\ref{step:too_many_neighbors} and \ref{step:too_many_PCs}) and if no personal committee will be labeled inactive. The former holds with negligible probability and it directly follows from Chernoff's inequality (\cref{fact:chernoff}) and the union bound. Then latter follows from \cref{item:all_PC_many_honest} and the security of the broadcast protocol.
\end{proofof}

We now prove \cref{lem:PC_secure_mal}, where the server is assumed to be malicious.

\begin{proofof}{\cref{lem:PC_secure_mal}}
	For simplicity, throughout the entire proof we assume the malicious server always opened the commitments properly. Furthermore, for an alive honest user $i\in\LIVE$ we let $\vv_i$ denote the vector to which the (malicious) server send a commitment to user $i$.\footnote{Note that the existence of $\vv_i$ is only guaranteed from a computational perspective, namely, the computationally bounded malicious server cannot decommit to any other value. To formalize this argument, observe that if there exists another vector $\vv'_i$ to which the server is able to decommit without user $i$ aborting, then this directly translate to an attacker breaking the binding assumption of the vector commitment scheme.}
	
	\paragraph{Proof of \cref{item:many_remain}.} Assume that less than $(1-2\alpha-\eps)n$ honest users are alive, and let $\rndsuT_i$ denote the set of users sampled by user $i$ at Step~\ref{step:many_remain}. Then $\ex{|\rndsuT_i\cap\LIVE|}<(1-2\alpha-\eps)\secParam$. Therefore, by Hoeffding's inequality (\cref{fact:hoeffding}) user $i$ aborts at Step~\ref{step:many_remain} except with probability at most
	$$\pr{|\rndsuT_i\cap\LIVE|\geq\paren{1-2\alpha-\frac{\eps}{2}}\cdot\secParam}
	\leq\pr{|\rndsuT_i\cap\LIVE|-\ex{|\rndsuT_i\cap\LIVE|}\geq\frac{\eps}{2}\cdot\secParam}
	\leq e^{-\frac{\eps^2}{2}\cdot\secParam}.
	$$
	\cref{item:many_remain} now follows from the union bound.
	
	\paragraph{Proof of \cref{item:PC_many_honest}.} Assume towards contradiction that the inequality is false for some adversary $\adv$ corrupting the server. Then there exists an honest user $i\in\LIVE$ whose PC does not contain at least $(\alpha+\eps/2)\secParam$ honest users, with probability at least $e^{-\eps^2\cdot\pcsize/2}$. Consider the following thought experiment, where two parties, a user $u$ and a server, interact. 
	At the beginning of the interaction, the (malicious) server chooses two subsets $\corrset$ and $\LIVE$, where $|\corrset|\leq\alpha n$. The server and user $u$	agree on a random subset $\P\su[n]$ of size $\pcsize$, using the same interaction as in Step~\ref{step:PC_sample} of the protocol.
	We say that the server wins if it succeeds in forcing the output $\P$ to admit $|\P\cap\corrset|\geq(\alpha+\eps/2)\secParam$. Similarly to the case where the server is honest (see \cref{item:all_PC_many_honest} of \cref{lem:PC_secure_hon}), by \cref{fact:hoeffding}, for every malicious server (that decommits properly) it holds that the probability that the adversary wins is
	\begin{align*}
		\pr{\abs{\P\cap\corrset}\geq\paren{\alpha+\frac{\eps}{2}}\cdot\pcsize}\leq e^{-\frac{\eps^2}{2}\cdot\pcsize}.
	\end{align*} 
	We now show that using the above adversary $\adv$ for the multiparty protocol, the server in the two-party experiment can achieve its goal with probability higher than the above quantity. The server (in the two-party setting) sets $\corrset$ to be the set of users corrupted by $\adv$, and interacts with the user $u$ the same as $\adv$ interacts with user $i$, emulating in its head the entire multiparty protocol (except the interaction with $i$). Clearly, the server in the two-party protocol succeeds with too high probability, resulting in a contradiction.
	
	\paragraph{Proof of \cref{item:small_diam}.} Consider the directed graph $G$ over the users (and not the PCs), where the edges between honest alive users are defined via the sets $\sset{\isample_i}_i$ (\ie there is an edge from $i$ to $j$ if and only if $j\in\isampled_i$). Then at least $(1-2\alpha-\eps/2)n>\changed{3n/4}$ of its vertices are users in $\LIVE$. Now,  by \cref{item:PC_many_honest} there exists at least one alive honest PC. Therefore, by \cref{lem:smalldiam} applied for the graph $G$ with $\badset:=[n]\setminus\LIVE$, it follows that if one PC of an alive honest user aborted, then all such PCs would have aborted at Step~\ref{step:check_alive_setup}, except with negligible probability. Thus, the PCs are all active, and the corresponding graph has diameter \FILreps, except with negligible probability. 
	
	To see the ``moreover'' part, observe that if a user has more than $4\secParam$ neighbors, then more than $3\secParam$ sampled it. Thus, it must have aborted at Step~\ref{step:too_many_neighbors}. 
	
	\paragraph{Proof of \cref{item:agreement}.} Fix $i\in\LIVE$. We prove the stronger statement asserting that for every alive honest user $j\in\LIVE$, the proper decommitment to $c_j$ at position $i$ returns $\P_i$. We show this in two steps. First, we show that for all $i\in\LIVE$, it holds that the majority (rather than all) of alive honest users $j$ are such that the proper decommitment of $c_j$ at position $i$ results in $\P_i$. In the second step, we conclude that the above statement holds for
	\emph{all} users $j\in\LIVE$.%
	
	Before presenting the formal proof, let us first provide an intuitive explanation. For the first step, assume towards contradiction that the majority of alive honest users $j$ disagree with user $i$ at position $i$ in their committed vector. Then the set of users $\rndsu_i$ to which user $i$ received the commitments at Step~\ref{step:get_commitments}, on expectation will contain a majority of users $j$ that disagree at position $i$, thus user $i$ will sample many disagreeing users with overwhelming probability, causing it to abort. For the second step of the proof, if most users agree at position $i$, then the query set $\queryset$ of any user that disagrees with the majority, will contain at least one of those in the majority with overwhelming probability, thus causing the user in the minority to abort.
	
	We now formalize the above intuition. Let $\LIVE_i:=\sset{j\in\LIVE:\vv_j(i)=\P_i}$ denote the set of all alive honest users $j$ that agree with $i$ at position $i$. We first show that $|\LIVE_i|\geq(1/2+\eps/2)n=(\changed{9/16}-\alpha/2)n>n/2$. Assuming otherwise, it holds that $\ex{\rndsu_i\cap\LIVE_i}<(1/2+\eps/2)\secParam$, hence by \cref{fact:hoeffding} it follows that
	\begin{align*}
		\pr{\abs{\rndsu_i\cap\LIVE_i}\geq\paren{\frac12+\eps}\cdot\secParam}
		\leq\pr{\abs{\rndsu_i\cap\LIVE_i}-\ex{\abs{\rndsu_i\cap\LIVE_i}}\geq\frac{\eps}{2}\cdot\secParam}
		\leq e^{-\frac{\eps^2}{2}\cdot\secParam}.
	\end{align*}
	Thus, except with negligible probability, user $i$ sampled less than 
	$$\paren{\frac12+\eps}\secParam=\paren{\changed{\frac{5}{8}}-\alpha}\secParam<\paren{\changed{\frac{15}{16}}-\frac{3\alpha}{2}}=\paren{1-2\alpha-\frac{\eps}{2}}\secParam$$
	users with consistent information.
	However, since we assume $i$ did not abort, it follows that \emph{all} honest users it sampled must be from $\LIVE_i$, resulting in a contradiction. Therefore, the event $|\LIVE_i|<(1/2+\eps/2)n$ occurs with negligible probability.
	
	To conclude the proof of \cref{item:agreement}, we next show that any user $j\notin\LIVE_i$ aborts Step~\ref{step:queryset_consistency}, except with negligible probability. Indeed, the probability that $j$ does not abort is at most
	$$\pr{\queryset_j\cap\LIVE_i=\emptyset}\leq\pr{|\LIVE_i|<(1/2+\eps/2)n}+\pr{\queryset_j\cap\LIVE_i=\emptyset\cond|\LIVE_i|\geq(1/2+\eps/2)n}.$$
	The above quantity is negligible since
	$$\pr{\queryset_j\cap\LIVE_i=\emptyset\cond|\LIVE_i|\geq(1/2+\eps/2)n}\leq\paren{1-\frac{(1/2+\eps/2)n}{n}}^{\secParam}<e^{-\frac{1+\eps}{2}\cdot\secParam}.
	$$
\end{proofof}

%%%%%%%%%%%%%%%%%%%%%%%%%%%%%%%%%%%%%%%%%%%%%%%%%

\section{Towards Reducing User Complexities: Committee Election}\label{sec:committee}
%%%%%%%%%%%%%%%
% 
Similarly to \cite{BCDH18,BCG21,BGT13}, a key component of our construction is a committee election protocol. Our protocol follows ideas similar to Feige's lightest-bin protocol~\cite{Feige99}, but is adjusted to the special \LOVE setting, where the users are polylogarithmic (in the number of users $n$) and communication be blocked by the server. These requirements make the process of jointly agreeing on a single committee quite challenging.

\changed{We show how the parties can perform this task assuming the \PCSG event occurred, namely, the parties are replaced with personal committees and they hold together a graph with small diameter. %
}
Specifically, even if the server is malicious, then with overwhelming probability more than \changed{7/8} of users in the elected committee are honest (note that this guarantee is not the same as in Feige's protocol). In addition, the committee will be of size $\Theta(\secParam)=\omega(\log n)$.

We next formally define the committee election functionality $\felectstar$. To simplify the presentation, we split the description into two cases, depending on whether the server is honest or malicious. Note that unlike in Feige's protocols, here a malicious server can abort some honest users, hence the users do not know the number of remaining users. %

\paragraph{Ideal world for electing a committee assuming an honest server.}

We next describe the interaction in the ideal world of \felectstar assuming the server is honest. Let $\adv$ be an adversary corrupting a subset $\corrset\su[n]$ of the users, and which does \emph{not} corrupt the server.
\begin{description}
	\item[Inputs:] All parties hold the security parameter $1^{\secParam}$ and the number of users (held in unary by the server and in binary by the users). The adversary is given auxiliary input $\aux\in\zos$.
	
	\item[Adversary aborts some malicious users:] The adversary sends to $\trustp$ a set $\LIVE\su[n]$ of alive users, where $\honset\su\LIVE$.
	
	\item[Trusted party chooses a committee:] The trusted party does the following.
	\begin{enumerate}[topsep=0pt]
		\item Set $b=\ceil{|\LIVE|/\comsize}$ to be the number of bins.
		
		\item For every alive (possibly malicious) user $i\in\LIVE$, sample a bin $x_i\in[b]$ independently and uniformly at random.
		
		\item Set $\C$ to be the lightest bin and send it to all parties (note that $|\C|\leq\secParam$ by the pigeonhole principle).
	\end{enumerate}
	
	\item[Output:] Each honest user and the server output whatever it received from the trusted party, the corrupted parties output nothing, and the adversary outputs some function of its view.
\end{description}

\paragraph{Ideal world for electing a committee assuming a malicious server.}
We next describe the interaction in the ideal world of \felectstar assuming the server is corrupted. Let $\adv$ be an adversary corrupting a subset $\corrset\su[n]$ of the users, which also corrupts the server. %

\begin{description}
	\item[Inputs:] All parties hold the security parameter $1^{\secParam}$ and the number of users (held in unary by the server and in binary by the users). The adversary is given auxiliary input $\aux\in\zos$.
	
	\item[Adversary aborts some users and choose number of bins:] The adversary either sends $\abort$ to $\trustp$ or sends it a set $\LIVE\su[n]$ of alive (possibly malicious) users of size at least $\sabs{\LIVE}\geq\changed{7n/8}$, and the number of bins $b$, where $b\geq\ceil{\changed{7n/8\secParam}}$.
	
	\item[Trusted party splits users into bins:] If $\trustp$ received $\abort$, then it sends $\bot$ to all parties and halts. Otherwise, for every alive honest user $i\in\LIVE\cap\honset$, the trusted party samples a bin $x_i\in[b]$ independently and uniformly at random. It then sends $\sset{(i,x_i)}_{i\in\LIVE\cap\honset}$ to the adversary.
	
	\item[The adversary chooses a bin and adds corrupted users:] For every $j\in[b]$, let $\bin_j=\sset{i\in\LIVE\cap\honset:x_i=j}$ denote the set of honest users that sampled the \jth bin. The adversary sends back either $\abort$ or a set $\C$ of size at most $\comsize$, such that there exists $j\in[b]$ for which it holds that $\bin_j\su\C$ and $\C\setminus\bin_j\su\corrset$ (note that such a $j$ always exists if the server is honest, hence, when such a $j$ does not exist then it implies the server must behaved maliciously).
	
	\item[The trusted party sends outputs:] If $\trustp$ received $\abort$ it sends $\bot$ to all parties and halts. Otherwise, it sends $\C$ to all users in $\LIVE$ and sends $\bot$ to all other users.
	
	\item[Output:] Each honest user outputs whatever it received from the trusted party, the corrupted parties output nothing, and the adversary outputs some function of its view.
\end{description}

We prove that conditioned on the event \PCSG, the parties can compute $\felectstar$ securely. Formally, we prove the following.
\begin{theorem}\label{thm:committee}
	Let $\alpha<\changed{1/8}$ be a constant (\ie it does not depend on $\secParam$ or the number of users $n$). Then there exists an $n$-user \LOVE protocol computing $\felectstar$, such that conditioned on the event \PCSG, the protocol is $\alpha n$-secure. %
\end{theorem}

The proof of \cref{thm:committee} is presented in \cref{subsec:committee}. Then, in \cref{sec:many_committees} we present a protocol for electing many committees, such that each committee will be assigned to one user. This will later be used in \cref{sec:proof,sec:combine} for the construction of secure protocols for general \LOVE functionalities. Before proving the theorem, we first show that it contains a vast majority of honest users. In fact, we take into account the number of users stated to be alive by the adversary. This will later allow us to claim that users in the committee can securely compute any functionality (see \changed{\cref{sec:comm_love}}). Recall that we let $\eps=\changed{1/8}-\alpha>0$.
\begin{lemma}\label{lem:committee_vast_hon_maj}
	Let $\alpha<\changed{1/8}$, let $\adv$ be an adversary corrupting $\alpha n$ of the users (which possibly corrupts the server), and let $\beta=|([n]\setminus\LIVE)\cup\corrset|/n$ denote the fraction of users that are either malicious or not alive, among the set of \changed{all} users.  Then in an execution of $\felectstar$ holds that if the parties output a committee $\C$, then
	$$\pr{\sabs{\C\setminus\corrset}<\paren{1-\beta-\eps/2}\cdot\secParam} <n\cdot e^{-\frac{\eps^2}{8}\cdot\secParam}.$$
	In particular, $\C$ contains strictly more than \changed{3/4} fraction of honest users, except with negligible probability. Moreover, if the server is honest then all honest users always output $\C$.
\end{lemma}
The proof follows from an immediate application of \cref{lem:FeigeHonestEvenSplit}. \changed{To see why the fraction of honest users in $\C$ is more than $3/4$, observe that
$$\beta=\frac{|([n]\setminus\LIVE)\cup\corrset|}{n}\leq\frac18+\alpha<\frac14.$$}

%-----------------------------

\subsection{Committee Election Protocol}\label{subsec:committee}
We are now ready to introduce our protocol for electing a committee of size $\comsize$. Given the event \PCSG, the protocol for electing a committee proceeds as follows. First, the server will send to all personal committees the number of remaining users. Each PC then checks with its neighbors if they received the same number. For $O(\log n/\log\secParam)$, each PC notify its neighbors that it did not abort. 

The PCs then execute the following variant of Feige's protocol: each PC randomly choose a bin and send it to the server. The server will send back the set of users $\C$ who chose the lightest bin alongside the index of this bin. Each PC $\P_i$ then verifies that $i\in\C$ if $\P_i$ chose the bin, and $i\notin\C$ otherwise. Additionally, the PCs will check that $\C$ is not larger than $\secParam$. Finally, each PC will ask if its neighbors on the graph are still active for $O(\log n/\log\secParam)$ steps. 

Similarly to \cref{proto:setup}, we describe the protocol in a hybrid world, where the hybrid functionalities compute the ``next-message'' function of each PC.%
We also abuse notions and describe the protocol as if each PC is a single party. Formally, whenever we say that the server sends a message to $\P_i$, it means that it (supposedly) shares it among the users of $\P_i$. Similarly, when we say that $\P_i$ sends a message to the server, it means that every user in $\P_i$ sends its share of the message to the server. Finally, whenever a PC $\P_i$ sends a message $\mathsf{msg}$ to another PC $\P_j$, it means that each user in $\P_i$ shares its share of $\mathsf{msg}$ among the users in $\P_j$.
We let $\nxtmsg=\sset{\nxtmsg_i}_{i=1}^n$ and describe the protocol in the $\nxtmsg$-hybrid model. To obtain a protocol in the real world (\ie without access to $\nxtmsg$), the parties will implement each call via the protocol from \changed{\cref{sec:comm_love}}.

\begin{protocol}{\Committee}\label{protocol:comm}
{\bf Common inputs:} All parties hold the security parameter $1^{\secParam}$ and the number of users $n$ (held in binary by the users, and in unary by the server). 

{\bf Event assumption:} We assume that the event \PCSG occurred. Recall that an honest server holds the list $\LIVE=\sset{(i,\P_i)}_i$ of all alive users and their PCs. Additionally, it holds the graph $G=(\LIVE,E)$ induced by the neighbors of each alive user. The neighbors $\totalsample_i$ of users $i$ is held by user $i$ and all users in its personal committee $\P_i$.

\smallskip
{\bf The users agree on the number of alive users:}
\begin{enumerate}[topsep=0pt]
	\item The server sends to all PCs the number of alive users. Let $n'$ denote this value.
	
	\item Each PC $\P_i$ sends $n'$ to all of its neighbors, \ie to all $\P_j$ where $j\in\totalsample_i$.
	
	\item\label{step:compare_size} A PC aborts if it received a value from one of its neighbors that differs from the value it received from the server.
	
	\item\label{step:check_alive_committee_size} For \FILreps iterations, each personal committee $\P_i$ sends \alive to every $\P_j$ where $j\in\totalsample_i$. At any iteration, if $\P_i$ did not receive the message \alive from any of its neighbors, it aborts.
\end{enumerate}

\smallskip
{\bf The parties run a variant of Feige's protocol:}
\begin{enumerate}[topsep=0pt]
	\item Let $b=\ceil{n'/\secParam}$ denote the number of bins. Each PC $\P_i$ samples a bin $x_i\in[\numbin]$ uniformly at random and sends it to the server.
	
	\item The server:
	\begin{enumerate}[topsep=0pt]
		\item Let $x=\minority\sof{x_1,\ldots,x_{|\cV|}}$.
	
		\item Let $\C=\sset{i\in\cV:x_i=x}$.
		
		\item Send $\C$ and $x$ to all personal committees.
	\end{enumerate}
	
	\item Each personal committee $\P_i$ sends $\C$ and $x$ to each of its neighbors, \ie to every $\P_j$ where $j\in\totalsample_i$.
	
	\item A personal committee $\P_i$ sends $\bot$ to every neighbor $\P_j$ where $j\in\totalsample_i$ and aborts if one of the following holds.
	\begin{itemize}[topsep=0pt]
		\item It received different $\C$ or $x$ from one of its neighbors, or
		
		\item $\sabs{\C}>\secParam$, or
		
		\item $x_i=x$ and $i\notin\C$, or
		
		\item $x_i\ne x$ and $i\in\C$.
	\end{itemize}
	
	\item\label{step:check_alive_committee} For \FILreps iterations, each personal committee $\P_i$ sends \alive to every $\P_j$ where $j\in\totalsample_i$. At any iteration, if $\P_i$ did not receive the message \alive from any of its neighbors, it aborts.
	
	\item Each user outputs $\C$. %
\end{enumerate}
\end{protocol}

Clearly, since we assume \PCSG to occur, no user is in more than $3\secParam$ PCs, thus the running time of all users is polylogarithmic in $n$. The next lemma states that the protocol securely computes $\felectstar$. In \changed{\cref{sec:comm_love}}, we show how implement the $\nxtmsg$ hybrid functionalities, with security under parallel composition. Together with the composition theorem, this proves \cref{thm:committee}. %

\begin{lemma}\label{thm:committee_hybrid}
	Let $\alpha<\changed{1/8}$ be a constant. Then protocol \Committee computes $\felectstar$ with $\alpha n$-security in the $\nxtmsg$-hybrid model, conditioned on the event \PCSG.
\end{lemma}
\begin{proof}
	Fix an adversary $\adv$ corrupting a subset $\corrset\su[n]$ of the users, of size at most $|\corrset|\leq\alpha n$. We separate the proof into two cases, depending on whether the server is honest or corrupted.
	
	\paragraph{Honest server.}
	Since \PCSG is assumed to have occurred, it follows that all PCs honest users are active, and all PCs (whether of malicious or honest users) contain at least $(1-\alpha-\eps/2)\secParam>\changed{7\secParam/8}$ honest users. Therefore, the adversary cannot prevent any of the calls to each $\nxtmsg_i$ functionality to be delayed. Furthermore, at any round of the protocol, the adversary sees at most $(\alpha+\eps/2)\secParam<\changed{\secParam/8}$ shares from each PC. Thus, its view is comprised of only random messages sampled uniformly at random and independently.
	
	\paragraph{Malicious server.}
	We next define the simulator $\Sim_{\adv}$. To simplify the presentation, we do not concern ourselves with the view of malicious users in honest PCs, similarly to the case where the server was honest, they correspond to shares that reveal no information to $\adv$. Furthermore, we assume the server does not block any message sent from one honest PC to another, as this will immediately cause all honest users to abort. The simulator does as follows.
	\begin{enumerate}
		\item{\bf Simulate ``the users agree on the number of alive users'':}
		\begin{enumerate}			
			\item Query the adversary $\adv$ for the number of alive users that it sends to the honest PCs.\footnote{Formally it is given as shares, on for each party. However, it is clear that the simulator can reconstruct the message.} If two of them differ or the number is less than \changed{$7n/8$}, then send \abort to the trusted party $\trustp$, output whatever $\adv$ outputs, and halt.
			
			\item Otherwise, let $n'$ denote the number of alive users as was sent by the server. 
			
			\item Query $\adv$ for the messages that each malicious PC sends to its neighbors (recall that messages sent to other PCs are considered DDoS attacks, and where handled in \cref{sec:Background}). If any of the message differ from $n'$ during the simulation of Step~\ref{step:compare_size}, or differ from \alive during the simulation of Step~\ref{step:check_alive_committee_size}, then send \abort to $\trustp$, output whatever $\adv$ outputs, and halt.
			
			\item Otherwise, send to $\trustp$ the set $\LIVE$ of alive users as given by the event \PCSG (note that it could be the case that $n'\ne|\LIVE|$), and the number of bins $b=\ceil{n'/\secParam}$.
			
			\item The simulator obtains from the trusted party the bins that each alive honest user sampled. That is, it receives $\sset{(i,x_i)}_{i\in\LIVE\cap\honset}$, where $x_i\in[b]$.
		\end{enumerate}
		
		\item{\bf Simulate ``the parties run a variant of Feige's protocol'':}
		\begin{enumerate}
			\item Send the collection of choices $\sset{(i,x_i)}_{i\in\LIVE\cap\honset}$ to $\adv$.
			
			\item The adversary replies with a set $C$ and a bin-index $x$ for every alive honest PC. If one of the following occur, then send \abort to $\trustp$, output whatever $\adv$ outputs, and halt.
			\begin{itemize}
				\item Two PCs received different messages.
				\item $|\C|>\secParam$.
				\item There exists $i\in\LIVE\cap\honset$ such that either $x_i=x$ and $i\notin\C$, or $x_i\ne x$ and $i\in\C$.
			\end{itemize}
		
			\item Query the adversary for the message that the malicious PCs sends during Step~\ref{step:check_alive_committee}. If any of them is not \alive, then send \abort to $\trustp$, output whatever $\adv$ outputs, and halt.
			
			\item Otherwise, send $\C$ to $\trustp$, output whatever $\adv$ outputs, and halt.
		\end{enumerate}
	\end{enumerate}

	We claim that real (hybrid) world is identically distributed as the ideal world. Clearly, the messages the adversary sees in the real world are identically distributed as the messages it receives from the simulator in the ideal world. In particular, it responds with the same messages. Moreover, by the assumption that the event \PCSG occurred, it follows that if any PC does not send the message \alive in the protocol at Steps~\ref{step:check_alive_committee_size} and \ref{step:check_alive_committee}, then all users abort. By the construction of the simulator, this occurs in the real world with exactly the same probability as in the ideal world. Finally, observe the committee $\C$ send by the simulator in the ideal world satisfy the constrains it must abide by, namely, it is of size at most $\comsize$, and there exists $j\in[b]$ for which it holds that $\bin_j\su\C$ and $\C\setminus\bin_j\su\corrset$, where $\bin_j=\sset{i\in\LIVE\cap\honset:x_i=j}$ is the \jth bin.
\end{proof}

\subsection{From a Single Committee to Many Committees}\label{sec:many_committees}

Recall that our final goal is to construct a secure protocol for computing an arbitrary function over the users' inputs. The first step towards constructing such a protocol, is to elect \emph{many committees}, where each will effectively act as an honest party, replacing one user. Once the committees are elected, we let each of them hold an input of one party, and then run a secure protocol with the server (see \cref{sec:proof,sec:combine}). We show how to securely elect $n+1$ committees so that each committee contains sufficiently many honest users, and can thus be treated as an honest party.

We next formally describe the functionality, denoted $\fmany$. Roughly, the first committee sampled by the functionality, is sampled similarly to $\felectstar$, while the other are sampled honestly regardless of whether the server is corrupted or not. The functionality then sends a different committee to each user. Additionally, the users holds the committees in a ``tree-like'' structure, that is, all users in one committee will hold all users in the corresponding children and parent committees.

\paragraph{Ideal world for electing many committees assuming an honest server.}

We next describe the interaction in the ideal world of \fmany assuming the server is honest. Let $\adv$ be an adversary corrupting a subset $\corrset\su[n]$ of the users, and which does \emph{not} corrupt the server.

\begin{description}
	\item[Inputs:] All parties hold the security parameter $1^{\secParam}$ and the number of users (held in unary by the server and in binary by the users). The adversary is given auxiliary input $\aux\in\zos$.
	
	\item[Adversary aborts some malicious users:] The adversary sends to $\trustp$ a set $\LIVE\su[n]$ of alive users, where $\honset\su\LIVE$.
	
	\item[Trusted party samples committees:] The trusted party does the following.
	\begin{enumerate}[topsep=0pt]
		\item Set $b=\ceil{|\LIVE|/\comsize}$ to be the number of bins.
		
		\item For every alive (possibly malicious) user $i\in\LIVE$, sample a bin $x_i\in[b]$ independently and uniformly at random.
		
		\item Set $\C_0$ to be the lightest bin (note that $|\C_0|\leq\secParam$ by the pigeonhole principle).
		
		\item For every $i\in[n]$, sample a committee $\C_i\su\LIVE$ of size $\secParam$ uniformly at random.
		
		\item Send $\sset{\C_i}_{i=0}^n$ to the server, and for every user $i\in[n]$, send $\C_i$ to user $i$, and to all users in $\C_{2i+1}$ and $\C_{2i+2}$ (assuming $2i+1,2i+2\in[n]$) and vice versa.
	\end{enumerate}
	
	\item[Output:] Each honest user and the server output whatever they received from the trusted party, the corrupted parties output nothing, and the adversary outputs some function of its view.
\end{description}

\paragraph{Ideal world for electing many committees assuming a malicious server.}
We next describe the interaction in the ideal world of \fmany assuming the server is corrupted. Let $\adv$ be an adversary corrupting a subset $\corrset\su[n]$ of the users, which also corrupts the server. %

\begin{description}
	\item[Inputs:] All parties hold the security parameter $1^{\secParam}$ and the number of users (held in unary by the server and in binary by the users). The adversary is given auxiliary input $\aux\in\zos$.
	
	\item[Adversary aborts some users and choose number of bins:] The adversary either sends $\abort$ to $\trustp$ or sends it a set $\LIVE\su[n]$ of alive (possibly malicious) users of size at least $\sabs{\LIVE}\geq\changed{7n/8}$, and the number of bins $b$, where $b\geq\ceil{\changed{7n/8\secParam}}$.
	
	\item[Trusted party splits users into bins:] If $\trustp$ received $\abort$, then it sends $\bot$ to all parties and halts. Otherwise, for every alive honest user $i\in\LIVE\cap\honset$, the trusted party samples a bin $x_i\in[b]$ independently and uniformly at random. It then sends $\sset{(i,x_i)}_{i\in\LIVE\cap\honset}$ to the adversary.
	
	\item[The adversary chooses a bin and adds corrupted users:] For every $j\in[b]$, let $\bin_j=\sset{i\in\LIVE\cap\honset:x_i=j}$ denote the set of honest users that sampled the \jth bin. The adversary sends back either $\abort$ or a set $\C_0$ of size at most $\comsize$, such that there exists $j\in[b]$ for which it holds that $\bin_j\su\C_0$ and $\C_0\setminus\bin_j\su\corrset$ (note that such a $j$ always exists if the server is honest, hence, when such a $j$ does not exist then it implies the server must behaved maliciously).
	
	\item[Trusted party samples other committees and sends output:] If $\trustp$ received $\abort$ it sends $\bot$ to all parties and halts. Otherwise, it does the following.
	\begin{enumerate}
		\item For every $i\in[n]$, sample a committee $\C_i\su[n]$ of size $\secParam$ uniformly at random.
		
		\item Send $\sset{\C_i}_{i=1}^n$ to the server.
	\end{enumerate}

	\item[The adversary chooses whether to end to computation or not:] The adversary replies with either \abort or a new set of alive $\LIVE'\su[n]$ of users of size at least $|\LIVE'|\geq\changed{7n/8}$. In the former case, $\trustp$ send $\bot$ to all parties and halt. In the latter case, for every alive user $i\in\LIVE'$, send $\C_i$ to user $i$, and to all (alive) users in $\C_{2i+1}$ and $\C_{2i+2}$ (assuming $2i+1,2i+2\in[n]$) and vice versa.
	
	\item[Output:] Each honest user outputs whatever it received from the trusted party, the corrupted parties output nothing, and the adversary outputs some function of its view.
\end{description}

We prove the following result, asserting there exists a secure protocol for computing \fmany assuming the event \PCSG occurred.
\begin{lemma}\label{lem:many_committees}
	Let $\alpha<\changed{1/8}$ be a constant. Then there exists an $n$-user \LOVE protocol computing $\fmany$, such that conditioned on the event \PCSG, the protocol is $\alpha n$-secure.
\end{lemma}

We now briefly explain the idea for sampling many committees. First, the parties will generate a global common committee $\C_0$ known to all parties, using \cref{protocol:comm}. Then, $\C_0$ constructs a tree of $n+1$ committees, with itself as the root. This is done by having each node (committee) sample two children (\ie two new committees). Doing this for $\log n$ iterations results in the foramtion of  $n+1$ committees. By \cref{lem:committee_vast_hon_maj} and the union bound, with overwhelming probability, all of these committees will have a vast majority of honest users. Hence, they can be effectively viewed as honest parties. Moreover, it is highly unlikely that any user will appear in more then $\polylog(n)$ committees, hence efficiency is maintained.

After the tree of (honest) committees is completed, each committee is assigned to one user, that is, committee $\C_i$ will be assigned to user $i$. 
There are two main difficulties in doing so. The first one is that the committee must prove to the user it was sampled honestly during the creation of the tree. The second issue is that the committee cannot distinguish an aborting user from one that is blocked by the server, hence the input could be lost. 

To overcome to the first issue, we let $\C_0$ sample signature keys, and pass them to all committees in the tree via its children. Thus, each committee can prove honesty by signing a random message. For a user to be able to verify the signed message, $\C_0$ must deliver it the public key. \changed{To do so efficiently, we use the assumption that the event \PCSG occurred. Recall that the event asserts that each user has an associated personal committees replacing it, and furthermore, the PCs agree on a communication graph of diameter $O(\log n/\log\secParam)$, in the sense that each PC knows its set of neighbors.

With this in mind, we let all users in $\C_0$ send the public key to their respective PCs, as well as to the server. Each PC then compares the public key to the one sent to the server. If they are the same, then it passes it to all other PCs along the communication graph. Otherwise, it sends $\bot$ to all its neighbors (without aborting). If there is no value that appears more than (roughly) 3/4 of the time, then the users know the server is malicious. Finally, all PCs send the public key to their associated user (assuming the user is still active). By the small diameter property of the graph (assumed to hold by the \PCSG event), it follows that in $O(\log n/\log\secParam)$ rounds all users, will obtain the public key.}

As for the second issue, we let the committees request a response from its assigned user, and simply aggregate along the tree responds. If too many users did not respond, then they necessarily know the server is corrupted and they all abort (note that this holds true regardless of how many users aborted previously).

We next formally describe the above process. Similarly to protocols \CommSetup and \Committee, the protocol is described in a hybrid world, denoted $\nxtmsg$, where the hybrid functionalities compute the next-message function of each committee and each PC. The $\nxtmsg$ functionalities that correspond to the committees, are defined similarly to those that correspond to PCs, with the difference being the sharing threshold of the sharing scheme used for the inputs and outputs of the participating users. Specifically, the shares are generated in a $(1-\beta-\eps/2)\secParam$-out-of-$\secParam$ Shamir's secret sharing scheme, where $\beta=|\LIVE\cap\corrset|/|\LIVE|$ is the fraction of malicious users among the alive users.\footnote{Formally, the functionality cannot obtain $\LIVE$, and hence it cannot compute $\beta$. However, the fraction of alive users can be estimated by sampling $\secParam$ random users and requesting a response. To simplify the presentation, we assume the functionality receives $\beta$. We stress that the implementation provided in \changed{\cref{sec:comm_love}}, does not assume this.} 
In \cref{sec:comm_love}, we show how to implemented these functionalities. 

If at any point during the computation, the server blocks the communication between two committees (\ie the users in the receiving committee do not hold sufficiently many shares), then the receiving committee aborts. If the parent (global) committee aborts, it sends this information to $\secParam$ random PCs who passes this information through the communication graph. To simplify the presentation of the protocol, we will not concern ourselves with these issues. We are now ready to describe the protocol. It is presented in the $\sset{\nxtmsg,\felectstar}$-hybrid model.\\

\begin{protocol}{\ManyCommittees}\label{protocol:many_comm}
{\bf Common inputs:} All parties hold the security parameter $1^{\secParam}$ and the number of users $n$ (held in binary by the users, and in unary by the server). 

{\bf Event assumption:} We assume that the event \PCSG occurred. Recall that an honest server holds the list $\LIVE=\sset{(i,\P_i)}_i$ of all alive users and their PCs. Additionally, it holds the graph $G=(\LIVE,E)$ induced by the neighbors of each alive user. The neighbors $\totalsample_i$ of users $i$ is held by user $i$ and all users in its personal committee $\P_i$.

\smallskip
{\bf Sample first committee and send public verification key to all users:}
	\begin{enumerate}[topsep=0pt,itemsep=1pt]
		\item\label{step:First_Committee} The parties call $\felectstar$. If they receive $\bot$, the abort. Otherwise, let $\C_0$ denote the sampled committee.
	
		\item $\C_0$ prepares signature keys $(\pk,\sk)\from\SigGen\sof{1^{\secParam}}$. It then sends $\pk$ to the server and to all users $i\in\C_0$.

		\item For every $i\in\C_0$, the server sends $\pk$ to PC $\P_i$.
		
		\item Each PC $\P_i$ sets $\pk_i$ to be an empty array of length $|\C_0|$, indexed with the users in $\C_0$. If $\P_i$ received the same value $\pk$ from user $i$ and from the server, it sets $\pk_i(i)=\pk$. Otherwise, it sets $\pk_i(i)=\bot$. It then sends $\pk_i$ to all its neighbors $\P_j$, \ie where $j\in\totalsample_i$.
		
		\item\label{step:update_keys} For \FILreps iterations, each PC updates its array according to values it received from its neighbors, and sends the updated array to its neighbors.
		
		\item If there exists a value in $\pk_i$ that appears at least $(\changed{3/4}-\eps/2)\secParam$ times (where $\eps=\changed{1/8}-\alpha$), then $\P_i$ sets it to be the public key and sends this value to users $i$. Note that by \cref{lem:committee_vast_hon_maj}, with overwhelming probability, this value will be the public key $\pk$ that was sampled by $\C_0$. To alleviate notations, we will use $\pk$ to denote this value.
	\end{enumerate}
	
\smallskip
{\bf Sample the rest of the committees:}
	\begin{enumerate}[topsep=0pt,itemsep=1pt]
		\item For $i=0$ to $\ceil{\log n}-1$:
		\begin{enumerate}[topsep=0pt,itemsep=1pt]
			\item Committee $\C_i$ samples two more committees $\C_{2i+1},\C_{2i+2}\su[n]$, each of size $\comsize$, independently and uniformly at random, and sends them the secret signature key $\sk$ (held shared among the users in the committee) and its identity $\C_i$.
		
			\item For $b\in\sset{1,2}$ such that $2i+b\in[n]$:
		
			\begin{enumerate}[leftmargin=15pt,rightmargin=10pt,itemsep=1pt,topsep=0pt]
				\item Committee $\C_{2i+b}$ sends to user $2i+b$ the pair $(r_{2i+b},\sigma_{2i+b})$, where $r_{2i+b}\from\zo^{\secParam}$ is a random string and $\sigma_{2i+b}\from\SigSign_{\sk}\sof{r_{2i+b}}$ is its signature.
			
				\item If $\SigVerify_{\pk}\sof{r_{2i+b},\sigma_{2i+b}}=1$, then user $2i+b$ responds to $\C_{2i+b}$ with the message \alive. Otherwise, if it never received a signed message it aborts.
			\end{enumerate} 
		\end{enumerate}
		
		\item\label{step:aggregate} The committees aggregate the total number of \alive message they received. Let $\ell$ denote the final result, held by $\C_0$. 
	
		\item If $\ell<(1-\alpha-\eps/2)n$ then each committee notifies its user to abort.
	
		\item Otherwise, an honest user $i$ outputs $\C_i$, and for every $j\in\sset{0,\ldots,n}$, where $i\in\C_j$, the user outputs $\C_{2j+1}$, $\C_{2j+2}$, and $\C_{\floor{(j-1)/2}}$ as well (assuming the indexes are in $\sset{0,\ldots,n}$). An honest server outputs $\sset{\C_i}_{i=0}^n$.
	\end{enumerate}
\end{protocol}
By \cref{fact:chernoff}, no user will be in more than $3\secParam$ committees except with negligible probability, hence the protocol maintains efficiency. As for security, the intuition is that by \cref{lem:committee_vast_hon_maj} and the union bound, all committees sampled during the execution of the protocol will contain a vast honest majority with overwhelming probability. In particular, the number of honest users is sufficient so that the calls to $\nxtmsg$ will be be executed honestly. Finally, since the adversary cannot forge a signed random message, no user will appear in a maliciously fabricated committee. We next formalize this intuition. %

\begin{lemma}\label{thm:many_committees}
	Let $\alpha<\changed{1/8}$ be a constant. Then protocol \ManyCommittees computes $\fmany$ with $\alpha n$-security in the $\sset{\nxtmsg,\felectstar}$-hybrid model, conditioned on the event \PCSG.
\end{lemma}
\begin{proof}
	Fix an adversary $\adv$ corrupting a subset $\corrset\su[n]$ of the users, of size at most $|\corrset|\leq\alpha n$. The case where the server is honest follows similar reasoning to the proof of \cref{thm:committee_hybrid} and is therefore omitted.
	
	We next assume the adversary corrupts the server. We assume \wlg that the adversary did not send \abort to the functionality \felectstar, as this causes all users to abort. Let $\LIVE\su[n]$ denote the set of alive users the adversary sent to the functionality \felectstar at the start of the protocol (where $|\LIVE|\geq\changed{7n/8}$). Observe that by \cref{lem:committee_vast_hon_maj} and the union bound, it follows that
	$$\abs{\C_i\setminus\corrset}\geq(\beta+\eps/2)\secParam$$
	for all $i\in\sset{0,\ldots,n}$, except with negligible probability, where $\beta=\changed{|([n]\setminus\LIVE)\cup\corrset|/n}$ is the fraction of users that are either not alive or malicious, among the total number of users. In the rest of the proof we condition on this event occurring. 
	
	We next define the simulator $\Sim_{\adv}$. Similarly to the proof of \cref{thm:committee_hybrid}, to simplify the presentation, we do not concern ourselves with the view of malicious users in honest PCs or in committees, as they correspond to random share. Furthermore, we assume the server does not block any message sent from one honest PC to another, as this will immediately cause all honest users to abort. The simulator does the following.
	\begin{enumerate}
		\item{\bf Simulate ``Sample first committee and send public verification key to all users'':}
		\begin{enumerate}
			\item Simulate the interaction between $\adv$ and the functionality \felectstar:
			\begin{enumerate}
				\item Query $\adv$ for the set $\LIVE$ of alive users, and the number of bins $b$. Send these to the trusted party, and obtain $\sset{(i,x_i)}_{i\in\LIVE\cap\honset}$.
			
				\item Send $\sset{(i,x_i)}_{i\in\LIVE\cap\honset}$ to the adversary $\adv$, who responds with a set $\C_0$ of size at most $\secParam$, such that $\bin_j\su\C_0$ and $\C_0\setminus\bin_j\su\corrset$ for some $j\in[b]$, where $\bin_j=\sset{i\in\LIVE\cap\honset:x_i=j}$, is the \jth bin.
			\end{enumerate}
			
			\item Generate signature keys $(\pk,\sk)\from\SigGen(1^{\secParam})$ and send $\pk$ to the adversary.
			
			\item For every honest $i\in\C_0$, the adversary sends a value $\pk'_i$ to the PC $\P_i$.
			
			\item If less than $(\changed{3/4}-\eps/2)\secParam$ of them received $\pk$ from $\adv$, then send \abort to the trusted party $\trustp$, output whatever $\adv$ outputs, and halt (note that simulating the view of $\adv$ during Step~\ref{step:update_keys} is redundant since the adversary can compute it by itself).
		\end{enumerate}
		
		\item{\bf Simulate ``Sample the rest of the committees'':}
		\begin{enumerate}
			\item Send $\C_0$ to the trusted party, who replies with committees $\sset{\C_i}_{i=1}^n$, where $\C_i\su[n]$ is of size $\secParam$ for all $i\in[n]$.
			
			\item For $i=0$ to $\ceil{\log n}-1$:
			\begin{enumerate}
				\item Send $\C_{2i+1}$ and $\C_{2i+2}$ to $\adv$ (and shares of $\sk$).
				
				\item For $b\in\sset{1,2}$ such that $2i+b\in\corrset$ is corrupted, send to $\adv$ the pair $(r_{2i+b},\sigma_{2i+b})$, where $r_{2i+b}\from\zo^{\secParam}$ is a random string and $\sigma_{2i+b}\from\SigSign_{\sk}\sof{r_{2i+b}}$ is its signature.
				
				\item\label{step:block_alive} For $b\in\sset{1,2}$ such that $2i+b\in\honset$ is honest, the adversary can either block the message \alive that user $2i+b$ is suppose to send, or not.\footnote{Observe that the simulator can obtain the set of blocked users from the malicious server, since at the beginning of every round, the server expects to receive the set of pairs $(i,j)$, where user $i$ sends a message to user $j$.}
			\end{enumerate}
			\item If the server blocked at least \changed{$7n/8$} users during Step~\ref{step:block_alive}, then send abort to $\trustp$, output whatever $\adv$ outputs, and halt.
			
			\item Otherwise send the set $\LIVE'\su[n]$ of all users that are either malicious or honest that were not blocked, to $\trustp$, output whatever $\adv$ outputs, and halt.
		\end{enumerate}
	\end{enumerate}
Clearly, the view of $\adv$ in both worlds are identically distributed. In particular, the adversary's replies are identical. We next show that, conditioned on the view of $\adv$ being the same in both worlds, the output of the honest users in both worlds is statistically close. 

Indeed, in the ideal world, the committees are sampled honestly by the trusted party. Conversely, in the real world, they are sampled honestly if and only if all of them contain sufficiently many honest users, which holds due to \cref{lem:committee_vast_hon_maj} except with negligible probability. Additionally, in the ideal world no honest user will output a committee fabricated by the adversary. In the real world, such an event can occur if and only if the adversary forged a signature of a $\secParam$-length random message. By the security of the signature scheme, the adversary will be able to do this with only a negligible probability.
\end{proof}

\changed{\begin{remark}
Recall that when we formally defined the GMPC model in \cref{sec:Background}, we noted the adversary can perform a DDoS attack, even without corrupting the server. We proposed two solutions to this issue, where one of them was a combinatorial solution, that in general could allow the adversary to cause a small fraction of honest users to abort even when the server is honest. 

Observe that given these committees, they parties can overcome this issue, and have the committees represent honest users that aborted during the election of the first committee. To do so, the users that previously aborted, can obtain the public signature keys by requesting it from $\secParam$ randomly sampled PCs and comparing the information. Since the server is honest, the information will always be consistent (since all PC have more than 7/8 fraction honest users in them), hence the users will not abort. Furthermore, the attacker cannot perform a DDoS attack at this stage, since the server knows which users suppose to interact with each other, hence it can block the malicious users in case they attack.
\end{remark}}

\section{Server's Proof of Correctness of  Computation}\label{sec:proof}

We present a protocol allowing a server to compute an arbitrary function over the users' inputs, while proving to the users that the computation was done correctly. For the entire section, we assume that \emph{all users are honest}, while only the server might be corrupted. We further assume the users are connected via a binary tree network, with the root of the tree known to all parties. These assumption are then removed in \cref{sec:combine} by electing many committees using the protocol from \cref{sec:committee}, and letting each committee replace a user.

Before presenting the protocol in \cref{sec:arbitrary},  we present in \cref{sec:VerComp} a protocol that allows the users to verify a certain computation by the server was done properly.

\subsection{Verifying the Computation}\label{sec:VerComp}

We present a protocol that allows a prover to prove a $\operatorname{P}$ statement\footnote{Though we state and prove assuming the language belong to P, the same protocol works for any language in $\operatorname{NP}$, assuming the prover is given the witness to the input.} to $n+1$ (honest) verifiers, so that a cheating prover is caught with overwhelming probability. The prover holds an input $x$ of length $\Tilde{O}(n)$, and each verifier holds a single $\polylog(n)$-sized substring of $x$. We assume the verifiers run in time that is polylogarithmic in the number of users and polynomial in the security parameters. Furthermore, the verifiers are connected to each other via a binary tree network. 

In the construction of the protocol we use \emph{probabilistically
checkable proofs of proximity} (PCPP) \cite{BGH+04a,DR04}. Roughly speaking, a PCPP allows a verifier to be convinced that the input is close to being in a language $L$. In more details, the verifier has an explicit input $x$ and an implicit input $y$ given as an oracle. The verifier accepts with high probability if $y$ is close to some $y'$ such that $(x,y')\in L$. 

We now present the protocol. We first describe its setting. There is a single verifier $\Vc$, $n$ helping parties $\Hc_1,\ldots,\Hc_n$, and a single prover $\Pc$. The prover is connected to every
other party via a secure point-to-point channel, and the verifier and the helpers are connected via a binary tree network. Specifically, the root of the tree is $\Vc$ whose identity is known to all parties, and is connected to both $\Hc_1$ and $\Hc_2$. Additionally, for every $i\in[n]$, helper $\Hc_i$ is connected to $\Hc_{2i+1}$ and $\Hc_{2i+2}$ (assuming $2i+1,2i+2\in[n]$).

The goal of the prover is to prove some $\operatorname{P}$ statement (encoded as $x\in L$ for some language $L\in\operatorname{P}$) to the verifier. Our setup assumption is that the string $x$ is divided among the helpers and the verifier (in addition to the prover knowing $x$ in its entirety). The main difficulty is that the verifier and the helpers are polylogarithmic in $n$, and hence, any one of them cannot even read the entire statement. To overcome this, we let them work together as one verifying unit, where the helpers assist the (designated) verifier in completing this process. 

As a first step,  the prover encodes $x$ using a linear error-correcting code, partitions the codeword, and sends to the verifier and each helper a single $\polylog(n)$-sized substring of the codeword. The verifier and the helpers then check that the encoding was done properly. Since the ECC is linear, this can be done as follows.  The verifier samples a random subset of size $\polylog(n)$ of the rows in the generating matrix of the ECC. Each helper computes its part of the encoding of the input with respect to this subset of rows. 

To complete the verification of the encoding, the helpers sum up their computed values, and send the sum to the verifier. 
 This summation can be performed efficiently, by propagating the summed values along the topology of the tree, where the helper at each node sums the values it receives from its children and  pass it on to its parent. In addition, the verifier gathers the encoding values attributed with the selected subset of rows, directly from the helpers that were assigned this values (by the server). The verifier then compares the two values for each selected row, and accepts the encoding if they all match.

 If the ECC verification is accepted, the prover generates a PCPP proof for the encoding of the language, which is then verified by the verifier. To implement the oracle access of the PCPP's verifier, the protocol's verifier $\Vc$ will use the aid of the helpers. Specifically, since they hold the encoding of the input, $\Vc$ can simply query them via the correct path on the tree topology. 
 
 Intuitively, if the ECC can tolerate more errors than the distance parameter $\delta$ of the PCPP proof, then accepting the proof indicates that the codeword is close to a codeword in the encoding of the language, thus the original input $x$ is in the language.

In order to simplify the presentation, whenever we say the verifier sends a message to some helper or vice versa, this is shorthand for having the helpers pass the message along the tree to its correct destination. Additionally, we will assume that $n+2$ is a power of 2, \ie the binary tree is full. \changed{Finally, we do not assume the prover can block message sent between two other (honest) parties. This will be handled in \cref{sec:combine}.}

In the following, we let $L\su\zos\times\zos$ be a pair language in $\operatorname{P}$, and let $\LECC$ be an efficient $[3(n+1)\log^c n,(n+1)\log^c n, 2(n+1)\log^c+1]$ linear error-correcting code, where $c\in\NN$ is a constant (possibly 0). Additionally, we denote by $G$ be the generator matrix of $\LECC$ and we assume that any entry in $G$ can be computed in time $\polylog(n)$.\footnote{Reed-Solomon encoding is an example of such encoding.} Define $L'=\sset{(x,\LECCEnc(y)):(x,y)\in L}$.  We let $\delta(m)=m/\log^d m$, where $d\in\NN$ is a constant, be the distance parameter of the PCPP proof for the pair language $L'$, as given by \cref{cor:PCPP}. Finally, we let $\VC=(\VCSetup,\VCCom,\VCOpen,\VCVerify)$ be a vector commitment scheme.\\

\begin{protocol}{\ProofOfCorrectness}
{\bf Inputs:} Each helper $\Hc_i$ holds an input $x_i\in\zo^{\log^c n}$, where $c\in\NN$ is some constant, the verifier holds $x_0\in\zo^{\log^c n}$, and the prover $\Pc$ holds $(x_0,\ldots,x_n)$. 
	
{\bf Common inputs:} All parties hold the security parameter $1^{\secParam}$, the number of helpers $n$ (held in binary by the verifier and the helpers, and in unary by the prover), and a public input $\pub\in\zo^{\poly(\secParam)}$. 
	
{\bf Goal:} $\Pc$ wishes to prove the statement $(\pub,(x_0,\ldots,x_n))\in L$. 
	
\smallskip
		
\begin{enumerate}[leftmargin=15pt,rightmargin=10pt,itemsep=1pt,topsep=0pt]
	\item The prover $\Pc$ computes $\codex=\LECCEnc(x_0,\ldots,x_n)$, where we view each $x_i$ as a field element.
		
	\item $\Pc$ partitions $\codex$ into $n+1$ strings $(\codex_0,\ldots,\codex_n)=\codex$, where $\codex_i\in\zo^{3\log^c n}$ for all $i\in\sset{0,\ldots,n}$. It then sends $\codex_0$ to $\Vc$ and for all $i\in[n]$ it sends $\codex_i$ to $\Hc_i$.
		
	\item The verifier and the helpers check that the encoding was done correctly:
	\begin{enumerate}[topsep=0pt,itemsep=1pt]
		\item $\Vc$ samples a set $\cS\su\sset{0,\ldots,n}$ of size $\secParam\cdot\log^{d} 4n$ uniformly at random and sends $\cS$ to all helpers (recall that $d$ is the exponent of the logarithm for the distance parameter $\delta$ of the PCPP proof for $L'$).
		
		\item For every $i\in[n]$ and every $j\in\cS$, helper $\Hc_i$ computes $\codex_j[i]:=G(j,i)\cdot x_i$.
		
		\item For every $i\in\cS\setminus\sset{0}$ helper $\Hc_i$ sends $\codex_i$ to $\Vc$.
		
		\item For $\ell=1$ to $\log(n+2)-1$:
		\begin{enumerate}[topsep=0pt,itemsep=1pt]
			\item For every node $i\in\sset{(n+2)/2^{\ell}-1,\ldots,(n+2)/2^{\ell-1}-2}$ on the $\ellth$ level and every $j\in\cS$, helper $\Hc_i$ computes $$\codes_j[i]:=\codes_j[2i+1]+\codes_j[2i+2]+\codex_j[i],$$
			where $\codes_j[2i+1]=\codes_j[2i+2]=0$ if $2i+1,2i+2>n$.
			
			\item Helper $\Hc_i$ sends $\codes$ to its parent $\Hc_{\floor{(i-1)/2}}$ if $i>2$, and to $\Vc$ otherwise.
		\end{enumerate}
			
		\item For every $j\in\cS$, $\Vc$ computes 
	$\codes_j:=\codes_j[1]+\codes_j[2]+G(j,0)\cdot x_0.$
			
		\item If there exists $j\in\cS$ such that $\codex_j\ne\codes_j$, then $\Vc$ sends $\bot$ to all parties and outputs \rej.
	\end{enumerate}
		
	\item The prover $\Pc$ generates a PCPP proof $\pi$ for the assertion $(\pub,\codex)\in L'$.\footnote{Observe that this step can be done in parallel to the verification of the LECC.}
	
	\item $\Vc$ computes $\pp\from\VCSetup\sof{1^{\secParam},T(\secParam,n),\zo}$, where $T(\secParam,n)$ is the time it takes for $\Pc$ to generate $\pi$, and sends $\pp$ to $\Pc$.
	
	\item $\Pc$ sends to $\Vc$ a commitment to $\pi$.
		
	\item The verifier $\Vc$ executes the PCPP's verifier $\Vc_{\PCPP}$ for $\pi$, with the oracle access being implemented as follows:
	\begin{itemize}[topsep=0pt,itemsep=1pt]
		\item Whenever $\Vc$ queries a bit from $\pi$, it sends the query to the prover who responds with a decommitment of the bit.\footnote{Observe that since the verifier is honest, it will never query $\pi$ on an index larger than $|\pi|\leq T(\secParam,n)$, hence the prover can always decommit.} If the verification of the decommitment rejects, then the verifier sends $\bot$ to all parties and outputs \rej.
			
		\item Whenever $\Vc$ queries a bit from the implicit input $\codex_i$ for some $i\in[n]$, it asks from the corresponding helper $\Hc_i$ to send the bit to it.
	\end{itemize}
		
	\item $\Vc$ outputs whatever $\Vc_{\PCPP}$ outputs.
\end{enumerate}
\end{protocol}

Clearly, the verifier $\Vc$ and helpers $\Hc_1,\ldots,\Hc_n$ are all polynomial in $\log n$ and $\secParam$. The next lemma states the correctness and soundness of the protocol.

\begin{lemma}\label{lem:proof}
	Assume position binding vector commitment scheme exists. If $(\pub,(x_0,\ldots,x_n))\in L$ and the prover is honest, then $\Vc$ always outputs \acc. If $(\pub,(x_0,\ldots,x_n))\notin L$, then for any malicious \ppt (in $\secParam$ and $n$) prover $\Pc^*$, it holds that $\Vc$ outputs \acc with negligible probability.
\end{lemma}
\begin{proof}
	Clearly, if $(\pub,(x_0,\ldots,x_n))\in L$ and the prover is honest, then by the completeness of the PCPP proof, and the fact that for every $j\in\cS$ it holds that
	$$\codes_j=\sum_{i=0}^n G(j,i)\cdot x_i,$$
	it follows that $\Vc$ accepts the proof with probability 1.
	 
	Next, assume that $(\pub,(x_0,\ldots,x_n))\notin L$ and fix a malicious \ppt prover $\Pc^*$. For simplicity we assume that $\Pc^*$ always decommits properly whenever it is required to do so. Let $\codex^*$ be the purported encoding sent to the helpers and the verifier, and let $\cD$ be the set of indexes $i\in\sset{0,\ldots,n}$ such that $\codex^*_i\ne\codex_i$. Recall that the encoding of any given input $x_i$ can be computed efficiently by $\Vc$ and the helpers. Moreover, this is done for all inputs $x_i$, where $i\in\cS$. Therefore, as
	$$\delta\of{\abs{\codex}}=\frac{\abs{\codex}}{\log^d\abs{\codex}}=\frac{3(n+1)\log^c n}{\log^d\of{3(n+1)\log^c n}}\geq\frac{3(n+1)\log^c n}{\log^d 4n},$$
	for all sufficiently large $n$'s, it follows that if $|\cD|\geq\delta$
	then the probability that $\Vc$ accepts the encoding is at most
	\begin{align*}
		\pr{\Vc\text{ accepts encoding}}\leq\pr{\cS\cap\cD=\emptyset}\leq\paren{1-\frac{\abs{\cD}}{n+1}}^{\abs{\cS}}\leq e^{-\frac{\frac{3(n+1)\log^c n}{\log^d 4n}\cdot\secParam\cdot\log^{d} 4n}{n+1}}=e^{-3\secParam\cdot\log^{c}n}.
	\end{align*}
	
	We may now assume that $|\cD|<\delta(|\codex|)$, that is, $\Delta(\codex^*,\codex)<\delta(|\codex|)$. Since the LECC can tolerate $(n+1)\log^c n\geq\delta(|\codex|)$ errors and by the assumption that $(\pub,\codex)\notin L'$, it follows that for any possible codeword $\codey\in\zo^{|\codex^*|}$ where $(\pub,\codey)\in L'$, it holds that $\Delta(\codex^*,\codey)\geq\delta(|\codex^*|)$. Thus, $(\pub,\codex^*)$ is such that $\codex^*$ is $\delta$-far from any codeword $\codey$ satisfying $(\pub,\codey)\in L'$, hence the soundness of the PCPP implies that the verifier accepts with only negligible probability.
\end{proof}

\subsection{Computing Short Output Functionalities}\label{sec:arbitrary}
We are now ready to present a protocol for computing any function with a short output, assuming that all users are honest. Let us first describe the setting. Similarly to the previous section, there are $n+1$ users $\hatPc_0,\ldots,\hatPc_n$ connected via a binary tree network, and a server that is connected to every party. The root of the tree is $\hatPc_0$ whose identity is known to all parties. The parties wishe to compute a function $f:\sparen{\zo^m}^{n+1}\mapsto\zo^{m\cdot(n+1)}$, where $m=m(n,\secParam)=\poly(\log n,\secParam)$, over the inputs of the users such that only the server obtains the output.

Given \ProofOfCorrectness the idea is rather simple: The users first encrypt their inputs using a fully homomorphic encryption scheme. Then, the server homomorphically evaluates the function over the encrypted inputs, and proves to the users it did so honestly using \ProofOfCorrectness. If the proof is accepted, the users decrypt the output. \changed{Finally, since we aim to handle the case where there are $n$ users and the server has an input, we will assume that $\hatPc_0$ has no input and the server sends to it an input at the start of the protocol.} \changed{Similarly to \cref{sec:VerComp}, in order to simplify the presentation in this section, we do not assume the server can block message sent between two other (honest) users. This will be handled in \cref{sec:combine}.}

In the following, write the function $f$ as $f(x)=(f_0(x),\ldots,f_n(x))$, where $f_i(x)\in\zo^m$ is the \ith block of the output. Additionally, we let $\FHE=(\FHEGen,\FHEEnc,\FHEDec,\FHEEval)$ be a fully homomorphic encryption scheme and let $s=s(\secParam)$ be the number of random bits used in $\FHEEval$. Finally, let $\PRG:\zo^{n+1}\mapsto\zo^{s\cdot(n+1)}$ be a pseudorandom generator, and write it as $\PRG(r)=(\PRG_0(r),\ldots,\PRG_n(r))$ where  $\PRG_i:\zo^{n+1}\mapsto\zo^{s}$ for every $i\in\sset{0,\ldots,n}$.

\begin{protocol}{\GenProto}\label{protocol:general_func}

\noindent {\bf Inputs:} Each user $\hatPc_i$, where $i\in[n]$, holds an input $x_i$. The server holds $x_0$. 

\noindent {\bf Common inputs:} All parties hold the security parameter $1^{\secParam}$ and the number of users $n+1$ (held in binary by the users and in unary by the server).

	\begin{enumerate}[leftmargin=15pt,rightmargin=10pt,itemsep=1pt,topsep=0pt]
		\item The server sends $x_0$ to $\hatPc_0$.
		
		\item $\hatPc_0$ generates FHE keys $(\pk,\sk)\from\FHEGen\sof{1^{\secParam}}$, sends $(\pk,\sk)$ to all other users, and sends the public key $\pk$ to the server.
		
		\item Each user $\hatPc_i$, where $i\in\sset{0,\ldots,n}$, encrypts its input $x_i$ to obtain $\encin_i\from\FHEEnc_{\pk}\sof{x_i}$, sample a uniform random bit $r_i\from\zo$, and sends $\encin_i$ and $r_i$ to the server.
		
		\item The server homomorphically evaluates all functions $f_i$, where $i\in\sset{0,\ldots,n}$, over the encrypted inputs using the randomness provided by the users, and obtains the encrypted outputs $$\encout_i=\FHEEval_{\pk}\sof{f_i,\encin_0,\ldots,\encin_n;\PRG_i(r_0,\ldots,r_n)}.$$

		\item The server sends $\encout_i$ to $\hatPc_i$, for every $i\in\sset{0,\ldots,n}$.

		\item The parties execute protocol \ProofOfCorrectness, where the server takes the role of the prover $\Pc$, user $\hatPc_0$ takes the role of the verifier $\Vc$, and for every $i\in[n]$ user $\hatPc_i$ takes the role of the helper $\Hc_i$, for the pair language $L$ defined as follow:

\begin{equation*}
	L=\left\{
	\paren{\pk,\paren{z_0,\ldots,z_n}}:\forall i\in\sset{0,\ldots,n}\quad
	\begin{aligned}
		z_i&=\paren{r_i,\encin_i,\encout_i}\text{, where}\\
		\encout_{i}&=\FHEEval_{\pk}\of{f_i,\encin_0,\ldots,\encin_n;\PRG_i(r_0,\ldots,r_n)}%
	\end{aligned}
	\right\}
\end{equation*}
Observe that $L\in\operatorname{P}$, hence the protocol can be executed efficiently.

	\item $\hatPc_0$ sends all parties the output it obtained from the computation.
		
	\item If the output is \rej then all users abort.
		
	\item\label{step:decrypt_output} Otherwise, each user $\hatPc_i$ decrypts $\encout_i$ to obtain $y_i\from\FHEDec_{\sk}\sof{\encout_i}$, and sends the decrypted values back to the server. %
		
	\item The server outputs $(y_0,\ldots,y_n)$.
	\end{enumerate}
\end{protocol}

\begin{lemma}\label{lem:general_only_mal_server}
	Let $m=m(\secParam,n)=\poly(\log n,\secParam)$ and let $f:\sparen{\zo^m}^{n+1}\mapsto\zo^{m\cdot(n+1)}$ be a \LOVE functionality. Assume the existence of pseudorandom generators, position binding vector commitment schemes, and $n$-secure fully homomorphic encryption schemes. Then \GenProto computes $f$ with security against any \ppt adversary corrupting the server.
\end{lemma}

\begin{proof}
	Correctness clearly holds. Assume that the server is corrupted by an adversary $\adv$. Its simulator $\Sim_{\adv}$ works as follows.
	
	\begin{enumerate}
		\item Query $\adv$ to obtain the input it sends to $\hatPc_0$, and send it to the trusted party.
		
		\item Receive an output $(y_0,\ldots,y_n)$ from the trusted party (recall that the server has no input and all users are honest).
		
		\item Generate FHE keys $(\pk,\sk)\from\FHEGen\sof{1^{\secParam}}$.
		
		\item Compute $n+1$ dummy encryptions, sample $n+1$ random bits $r_0,\ldots,r_n$, and send them and the public key $\pk$ to $\adv$.
		
		\item The adversary $\adv$ replies with purported encrypted values $(\encout_0,\ldots,\encout_n)$.
		
		\item Simulate \ProofOfCorrectness:
		\begin{enumerate}
			\item The adversary sends a purported codeword $\codex_i$ for each $i\in\sset{0,\ldots,n}$.
			
			\item Perform the same error correction verification as done in the real world. If the verification fails send $\bot$ to $\adv$, output whatever $\adv$ outputs, and halt.
			
			\item Otherwise, execute the PCPP verifier $\Vc_{\PCPP}$, while querying $\adv$ whenever $\Vc_{\PCPP}$ queries the proof.
			
			\item If $\Vc_{\PCPP}$ rejects the proof, send $\bot$ to $\adv$, output whatever $\adv$ outputs, and halt.
		\end{enumerate}
	
		\item If there was no abort during the simulation of \ProofOfCorrectness, send to $\adv$ the output $(y_0,\ldots,y_n)$ as given by the trusted party, output whatever $\adv$ outputs, and halt.
	\end{enumerate}

By the semantic security of the FHE scheme, it follows that the adversary's view in the real world is indistinguishable from the view generated in the ideal world. In particular, $\adv$'s replies are indistinguishable. Next, by \cref{lem:proof}, except with negligible probability, either all users abort and $\adv$ does not obtain the output in both worlds, or no user aborts and $\adv$ receives an output in both worlds. By the definition of the language $L$, the output in the real world are the decryptions of
$$\encout_i=\FHEEval_{\pk}\of{f_i,\encin_0,\ldots,\encin_n;\PRG_i(r_0,\ldots,r_n)}.$$
Observe that if the distribution of the above $(\encout_0,\ldots,\encout_n)$ can be distinguished from the encryptions 
$$\paren{\FHEEnc_{\pk}\of{y_i}}_{i=0}^n$$
in the ideal world, then $\FHEEval$ can be used to break the security of the PRG $\PRG$. Thus, the outputs of $\adv$ in both worlds are indistinguishable.
\end{proof}

\section{Putting It All Together}\label{sec:combine}

In this section, we combine the committee-tree election protocol from \cref{sec:committee} and the protocol for proving correct computation from \cref{sec:proof}, and present a protocol for securely computing any arbitrary short-output functionality assuming at most $\alpha$-fraction of the users are corrupted, where $\alpha<\changed{1/8}$ is a constant. Formally, we assume the parties are given ideal access to the functionality $\fmany$ for sampling the tree of committees. This assumption can be removed by implementing the functionality using \cref{protocol:many_comm}, however, this comes at the cost of conditioning on the event \PCSG defined in \cref{sec:setup}. 

\begin{theorem}\label{thm:main}
	Let $\alpha<\changed{1/8}$, let $m=m(\secParam,n)=\poly(\log n,\secParam)$, and let $f:\sparen{\zo^m}^{n+1}\mapsto\zo^{m\cdot(n+1)}$ be a \LOVE functionality. Assume the existence of pseudorandom generators and $n$-secure fully homomorphic encryption schemes. Then there exists an $\alpha n$-secure $n$-user \LOVE protocol computing $f$ in the $\fmany$-hybrid model.
\end{theorem}

The idea is to have the parties generate the committees in a tree-like structure by calling $\fmany$, and let each committee $\C_i$ sampled by the functionality emulate a single user $\hatPc_i$ from protocol \cref{protocol:general_func}. This result in a secure protocol computing any \LOVE functionality in the hybrid model. Similarly to the definitions used in \cref{sec:many_committees}, we formalize this by letting the parties have access to the ``next-message'' function of each committee. In \changed{\cref{sec:comm_love}} we show to implement these.

We first define the ``next-message'' hybrid functionality of the \ith committee $\C_i$, denoted $\nxtmsg_i$. The inputs of the users to each call of $\nxtmsg_i$ are a $(1-\beta-\eps/2)\secParam$-out-of-$\secParam$ Shamir's secret sharing scheme of the view of $\C_i$, where $\beta$ is the fraction of malicious users among the alive users. If at least $(1-\beta-\eps/2)\secParam$ inputs are provided to the functionality, then it proceeds to compute the next-message as specified by the protocol, and sharing the output among the users in $\C_i$. If less then $(1-\beta-\eps/2)\secParam$ inputs are provided (\eg the user is corrupted or a corrupted server blocked an honest user), the functionality sends $\bot$ to all of its users. If a user obtain $\bot$ as the output from some $\nxtmsg_i$, then it aborts.

We abuse notions and describe the protocol as if each committee is a single party. Formally, whenever we say that the server sends a message to $\C_i$, it means that it (supposedly) shares it among the users of $\C_i$. Similarly, when we say that $\C_i$ sends a message to the server, it means that every user in $\C_i$ sends its share of the message to the server. Finally, whenever a committee $\C_i$ sends a message $\mathsf{msg}$ to another committee $\C_j$, it means that each user in $\C_i$ shares its share of $\mathsf{msg}$ among the users in $\C_j$. In the following, we let $\nxtmsg=\sset{\nxtmsg_i}_{i=0}^n$. 

\begin{protocol}{\MainProto}\label{protocol:main}
	{\bf Inputs:} Each user $i\in[n]$ holds an input $x_i$. The server holds $x_0$. 
	
	{\bf Common inputs:} All parties hold the security parameter $1^{\secParam}$ and the number of users $n$ (held in binary by the users, and in unary by the server).

	\smallskip
	
	\begin{enumerate}[leftmargin=15pt,rightmargin=10pt,itemsep=1pt,topsep=0pt]
		\item The parties call \fmany to elect $n+1$ committees $\sset{\C_i}_{i=0}^n$. Recall that each user $i$ holds the committee $\C_i$, and all committee are held in a tree structure, \ie the users in $\C_i$ holds $\C_i$, $\C_{2i+1}$, $\C_{2i+2}$, and $\C_{\floor{(i-1)/2}}$ (assuming the indexes are in $\sset{0,\ldots,n}$).
		
		\item Each user $i$ sends its input $x_i$ to $\C_i$ (held shared). If user $i$ was blocked, then $x_i$ is replaced with a default value.\footnote{Formally, $x_i$ is held shared by the users in $\C_i$, and is supposed to be reconstructed by $\nxtmsg_i$. If the functionality cannot reconstruct the input, then instead of aborting, it replaces it with a default value.}
		
		\item The committees aggregate the total number of inputs they received. If it is less than $(1-\alpha)n$, then they abort (note that with overwhelming probability every user belong to at least one committee, hence all users abort).
		
		\item The parties emulate an execution of \GenProto, with each committee $\C_i$ taking the role of user $\hatPc_i$. Specifically, whenever $\hatPc_i$ supposes to compute a message, the users in $\C_i$ call $\nxtmsg_i$ to obtain the shares of this message.
		\begin{itemize}[topsep=0pt]
			\item If at any point during the emulation, the functionality $\nxtmsg_i$ outputs $\bot$ to the users in $\C_i$, then they all abort.
		\end{itemize}
		
		\item The server outputs whatever it receives from the emulation of \GenProto.
	\end{enumerate}
\end{protocol}

The following lemma asserts the security of \MainProto. Combined with the secure implementation of $\nxtmsg$ we will present later in \changed{\cref{sec:comm_love}} and the composition theorem, this proves \cref{thm:main}.

\begin{lemma}\label{thm:general_hybrid}
Let $\alpha<\changed{1/8}$, let $m=m(\secParam,n)=\poly(\log n,\secParam)$, and let $f:\sparen{\zo^m}^{n+1}\mapsto\zo^{m\cdot(n+1)}$ be a \LOVE functionality. Assume the existence of pseudorandom generators and $n$-secure fully homomorphic encryption schemes. Then \MainProto is an $\alpha n$-secure $n$-user \LOVE protocol computing $f$ in the $\sset{\fmany,\nxtmsg}$-hybrid world. %
\end{lemma}

\begin{proof}
	Fix an adversary $\adv$ corrupting a subset $\corrset\su[n]$ of the users, of size at most $|\corrset|\leq\alpha n$. Observe that if the server is honest, then except with negligible probability, all committees will contain sufficiently many honest users, so that the calls to $\nxtmsg$ will not output $\bot$. Thus, no honest user will ever abort, and the server obtains the output.
	
	Assume that the server is corrupted. Intuitively, security holds, since the view of $\adv$ when emulating \GenProto is, up to shares that reveal no information, exactly the same as in a real execution of \GenProto (with honest parties and no blocking of messages). We next define the simulator $\Sim_{\adv}$. To simplify the presentation, we will not concern ourselves with the shares the adversary receives throughout the execution of the protocol.
	\begin{enumerate}
		\item Query $\adv$ for its inputs to the functionality $\fmany$. If it sent \abort at any point, then all users abort and the protocol halts. In this case, the adversary does no receive any more messages. Output whatever $\adv$ outputs, and halt.
		
		\item Otherwise, the server blocks a set of users $\block\su[n]$ from sending their inputs to $\C_i$.\footnote{Observe that the simulator can obtain this set from the malicious server, since at the beginning of every round, the server expects to receive the set of pairs $(i,j)$, where user $i$ sends a message to user $j$.}
		
		\item If the server blocks more than $\alpha n$ users from providing their inputs to their committees, then the protocol halts and the adversary does no receive any more messages. The simulator outputs whatever $\adv$ outputs, and halt.
		
		\item Otherwise, query $\adv$ for the inputs of the malicious users that it sends to the committees (formally, $\adv$ sends it in the form of shares).
		
		\item Execute the simulator $\Sim'_{\adv}$ for \cref{protocol:general_func} with the following modification.
		\begin{itemize}[topsep=0pt]
			\item Suppose the adversary blocks too many users in some committee $\C_i$ from sending a message to another committee $\C_j$ (hence $\nxtmsg_j$ outputs $\bot$), \emph{before} the parties decrypt the output for the server (see Step~\ref{step:decrypt_output} of \cref{protocol:general_func}). In this case, \MainProto aborts. Then the adversary does no receive any more messages. and we let the simulator output whatever $\adv$ outputs, and halt.
		\end{itemize}
		
		\item Otherwise, send to the trusted party $\trustp$ whatever $\Sim'_{\adv}$ sent alongside the set $\block$ of blocked users. Then, send the output to $\adv$, output whatever it outputs, and halt (recall that this is exactly what $\Sim'_{\adv}$ does in case the adversary did not cause an abort).
	\end{enumerate}

Since the view of $\adv$ when emulating \cref{protocol:general_func} is the same as a real execution (up to shares), it follows that the view generated by $\Sim'_{\adv}$ is indistinguishable. Therefore, the probability the adversary causes all honest users to abort in the real world, is indistinguishable from the ideal world. Finally, conditioned on the protocol not aborting, the output received by the adversary in the ideal world, is the same output it receives when simulating an execution of \cref{protocol:general_func}. Therefore, as shown in \cref{lem:general_only_mal_server}, the output of simulator is indistinguishable from the output of $\adv$ in the real world.
\end{proof}

\subsection{Secure Implementation of The Next-Message Functions}\label{sec:comm_love}

We explain how to implement the next-message functionality $\nxtmsg$ used in our protocols to compute the messages generated for each personal committee or each (global) committee. \changed{Specifically, the implementation will be secure under {\em parallel composition}. The implementation for both cases is done using the same protocol. \changed{Essentially, our protocol is a slight variation of the protocol given in Goldwasser and Lindell~\cite[Corollary 2]{GL03}}. % 

\begin{remark}
We view every committee as a different ``computation unit''. In particular, if an honest user aborts from one committee, it does {\em not} automatically abort from other committees it participates in. This is important as otherwise a malicious user $i$ might affect committees it does not participate in, by causing an honest user $j$ to abort from {\em another} committee which both $i$ and $j$ participate in.
\end{remark}

We first show how to implement a broadcast channel. Then, we will show how the parties can securely compute any functionality, given the broadcast channel. We stress that these computations are always performed by subsets of users that are of size that is poly-logarithmic in the number of users. Hence, the protocols implementing these computations may be of polynomial complexities in the number of participants. In the following we fix positive constants $\alpha<\changed{1/8}$ and $\eps=1/8-\alpha$.
}

\paragraph{Computing broadcast.}
We present a deterministic protocol for computing the broadcast functionality. % 
First, the sender sends its message to the server, which in turn passes it on to everyone else. Second, all users exchange with each other what they received from the server. A user aborts if less than $1-\alpha-\eps/2$ of the messages it received are the same as what it received from the server. In particular, the sender verifies that at least $1-\alpha-\eps/2$ of the users sent it  its input. Finally, the users check that at most $\alpha+\eps/2$ of them have aborted, and abort if this is not the case. 

Observe that, although when the server is malicious there are roughly $2\alpha$ corrupted users in the (global) committees, when it is honest there are roughly $\alpha$. Thus, any user seeing inconsistencies with more than $\alpha$ users can abort. Now, clearly, if the server is honest then all honest users agree on the output (even if the sender is corrupted). On the other hand, if the server is malicious, then any user that disagrees with at least $\alpha+\eps/2$ of the users aborts. Moreover, if more than $\alpha+\eps/2$ of the users abort at the second step, then everyone will abort at the last step. Finally, as the protocol admits perfect security, by the work of Kushilevitz, Lindell, and Rabin~\cite{KLR10} it follows that it is UC-secure and thus can be securely run in parallel.

Observe that in the above protocol, a malicious server can decide which users to block, depending on the messages it received. Furthermore, it can change the message at the cost of having the sender abort. As we show below, these two issues do not affect the security of the overall protocol.

\paragraph{Computing any functionality.}
We next explain how the user in each PC and committee can securely compute any functionality. The security that the protocol admits is such that if more than $\alpha+\eps/2$ honest users are blocked, then the protocol is secure-with-abort (\ie the adversary can abort the computation after obtaining the output). Otherwise, the protocol admits guaranteed output delivery.  Specifically, the computation will be secure under parallel composition. This will allow the \changed{committees} and PCs of honest users, to effectively be viewed as honest parties (recall that unlike honest users, if the server blocks a PC \changed{or a committee} then all users abort).

We let the users emulate an execution of the protocol of \citet{RB89}. The protocol admits information theoretic security, assuming secure point-to-point channels between every pair of users, an honest majority, and the availability of a broadcast channel. We show that the implementation of broadcast presented above is sufficient for the users to emulate the protocol of \cite{RB89}. Consider the case a user broadcasted a message using the above protocol at some round. %
Then, even though the server can block the users based on the value of the message to be sent, it could have done so at the beginning of the next round. Moreover, the server can only block at most $(\alpha+\eps/2)$-fraction of the honest users (viewed as aborting, malicious parties, with security guaranteed to hold by the adaptive security of \cite{RB89}), in addition to having at most $(2\alpha+\eps/2)$-fraction corrupted users (with overwhelming probability). Furthermore, changing the sender's message each computation of broadcast causes it to abort, hence it can be viewed as one of the aborting malicious parties. Thus, the security guarantees from the broadcast protocol suffices. 

\changed{Next, observe that the users \emph{do not} have secure channels, as a malicious server can block messages between any pair of users. To overcome this, whenever a party does not receive a message from another party, it uses the broadcast protocol to send an accusation as well as its identity. Then, all other users label \emph{both} users as malicious (even if one of them is honest). Note that a malicious server can abort $2\alpha+\eps/2$ honest users by sending false accusations, however, aborting more will cause all users in the committee to abort with overwhelming probability.

Finally, in order to maintain an honest majority, recall that in every (global) committee, a malicious server can have $2\alpha+\eps/2$ malicious users. Since it can abort an additional $2\alpha+\eps/2$, we require that $4\alpha+\eps<1/2$, which indeed holds since $\alpha<\changed{1/8}$.}

We stress that even if the server blocks the communication with different users for each user, the protocol is still secure due to the use of broadcast (recall that the users check that sufficiently many users are active at the end of the broadcast protocol). Furthermore, even though the server can block users adaptively, the protocol remains secure as the protocol of Rabin and Ben-Or~\cite{RB89} is secure against adaptive adversaries.\footnote{\citet{RB89} referred to them as dynamic adversaries.} Finally, as shown by Canetti~\cite{Canetti01} and Kushilevitz el al.~\cite{KLR10}, the protocol of Rabin and Ben-Or~\cite{RB89} is secure under concurrent general composition,\footnote{Formally, for security with concurrent general composition to hold, when a user obtains its input to the protocol it first broadcasts a message indicating it is ready. The actual protocol of \citet*{RB89} starts after all users broadcast the message.} thus the parties can run it in parallel while maintaining overall security.

\changed{
\begin{remark}
We observe that the constant 1/8 can be improved assuming an honest server is allowed to block messages. This is because the server knows whether the sender of the message actually sent it or not, and block it from every other user in case it did not send any message. Therefore, unlike in the proposed solution, a corrupted server can only block $\alpha$ more honest users (in addition to the $2\alpha$ corrupted users there are in each committee). Thus, for honest majority it is required to have $\alpha<1/6$.
\end{remark}
}

\subsection{Computing Shallow and Small Circuits}\label{sec:log_depth}
In this section we sketch a simple protocol for computing any circuit of size $\tilde{O}(n)$, and polylogarithmic depth, fan-in, and fan-out. Unlike the general protocol, here we do not need to assume the existence of a fully homomorphic encryption scheme. Two important example of functionalities computable by such circuits are summation and sorting.

The idea is the following: the parties will use protocol \cref{protocol:many_comm} to generate $\tilde{O}$ committees in a ``tree-like manner'', where each committee corresponds to either a single gate, an input wire, or an output wire in the circuit. Then, using the help of the server, each committee receives the list of committees to which it should send the message according to the circuit. Since with overwhelming probability, each input and output wire, and each gate, will have a committee associated with it, where all committees contain sufficiently many honest users, the committees can perform this computation securely. % 

To ensure the server doesn't cheat, and connect committees to incorrect ones, or to fabricated ones, we let each committee sign its identity as well as its associated wire/gate in the circuit (recall that in \cref{protocol:many_comm}, all committees obtain a secret signature key, while all parties obtain the corresponding public key). Then, when the server sends a list of committees, each committee verifies that the signature is valid, and that it is consistent with the circuit.

Finally, as the size of the circuit is $\tilde{O}(n)$, it follows that each user appears in a polylogarithmic number of committees, except with negligible probability. In addition, the complexity of each committee is polylogarithmic since the fan-in, fan-out, and the depth of the circuit are polylogarithmic.

For a fixed circuit $C$, we denote the functionality for sampling the committees by $\fcircuit$.
Formally, we have the following result.

\begin{lemma}
    Let $\alpha<\changed{1/8}$ and let $C$ be a circuit of size $n\cdot\log^c n$, depth at most $\log^c n$, and fan-in and fan-out at most $\log^c n$. Assume the existence of $n$-secure signature schemes and position binding vector commitment schemes. Then there exists an $n$-user \LOVE protocol computing $\fcircuit$, such that conditioned on the event \PCSG, the protocol is $\alpha n$-secure.
\end{lemma}

\begin{theorem}\label{thm:logDepth}
	Let $\alpha<\changed{1/8}$, let $m=m(\secParam,n)=\poly(\log n,\secParam)$ let $c>0$, and let $f:\sparen{\zo^m}^{n+1}\mapsto\zo^m$ be a \LOVE functionality. Assume that $f$ can be computed by a circuit $C$ of size $n\cdot\log^c n$, depth at most $\log^c n$, and fan-in and fan-out at most $\log^c n$. Then there exists an $\alpha n$-secure $n$-user \LOVE protocol computing $f$ in the $\fcircuit$-hybrid world.
\end{theorem}

\subsubsection{Application to the Shuffle Model of Differential Privacy}\label{sec:shuffle}
Differential privacy~\cite{DMNS06} is a rigorous definition for privacy. It typically involves a server or a data curator that aggregates the data of $n$ individuals (or users) in a way that guarantees that the outcome of the computation does not leak too much information on the data of any single individual. Most of the work on differential privacy focuses on the case where the server is fully trusted with the raw data, and privacy is only provided w.r.t.\ an observer that sees the outcome of the computation. In recent years, however, there has been a growing interest in studying models where the server is untrusted. The first model in this vein, called the {\em local model}~\cite{DMNS06,KLNRS08}, requires each user to randomize its input before sending it to the untrusted server, who aggregates all the noisy reports. The amount of randomization needs to be sufficiently high such that the server learns almost no information about the data of any single user, and hence, this model generally suffers from significant loss of accuracy.\footnote{We stress that the server can still learn global statistics about the population.} 

The {\em shuffle model of differential privacy}~\cite{DMNS06,IKOS06,BittauEMMRLRKTS17,ErlingssonFMRTT19,CheuSUZZ19,BalleBGN19} aims to alleviate the loss of accuracy of the local model while maintaining its trust guarantees. This is achieved by augmenting the local model with a {\em shuffle functionality} that takes messages from the users and delivers them to the server after randomly permuting them. The fruitful line of work on the shuffle model shows that having access to such a shuffle functionality can lead to dramatic accuracy improvements over the plain local model of differential privacy. This naturally raises the question of efficiently implementing such a shuffle functionality. Recently, Bell at al.~\cite{BellBGL020} presented a secure implementation of the shuffle functionality with constant rounds of communication. However, the runtime and the communication complexity of the users in their protocol is linear in $n$ (the number of users), and reducing users' complexities to polylogarithmic in $n$ was stated as an open question. As we next explain, \cref{thm:logDepth} resolves this question as an important special case.

Observe that in order to implement the shuffle functionality, it suffices to implement a sorting functionality, in which the server obtains all the messages in an ascending order. (The reason is that the server can then permute the messages on its own.) Recall that sorting can be done using a sorting networks of depth $O(\log n)$ (and hence size $O(n\log n)$) using  the famed {\em AKS network}~\cite{AKS83}. Instantiating \cref{thm:logDepth} with the AKS network, denoted $\AKS$, we obtain the following theorem.

\begin{theorem}\label{thm:shuffleStatement}
Let $\shuffle:\sparen{\zo^m}^{n}\mapsto\zo^{m\cdot n}$ be the shuffle \LOVE functionality, whose output is the input of the users in an order chosen uniformly at random. Then the following holds. 
Let $\alpha<\changed{1/8}$, let $m=m(\secParam,n)=\poly(\log n,\secParam)$ let $c>0$. Then there exists an $\alpha n$-secure $n$-user \LOVE protocol computing $\shuffle$ in the $\fAKS$-hybrid world.
\end{theorem}

\cref{thm:shuffleStatement} shows that the shuffle functionality can be securely computed in the GMPC model with polylogarithmic runtime and communication complexity for the users, which resolves an open question of \cite{BellBGL020}. Note that this construction does not assume FHE. 

Alternatively, \cref{thm:shuffleStatement} can be proved using {\em bitonic sorting} \cite{Batcher68} instead of the AKS sorting network. While the AKS network has better asymptotics ($O(\log n)$ depth and $O(n \log n)$ size vs $O(\log^2 n)$ depth and  $O(n \log^2 n)$ size), it is far from being practical
(see Goodrich~\cite{Goodrich14} for an estimation of the constants in various variants and alternatives and Goodrich~\cite{Goodrich11} for  a randomized oblivious sort), and bitonic sorting is  much more efficient for relevant sizes.

\begin{remark}
For simplicity, in Theorem~\ref{thm:shuffleStatement} we considered a shuffle functionality that takes a single input from each user. Our techniques also apply to a setting where every user submits $\ell=\polylog(n)$ messages to the shuffle functionality, which then deliverers all $n\ell$ messages to the server in a random order.
\end{remark}

\section{Applications: Secure Protocols in the Star Topology}\label{sec:star}

In this section, we explain how to compile a protocol in the \LOVE model to the star network assuming a special form of PKI. \changed{We first formally define the model in the star network and define the PKI. We then show how to compile any protocol in the GMPC protocol to the star model.} % 
The parties are connected via a star network with the server in the middle, \ie each user is connected only to the server. Thus, the server has full control over communication traffic and can block any message it wants. All parties are given the security parameter $1^{\secParam}$, and the number of users $n$, held in binary by the users and in unary by the server.

\subsubsection*{The \texorpdfstring{$\board$}{Board} Functionality}

To emulate users wishing to securely send a message to another user, we assume the parties are given a special form of PKI. Observe that since the users are assumed to be polylogarithmic in the number of users $n$, we cannot simply give each of them the public keys of all other users. Instead, the parties are given a special \emph{board functionality}, denoted $\board$. This functionality holds the public keys (for both encryption and signature schemes) of all users, and it allows any user to request the public keys of any other user. \changed{This can be motivated by allowing the parties to access a third-party forum, whereupon registration the keys are sampled by the website.}

Intuitively, having access to the $\board$ functionality allows any user to send a message to any user of its choice via the server, such that it cannot learn the contents of the message nor modify it. In the following, we let $\Sig=(\SigGen,\SigSign,\SigVerify)$ be a signature scheme, and let $\PKE=(\PKEGen,\PKEEnc,\PKEDec)$ be a \changed{non-malleable} public key encryption scheme.\footnote{Roughly, in a non-malleable encryption scheme no computationally bounded adversary that is given an encryption of some message $m$, can generate an encryption of a different correlated message $m'$.}

\changed{Roughly, the \board functionality allows the following. Any user can register \emph{exactly once}, and obtain the public and secret keys for both the encryption scheme and the signature scheme. The functionality then stores the public keys in its database. Additionally, every user can request the functionality to send it the public keys of any other user that appears in the database.

We stress that the \board functionality \emph{never} interacts with the server. In particular, this means that it cannot invent new users or public keys, nor can it prevent any user from obtaining the public keys of another user.  We now formally describe the functionality.}

\begin{description}
	\item[Input:] The functionality holds the security parameter $1^{\secParam}$ and the number of users $1^n$.
	
	\item[Initialize:] Set up a new database $\cD$ as an array of length $n$.
	
	\item[Register:] Whenever user $i$ wishes to register, it samples keys $(\pk^{i}_{\Sig},\sk^{i}_{\Sig})\from\SigGen(1^{\secParam})$ and $(\pk^{i}_{\PKE},\sk^{i}_{\PKE})\from\PKEGen(1^{\secParam})$. It then sends the message $(\register,\pk^i_{\PKE},\pk^i_{\Sig},i)$ to the functionality. Upon receiving this message, verify that $i$ did not register previously. If this is not the case, then ignore the message. Otherwise, % 
	add $(\pk^{i}_{\Sig},\pk^{i}_{\PKE})$ to $\cD$ at position $i$.
	
	\item[Get Keys:] Upon receiving a message $(\get\text{:}\;i,\fromuser\text{:}\;j)$, verify that $i$ is in $\cD$. If not, then ignore the message. Otherwise, send to user $j$ the public keys $(\pk^{i}_{\Sig},\pk^{i}_{\PKE})$ of user $i$.
\end{description}

\subsubsection*{Compiling Secure GMPC Protocols to the Star Model}
Given access to the \board functionality, the parties can emulate the full network with blocking setting we considered throughout the paper. Whenever user $i$ wishes to send some message to user $j$, it will sign and encrypt its message using its secret signature key and $j$'s public encryption key. It will send to the server a request to transfer the encrypted and signed message to user $j$. 

Since the message is encrypted, the server cannot learn its contents. Moreover, as the message was also signed, it follows that the server cannot forge a new message without $j$ noticing.\footnote{Note we need to have both the encryption and signature schemes to be secure against adversaries that are polynomial in the number of users $n$, which is possibly superpolynomial in the security parameter $\secParam$.} Thus, the only attack a malicious server can do is to block the message from reaching $j$. Note that it could still be the case that user $i$ is malicious and purposely signed the message incorrectly, hence $j$ will not abort in this case and view it as if the message was blocked.

Observe that since the encryption scheme is assumed to be non-malleable, it follows that the protocol maintains security even if the parties call the \board functionality concurrently. We conclude that the above model is equivalent to the GMPC model considered throughout the paper, in the sense that any attacker for a protocol in one model can be mapped to an attacker in the other model. Thus, all of our results can be translated to this setting as well.

\ifdefined\IsAnon
\else
\subsection*{Acknowledgements}
The authors are very grateful to 
Muthuramakrishnan (Muthu) Venkitasubramaniam for many helpful discussions.

\subsubsection*{Funding}
The work of B.A.~was supported in part by grants from the Israel Science Foundation (no.152/17), and by the Ariel Cyber Innovation Center in conjunction with the Israel National Cyber directorate in the Prime Minister's Office.
The work of M.N.~was supported in part by grants from the Israel Science Foundation (no.2686/20), by the Simons
Foundation Collaboration on the Theory of Algorithmic Fairness and by a Data Science grant of the PCB.
The work of E.O.~was supported in part by grants from the Israel Science Foundation (no.152/17), by the Ariel Cyber Innovation Center in conjunction with the Israel National Cyber directorate in the Prime Minister's Office, and by the Robert L. McDevitt, K.S.G., K.C.H.S. and Catherine H. McDevitt L.C.H.S. endowment at Georgetown University. Part of this work was done
when E.O.~was hosted by Georgetown University.
The work of U.S.~was partially supported by the Israel Science Foundation (grant 1871/19)
and by Len Blavatnik and the Blavatnik Family foundation.
\fi 

\bibliographystyle{abbrvnat}

\appendix

\section{Ideal Worlds For Honest Server}\label{sec:honest_server}
\paragraph{Ideal world for full security with blocking assuming an honest server.}
We next describe the interaction in the ideal world assuming the server is honest. Let $\adv$ be an adversary corrupting a subset $\corrset\su[n]$ of the users, and which does \emph{not} corrupt the server. %
\begin{description}
	
	\item[Inputs:] Each party holds the security parameter $1^\secParam$ and the number of users $n$ (held in binary by the users and in unary by the server). Additionally, the server holds input $x_0$, user $\Pc_i$ holds $x_i$, and the adversary is given auxiliary input $\aux\in\zos$.

	\item[Parties send inputs to trusted party:] The honest parties send their inputs to the  trusted party. For each corrupted user, the adversary $\adv$ sends to the trusted party some value from their domain as input. Denote by $(\partyInputTag{1},\ldots,\partyInputTag{\partNum})$ the tuple of inputs received by the trusted party. %
	
	\item[Trusted party sends output to the server:] The trusted party computes $\outValue\la f(\partyInputTag{1},\ldots,\partyInputTag{\partNum})$
	with uniformly random coins and sends the output $\outValue$ to the server.
	
	\item[Outputs:] An honest server outputs the value sent by the trusted party, and a corrupted server outputs nothing. Additionally, all users output nothing and $\adv$ outputs a function of its view (its inputs and the auxiliary input $\aux$).
\end{description}

\end{document}